\iftrue %

\documentclass[11pt,letterpaper]{article}
\usepackage{amsthm,color,graphicx,url}
\usepackage[margin=1in]{geometry}
\usepackage[font=small,labelfont=bf]{caption}
\usepackage[labelformat=simple]{subcaption}

\usepackage{libertine}\usepackage[libertine]{newtxmath}
\usepackage[scaled=0.96]{zi4}
\usepackage[utf8]{inputenc}

\title{Planar Graph Orientation Frameworks, \\ Applied to KPlumber and Polyomino Tiling}

\author{%
  MIT Hardness Group%
    \footnote{Artificial first author to highlight that the other authors (in alphabetical order) worked as an equal group. Please include all
authors (including this one) in your bibliography, and refer to the authors as “MIT Hardness Group” (without “et al.”).}
\and
  Zachary Abel%
    \thanks{EECS, Massachusetts Institute of Technology, 32 Vassar St., Cambridge, MA 02139, USA, \protect\url{zabel@mit.edu}}
\and
  Erik D. Demaine%
    \thanks{CSAIL, Massachusetts Institute of Technology, 32 Vassar St., Cambridge, MA 02139, USA, \protect\url{{edemaine,jeli,zpeter}@mit.edu}}
\andlinebreak
  Jenny Diomidova%
    \thanks{Université d'Artois, CNRS, UMR 8188 CRIL, Lens, France\\CSAIL, Massachusetts Institute of Technology, 32 Vassar St., Cambridge, MA 02139, USA, \protect\url{diomidova@mit.edu}}
\and
  Jeffery Li%
    \footnotemark[2]
\and
  Zixiang Zhou%
    \footnotemark[2]
}
\date{}

\def\andlinebreak{\end{tabular}\linebreak\begin{tabular}[t]{c}}

\usepackage{hyperref}
\hypersetup{breaklinks,bookmarks,bookmarksnumbered,bookmarksopen,bookmarksopenlevel=2}
{\makeatletter \hypersetup{pdftitle={\@title}}}
{\makeatletter
 \gdef\xxxmark{%
   \expandafter\ifx\csname @mpargs\endcsname\relax %
     \expandafter\ifx\csname @captype\endcsname\relax %
       \marginpar{xxx}%
     \else
       xxx %
     \fi
   \else
     xxx %
   \fi}
 \gdef\xxx{\@ifnextchar[\xxx@lab\xxx@nolab}
 \long\gdef\xxx@lab[#1]#2{\textbf{[\xxxmark #2 ---{\sc #1}]}}
 \long\gdef\xxx@nolab#1{\textbf{[\xxxmark #1]}}
}

{\makeatletter \gdef\fps@figure{!htbp}}

\let\realbfseries=\bfseries
\def\bfseries{\realbfseries\boldmath}

\newtheorem{theorem}{Theorem}[section]
\newtheorem{lemma}[theorem]{Lemma}
\newtheorem{corollary}[theorem]{Corollary}
\newtheorem{observation}[theorem]{Observation}
\newtheorem{remark}[theorem]{Remark}
\theoremstyle{definition}
\newtheorem{definition}{Definition}[section]

\let\epsilon=\varepsilon
\def\defn#1{\textbf{\textit{\boldmath #1}}}

\else %

\documentclass[a4paper,UKenglish,cleveref,autoref,thm-restate,numberwithinsect]{socg-lipics-v2021}

\title{Planar Graph Orientation Frameworks, \\ Applied to KPlumber and Polyomino Tiling}

\titlerunning{Planar Graph Orientation Frameworks}

\author{MIT Hardness Group%
  \footnote{Artificial first author to highlight that the other authors (in alphabetical order) worked as an equal group. Please include all
authors (including this one) in your bibliography, and refer to the authors as “MIT Hardness Group” (without “et al.”).
}}
{CSAIL, Massachusetts Institute of Technology, 32 Vassar St., Cambridge, MA 02139, USA}{}{}{}

\author{Zachary Abel}{EECS, Massachusetts Institute of Technology, 32 Vassar St., Cambridge, MA 02139, USA}{zabel@mit.edu}{https://orcid.org/0000-0002-4295-1117}{}

\author{Erik D. Demaine}{CSAIL, Massachusetts Institute of Technology, 32 Vassar St., Cambridge, MA 02139, USA}{edemaine@mit.edu}{https://orcid.org/0000-0003-3803-5703}{}

\author{Jenny Diomidova}{Université d'Artois, CNRS, UMR 8188 CRIL, Lens, France\\CSAIL, Massachusetts Institute of Technology, 32 Vassar St., Cambridge, MA 02139, USA}{diomidova@mit.edu}{}{}

\author{Jeffery Li}{CSAIL, Massachusetts Institute of Technology, 32 Vassar St., Cambridge, MA 02139, USA}{jeli@mit.edu}{}{}

\author{Zixiang Zhou}{CSAIL, Massachusetts Institute of Technology, 32 Vassar St., Cambridge, MA 02139, USA}{zpeter@mit.edu}{}{}

\authorrunning{MIT Hardness Group} %

\Copyright{Zachary Abel, Erik D. Demaine, Jenny Diomidova, Jeffery Li, and Zixiang Zhou} %



\category{} %

\relatedversion{} %

\acknowledgements{This paper was initiated during open problem solving in the MIT class on Algorithmic Lower Bounds: Fun with Hardness Proofs (6.5440) taught by Erik Demaine in Fall 2023. We thank the other participants of that class --- specifically
Evgeniya Artemova,
Krit Boonsiriseth,
Josh Brunner,
Lily Chung,
Della Hendrickson,
Hayashi Layers,
Jayson Lynch,
Richard Qi,
Mark Saengrungkongka,
Nathan Sheffield,
Andy Tockman,
Anthony Wang,
William Wang, and
Alek Westover
--- for helpful discussions and providing an inspiring atmosphere.
The original concept is based on early discussions in 2012 with
Martin Demaine, Takahasi Horiyama, and Ryuhei Uehara,
who we thank for helping invent this problem domain. 
}

\EventEditors{John Q. Open and Joan R. Access}
\EventNoEds{2}
\EventLongTitle{42nd Conference on Very Important Topics (CVIT 2016)}
\EventShortTitle{CVIT 2016}
\EventAcronym{CVIT}
\EventYear{2016}
\EventDate{December 24--27, 2016}
\EventLocation{Little Whinging, United Kingdom}
\EventLogo{}
\SeriesVolume{42}
\ArticleNo{23}

\fi

\usepackage{enumitem}
\usepackage{graphicx}
\usepackage{amsmath}
\usepackage{xcolor}
\usepackage{tikz}
\usepackage{xfrac}
\def\defn#1{\textbf{\textit{\boldmath #1}}}
\usetikzlibrary{arrows.meta, positioning}
\usetikzlibrary{decorations.pathreplacing}
\usetikzlibrary{intersections}
\usepackage{pifont}

\usepackage{colortbl}
\definecolor{header}{rgb}{0.29,0,0.51}
\definecolor{gray}{rgb}{0.85,0.85,0.85}
\definecolor{hard}{rgb}{1,0.85,0.85}
\definecolor{open}{rgb}{1,1,0.85}
\definecolor{easy}{rgb}{0.85,0.85,1}
\def\header#1{\cellcolor{header}\textcolor{white}{\textbf{#1}}}
\def\zebra{\rowcolors{2}{gray!75}{white}}

\newcommand{\I}{\texorpdfstring{\,\vcenter{\hbox{\includegraphics[scale=0.2]{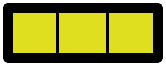}}}\,}I}
\renewcommand{\L}{\texorpdfstring{\,\vcenter{\hbox{\includegraphics[scale=0.2]{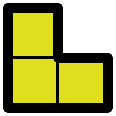}}}\,}L}

\newcommand{\OO}{\texorpdfstring{\,\vcenter{\hbox{\includegraphics[scale=0.2]{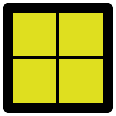}}}\,}O}
\newcommand{\TT}{\texorpdfstring{\,\vcenter{\hbox{\includegraphics[scale=0.2]{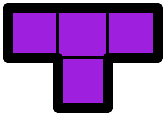}}}\,}T}
\newcommand{\LL}{\texorpdfstring{\,\vcenter{\hbox{\includegraphics[scale=0.2]{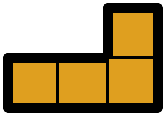}}}\,}L}
\newcommand{\JJ}{\texorpdfstring{\,\vcenter{\hbox{\includegraphics[scale=0.2]{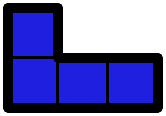}}}\,}J}
\renewcommand{\SS}{\texorpdfstring{\,\vcenter{\hbox{\includegraphics[scale=0.2]{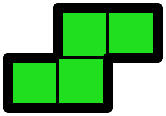}}}\,}S}
\newcommand{\ZZ}{\texorpdfstring{\,\vcenter{\hbox{\includegraphics[scale=0.2]{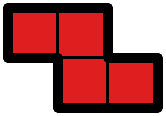}}}\,}Z}
\newcommand{\II}{\texorpdfstring{\,\vcenter{\hbox{\includegraphics[scale=0.2]{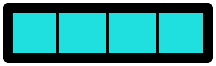}}}\,}I}
\newcommand{\ALL}{\II, \allowbreak \OO, \allowbreak \TT, \allowbreak \SS, \allowbreak \ZZ, \allowbreak \JJ, \allowbreak \LL}

\begin{document}

\maketitle

\begin{abstract}
Given a graph, when can we orient the edges to satisfy
local constraints at the vertices,
where each vertex specifies which local orientations
of its incident edges are allowed?
This family of graph orientation problems is a special kind of SAT problem, where each variable (edge orientation)
appears in exactly two clauses (vertex constraints) ---
once positively and once negatively.
We analyze the complexity of many natural vertex types
(patterns of allowed vertex neighborhoods),
most notably all sets of \emph{symmetric} vertex types
which depend on only the \emph{number} of incoming edges.
In many scenarios, including Planar and Non-Planar Symmetric Graph Orientation
with constants,
we give a full dichotomy characterizing
P vs.\ NP-complete problem classes.
We apply our results to obtain new polynomial-time algorithms,
resolving a 20-year-old open problem about KPlumber;
to simplify existing NP-hardness proofs for tiling with trominoes;
and to prove new NP-completeness results for tiling with tetrominoes.

\end{abstract}

\section{Introduction}

The main technique for proving NP-hardness is \emph{reduction} from a known NP-hard problem.
In principle, we could reduce from any NP-hard problem,
such as Cook's original Boolean satisfiability (SAT) \cite{Cook-1971}.
But in practice, some problems are easier to reduce from than others,
depending on the target problem being reduced to.
Still, the most common starting points for NP-hardness are
SAT or one of its many variants;
refer to the surveys \cite{HardnessBook,filhothesis}.
Notably, Schaefer's dichotomy theorem \cite{schaefer1978}
classifies the complexity of the family of SAT problems
parameterized by a set $\Gamma$ of clause types,
where a $j$-ary \defn{clause type} lists the allowed assignments
for its $j$ variables:
if the clause types all satisfy a common simplifying situation
(e.g., being satisfied when all variables are true),
then the SAT problem can be solved in polynomial time,
and otherwise it is NP-complete.
In practice, this theorem has provided several useful SAT variants
for reductions, such as Positive 1-in-3SAT and Positive Not-All-Equal (NAE) SAT
\cite{schaefer1978}.
A more complicated and recent dichotomy theorem \cite{delta-matroid-planar-schaefer}
characterizes the complexity of the \defn{planar} versions of these problems,
where the bipartite graph of connections between variables and clauses is assumed to be planar.
When the target problem is planar
(including many of the 2D games and puzzles analyzed at FUN),
Planar 3SAT and Planar Positive 1-in-3SAT \cite{planar-counting}
are common starting points for reductions.

In this paper, we aim to prove similar non-planar and planar dichotomies
for a family of ``graph orientation'' problems.
Our starting point is a problem introduced by
Horiyama, Ito, Nakatsuka, Suzuki, and Uehara in 2012 \cite{tromino-tiling},
as an intermediate problem
for an NP-hardness reduction to tiling a polyomino with trominoes.
Specifically, they proved NP-completeness of the following problem:
given an undirected 3-regular graph where each vertex $v$ is labeled
with a set $S_v$ that is either $\{0,3\}$, $\{1\}$, or $\{2\}$,
can we orient (direct) the edges so that each vertex $v$ has in-degree in $S_v$?%
\footnote{They called this problem ``1-in-3 Graph Orientation (GO)'', but a
more accurate name would be ``$\{0,3\}$-in-$3$, $1$-in-$3$, and $2$-in-$3$
Graph Orientation''.}

In this paper, we address two natural questions about a broader family of
graph orientation problems:
\begin{enumerate}
\item What if each vertex set $S_v$ is restricted to other possible sets?
  For which sets is the problem NP-complete?
  For example, we show that the sets $\{0,3\},\{1\}$ or the sets $\{0,3\},\{2\}$
  (but no single set)
  suffice for NP-completeness in 3-regular graphs,
  strengthening \cite{tromino-tiling}.
  As a consequence, the 2-in-3 or 1-in-3 gadgets in \cite{tromino-tiling}
  are no longer both necessary for proving NP-hardness of tromino tiling.
  We also show that the sets $\{0,4\},\{1\}$, the sets $\{0,4\},\{2\}$,
  or any \emph{one} of the sets
  $\{0,3\}, \{1,4\}, \{0,3,4\}, \{0,1,4\}$
  suffice for NP-completeness in 4-regular graphs.
  More generally, we also consider other vertex constraints that distinguish
  between (label) the incident edges.
\item What if the graph is assumed to be planar?
  In what cases does the problem remain NP-complete?
  Like Planar 3SAT, this constraint is useful for many 2D problems,
  such as the 2D tiling problem that originally motivated
  Horiyama et al.~\cite{tromino-tiling}.
  In particular, we show that all the aforementioned 3-regular and 4-regular problems
  remain NP-complete when the graph is planar.
  As a consequence, the crossover gadgets in \cite{tromino-tiling}
  are no longer necessary for proving NP-hardness of tromino tiling.

\end{enumerate}

\subsection{Problem}

The \defn{Graph Orientation (GO) problem} is parameterized by a set $\Gamma$ of vertex types,
where a degree-$j$ \defn{vertex type} specifies allowable patterns of directions
for $j$ labeled incident edges.
A vertex type can thus be viewed as a subset of $\{0,1\}^j$,
where $0$ represents an outgoing edge and $1$ represents an incoming edge.
Given an undirected graph where each degree-$j$ vertex
is labeled with a degree-$j$ vertex type in $\Gamma$
and a local numbering of the incident edges,
the problem is to decide whether the graph (edges)
can be oriented to satisfy every vertex.
For example, constraint graph satisfiability from
Nondeterministic Constraint Logic \cite{GPC,SharpSAT_ISAAC2024}
is an NP-complete problem that can be stated in this framework,
with $\Gamma$ consisting of two vertex types:
$\{(a,b,c) \mid a+b+c \geq 1\}$ (\textsc{or} vertex) and
$\{(a,b,c) \mid a=b=1 \text{ or } c=1\}$
(\textsc{and} vertex).

In \defn{Symmetric Graph Orientation}, the vertex types are all \defn{symmetric},
meaning that they treat all incident edges identically (as if they were unlabeled),
and thus depend only on the \emph{number} of incoming edges.
A degree-$j$ symmetric vertex type can thus be viewed as a subset $S$ of $\{0,1,\dots,j\}$,
specifying the set of allowed numbers of incoming edges;
we call this vertex type \defn{$S$-in-$j$}.
In Symmetric Graph Orientation, we are given an undirected graph
and a set $S_v$ for each vertex $v$ of degree~$j$,
where $S_v\text{-in-}j \in \Gamma$,
and the problem is to orient the edges so that every vertex $v$ has in-degree in $S_v$.
For example, the graph orientation problem of \cite{tromino-tiling}
uses $\Gamma = \{\{0,3\}\text{-in-}3,\{1\}\text{-in-}3,\{2\}\text{-in-}3\}$.

In \defn{Planar Graph Orientation}, we assume that the given graph is planar.
As mentioned above, this constraint is motivated by reductions to planar/2D problems,
as it generally means that the reduction can avoid implementing a \defn{crossover} gadget (an area in a planar embedding of a graph where two edges intersect).

Graph orientation problems are related to a special case of SAT called \defn{SAT-E$2$}
\cite{HardnessBook,filhothesis},
where each variable appears in exactly 2 clauses.
The connection is as follows: edges correspond to variables, and vertex constraints can be encoded as clauses.
Edge orientations represent variable settings, and edges connecting exactly two vertices correspond to each variable appearing in exactly two clauses.
Graph Orientation has the additional constraint that each variable appears in one clause negatively
and one clause positively, because a directed edge has the opposite orientation at its two endpoints.

\subsection{Results: Framework}

Our goal is to characterize for which collections $\Gamma$ of vertex types (Planar) $\Gamma$-GO is in P,
and for which $\Gamma$ it is NP-complete.
In this paper, we focus mainly on Symmetric Graph Orientation,
though
we also analyze some important asymmetric vertex types.
Together, our results provide a framework for NP-hardness reductions
based on Graph Orientation, providing a catalog of vertex types that suffice for hardness
(as well as many that do not),
both when restricted to planar graphs and not.

\begin{table}[]
\centering
\setlength{\extrarowheight}{5pt}
\zebra
\begin{tabular}{|c|c|c|}
\cline{2-3}
\multicolumn{1}{c}{} & \multicolumn{2}{|c|}{\header{Dichotomy Theorem}} \\ \hline
\header{Vertex types} & \header{General} & \header{Planar} \\ \hline
\begin{minipage}[t]{0.2\textwidth}\centering
$\{S$-in-$j\}$,\\$0$-in-$1$,\\$1$-in-$1$
\end{minipage} &
\begin{minipage}[t]{0.25\textwidth}
Theorem~\ref{thm:S-in-j-and-const}: in P iff
\begin{itemize}[nosep]
\item all bijunctive
\item all affine
\item $\max\Delta\Gamma \le 2$
\end{itemize}
else, NP-complete
\end{minipage} &
\begin{minipage}[t]{0.35\textwidth}
Theorem~\ref{thm:planar-S-in-j-and-const}: in P iff
\begin{itemize}[nosep]
\item all bijunctive
\item all affine
\item $\max\Delta\Gamma \le 2$
\item $\gcd\Delta\Gamma \ge 5$ and all\\$S \subseteq\{0,j\}$ or $S = \{1\}$
\item $\gcd\Delta\Gamma \ge 5$ and all\\$S \subseteq\{0,j\}$ or $S = \{j-1\}$
\end{itemize}
else, NP-complete\vspace{5pt}
\end{minipage} \\ \hline
\begin{minipage}[t]{0.2\textwidth}\centering
$\{S$-in-$j\}$,\\$0$-in-$k$ or $k$-in-$k$
\end{minipage} &
\multicolumn{2}{|c|}{\begin{minipage}[t]{0.3\textwidth}
Theorem~\ref{thm:S-in-j-term}: in P iff
\begin{itemize}[nosep]
\item any of the above
\item all bottom-up closed
\item all top-down closed
\end{itemize}
else, NP-complete\vspace{5pt}
\end{minipage}} \\ \hline
\begin{minipage}[t]{0.2\textwidth}\centering
$\{i$-in-$j\}$,\\duplicator ($f \ne 0$)
\end{minipage} &
\begin{minipage}[t]{0.25\textwidth}
Theorem~\ref{thm:i-in-j-and-dup}: in P iff
\begin{itemize}[nosep]
\item all bijunctive
\end{itemize}
else, NP-complete
\end{minipage} &
\begin{minipage}[t]{0.35\textwidth}
Theorem~\ref{thm:i-in-j-and-dup}: in P iff
\begin{itemize}[nosep]
\item all bijunctive
\item $f \ge 5$; all $i \in \{0,1,j\}$
\item $f \ge 5$; all $i \in \{0,j-1,j\}$
\end{itemize}
else, NP-complete\vspace{5pt}
\end{minipage} \\ \hline
\begin{minipage}[t]{0.2\textwidth}\centering
$\{i$-in-$j\}$,\\synchronizer
\end{minipage} &
\multicolumn{2}{|c|}{\begin{minipage}[t]{0.3\textwidth}
Theorem~\ref{thm:i-in-j-and-sync}: in P iff
\begin{itemize}[nosep]
\item all bijunctive
\item all $i < j/2$ or $1$-in-$2$
\item all $i > j/2$ or $1$-in-$2$
\item all $i \in \{0,1,j\}$
\item all $i \in \{0,j-1,j\}$
\end{itemize}
else, NP-complete\vspace{5pt}
\end{minipage}} \\ \hline
\begin{minipage}[t]{0.2\textwidth}\centering
$\{i$-in-$j\}$,\\alternator\vspace{5pt}
\end{minipage} &
Not applicable & Theorem~\ref{thm:i-in-j-and-alt}: always in P \\ \hline
\end{tabular}
\caption{Summary of our dichotomy results characterizing the complexity of (Planar) Graph Orientation. For the case described by each row, we give complete dichotomies for both general graphs and planar graphs. ``$\{S$-in-$j\}$'' denotes any set of symmetric vertex types. ``$\{i$-in-$j\}$'' denotes any set of symmetric vertex types each with exactly one valid number of incoming edges.
Duplicator, synchronizer, and alternator vertex types are defined in Section~\ref{subsec:duplicators}.}
\label{tab:all-results}
\end{table}

Our dichotomy results, detailed in Table~\ref{tab:all-results},
provide a near-complete classification of Symmetric Graph Orientation problems,
both in the planar and not-necessarily-planar settings (which we usually just call non-planar).
If we assume the existence of ``constant vertices'' that have degree $1$ and force an edge to
have a particular orientation, then the first row of Table~\ref{tab:all-results}
gives a complete characterization of P vs.\ NP-complete.
The second row generalizes to weaker forms of constants, namely,
either a $0$-in-$k$ or a $k$-in-$k$ vertex type for some $k \geq 1$.
The later rows specialize to symmetric vertex types with only one valid
number of incoming edges, but replace the required constants with
certain families of (sometimes asymmetric) vertex types called
``duplicators'', ``synchronizers'', and ``alternators'';
refer to Section~\ref{subsec:duplicators} for definitions.

In the planar setting, our hardness results are driven largely by analysis of which vertex types can build a crossover gadget. With the ability to create crossover gadgets, we can reduce the corresponding general versions of Graph Orientation problems to their planar variants by inserting crossover gadgets where any two edges intersect in a planar embedding.
This allows us to transfer hardness results for the relevant SAT variants \cite{S-in-k-SAT-2, delta-matroid-planar-schaefer} to Graph Orientation, even if the results are not for planar variants.

Along the way, we identify five surprising classes of (Planar) Graph Orientation problems that are solvable in polynomial time. Here are somewhat simplified statements of these results (the actual lemmas are slightly more general):

\begin{enumerate}
    \item In Lemma~\ref{lem:k>=5-easy}, we show that Planar Graph Orientation with $1$-in-$j$ and $\{0,k\}$-in-$k$ vertices is in P for all $k \ge 5$. The same also holds for $(j-1)$-in-$j$ and $\{0,k\}$-in-$k$.
    \item In Lemma~\ref{lem:top-down-easy}, we show that Graph Orientation with $S$-in-$j$ vertices is in P if $\forall i > j/2, i \in S \implies i - 1 \in S$. The same also holds if $\forall i < j/2, i \in S \implies i + 1 \in S$.
    \item In Theorem~\ref{thm:i-in-j-and-alt}, we show that Planar Graph Orientation with $i$-in-$j$ vertices and alternators is in P (an alternator is an even-degree vertex that forces its edges to alternate between pointing in and out).
\end{enumerate}

\subsection{Results: Applications}

To illustrate the usefulness of our Graph Orientation framework,
we show several applications of reductions from Graph Orientation:

\begin{enumerate}
    \item In Section~\ref{sec:kplumber}, we analyze the game \defn{KPlumber}, a video game involving rotating tiles in a grid to make corresponding edges match up. KPlumber was studied by Král et al.~\cite{kplumber}, who gave complexity classifications for when the input is restricted to some subsets of the tile types. Their paper left one family of cases unresolved, which they showed are equivalent to the case with all tile types but one (the ``curve'' tile). We resolve this open case in the positive: surprisingly, there is a polynomial-time algorithm for these instances, hence completing the complexity classification of KPlumber.
    \item In Section~\ref{sec:tromino-tiling}, we give simpler proofs for the known \cite{tromino-tiling,older-tromino-tiling} NP-hardness of \defn{tromino tiling}, allowing either a single tromino type ($\L$ or $\I$) or allowing both tromino types ($\L$ and $\I$).
    (See \cite{tromino-undecidability} for different simple proofs.)
    \item In Section~\ref{sec:tetromino-tiling}, we study the complexity of \defn{tetromino tiling}, extending past work on tromino tiling~\cite{tromino-tiling,older-tromino-tiling}. We determine the complexity of tiling a rectilinear polygon using translated and/or rotated copies of tetrominoes from a given subset $S$ of tetrominoes. We give full characterizations for the case where the subset $S$ consists of a single tetromino type ($|S|=1$), with or without reflection, and in the case where the subset $S$ is the set of all tetrominoes.
    See Table~\ref{tab:tetrominos}.
\end{enumerate}

\begin{table}[h!]
\centering
\setlength{\extrarowheight}{2pt}
\zebra
\begin{tabular}{|c|c|c|}
\hline
\header{Tetromino(es)} & \header{Complexity} & \header{Reference} \\ \hline
$\OO$ & P & Theorem~\ref{thm:O-tetromino} \\ \hline
$\SS$ & P & Theorem~\ref{thm:S-tetromino} \\ \hline
$\ZZ$ & P & Theorem~\ref{thm:S-tetromino} \\ \hline
$\II$ & NP-complete & \cite[Section 2.3.3]{tromino-tiling} \\ \hline
$\TT$ & NP-complete & Theorem~\ref{thm:T-tetromino} \\ \hline
$\LL$ & NP-complete & Theorem~\ref{thm:L-tetromino} \\ \hline
$\JJ$ & NP-complete & Theorem~\ref{thm:L-tetromino} \\ \hline
$\LL,\JJ$ & NP-complete & Theorem~\ref{thm:L-tetromino} \\ \hline
$\SS,\ZZ$ & NP-complete & Theorem~\ref{thm:SZ-tetromino} \\ \hline
all: $\ALL$ & NP-complete & Theorem~\ref{thm:T-tetromino} \\ \hline
\end{tabular}
\caption{Complexity of tiling with single tetromino sets (with or without reflection) and all tetrominoes.}
\label{tab:tetrominos}
\end{table}

\subsection{Outline}

The structure of the rest of the paper is as follows. In Section~\ref{sec:prelims}, we describe Schaefer's dichotomy theorems for SAT variants relevant to the results in our paper (including Planar SAT and SAT-E$2$). In Section~\ref{sec:graph-ori}, we describe our Graph Orientation framework in more detail and characterize the complexity of many graph orientation problems. In Section~\ref{sec:kplumber}, we apply our results in graph orientation to resolve the last open problem in the complexity of KPlumber. In Section~\ref{sec:tiling}, we apply our results in graph orientation to analyze the complexity of tromino and tetromino tiling problems.

\section{Schaefer-Type Dichotomies}
\label{sec:prelims}

In this section, we give some background on Schaefer's dichotomy theorems for SAT and its extensions to Planar SAT and SAT-E$2$ (where each variable appears in exactly two clauses). These results will form the backbone of our framework in Section~\ref{sec:graph-ori}.

A SAT problem is
parameterized by a set $\Gamma$ of \defn{relations},
where a $j$-ary relation specifies the allowed assignments
for its $j$ variables. Unless specified, each variable must appear positively in every relation (we do not allow negated literals). We call this problem $\Gamma$-SAT.

\begin{definition}
A relation is \defn{$0$-valid} if it can be satisfied by setting all variables to false. A relation is \defn{$1$-valid} if it can be satisfied by setting all variables to true.

A relation is \defn{Horn} if for any two satisfying assignments, their bitwise AND is also a satisfying assignment. A relation is \defn{dual-Horn} if for any two satisfying assignments, their bitwise OR is also a satisfying assignment.

A relation is \defn{affine} if for any three satisfying assignments, their bitwise XOR (exclusive or) is also a satisfying assignment. A relation is \defn{bijunctive} if for any three satisfying assignments, their bitwise MAJ (majority) is also a satisfying assignment.

A relation is \defn{self-complementary} if for any satisfying assignment, its bitwise NOT is also a satisfying assignment.
\end{definition}

\begin{theorem}[Schaefer's dichotomy \cite{schaefer1978}]\label{thm:schaefer}
$\Gamma$-SAT is NP-complete, except for the following 6 cases, which are in P:
\begin{enumerate}
    \item every relation in $\Gamma$ is $0$-valid;
    \item every relation in $\Gamma$ is $1$-valid;
    \item every relation in $\Gamma$ is Horn;
    \item every relation in $\Gamma$ is dual-Horn;
    \item every relation in $\Gamma$ is bijunctive;
    \item every relation in $\Gamma$ is affine.
\end{enumerate}
\end{theorem}

For Planar $\Gamma$-SAT, where the bipartite graph of connections between variables and clauses must be planar, there is one additional easy case.

\begin{theorem}[Planar Schaefer's dichotomy \cite{delta-matroid-planar-schaefer}]\label{thm:planar-schaefer}
Planar $\Gamma$-SAT is NP-complete, except for the following 7 cases, which are in P:
\begin{itemize}
    \item cases 1-6 from Theorem~\ref{thm:schaefer};
    \item every relation $R \in \Gamma$ is self-complementary and $dR$ is an even $\Delta$-matroid (refer to \cite{delta-matroid-planar-schaefer} for details).
\end{itemize}
\end{theorem}

Often, the set $\Gamma$ includes constants ($1$-ary relations that force a variable to be true or force it to be false). In this case, we can eliminate some of the easy cases because the true constant is not $0$-valid, the false constant is not $1$-valid, and neither constant is self-complementary.

\begin{corollary}[Schaefer's dichotomy with constants]\label{thm:schaefer-const}
If $\Gamma$ includes both constants, both $\Gamma$-SAT and Planar $\Gamma$-SAT are NP-complete, except for the following 4 cases, which are in P:
\begin{enumerate}
    \item every relation in $\Gamma$ is Horn;
    \item every relation in $\Gamma$ is dual-Horn;
    \item every relation in $\Gamma$ is bijunctive;
    \item every relation in $\Gamma$ is affine.
\end{enumerate}
\end{corollary}

We will focus on symmetric clause types, which depend only on the number of associated variables assigned true. For a set $S \subseteq \{0,1,\dots,j\}$, an $S$-in-$j$ clause is a $j$-ary relation that is satisfied if and only if the number of true variables among its arguments lies in $S$. %

\begin{definition}\label{def:bitstring}
The \defn{bitstring} of an $S$-in-$j$ clause is a length $j+1$ string with a \texttt{1} in ($0$-indexed) position $i$ if $i \in S$ and \texttt{0} if $i \not\in S$ (for example, the bitstring of $\{0,1,4\}$-in-$5$ is \texttt{110010}).

We say that the clause has a \defn{gap of size $k$} if the bitstring contains exactly $k$ consecutive zeros surrounded by ones. 
\end{definition}

Finally, consider $\Gamma$-SAT-E$2$, where each variable appears in exactly two clauses. This is more restrictive than $\Gamma$-SAT and thus admits more polynomial-time cases. If $\Gamma$ includes both constants, the following theorem \cite[Corollary~5.1]{S-in-k-SAT-2} gives a complete symmetric dichotomy:

\begin{theorem}[Symmetric $\Gamma$-SAT-E$2$ dichotomy \cite{S-in-k-SAT-2}]\label{thm:S-in-k-SAT-2}
    Let $\Gamma$ be a set of symmetric relations (i.e., each of the form $S$-in-$j$) that includes both constants. Then $\Gamma$-SAT-E$2$ is in P if no relation in $\Gamma$ has a gap of size at least $2$, and is equivalent to $\Gamma$-SAT otherwise.
\end{theorem}

\section{Graph Orientation Framework}
\label{sec:graph-ori}

In a graph orientation problem, we are given a graph $G = (V,E)$ where each vertex in $V$ has a specified \defn{vertex type}, and we must determine whether there exists an orientation of $G$ (an assignment of directions to every edge in $E$) that satisfies every vertex in $V$. Each vertex places some constraints on its neighboring edges as specified by its vertex type. More formally, a vertex type of degree $d$ is a $d$-ary binary relation, i.e., a subset of $\{0,1\}^d$. If $\Gamma$ is a finite set of vertex types, we denote by \defn{$\Gamma$-Graph Orientation}, or \defn{$\Gamma$-GO}, the graph orientation problem on graphs where the types of all vertices belong to $\Gamma$. We allow self-loops and multiedges. We denote by \defn{Planar $\Gamma$-Graph Orientation}, or \defn{Planar $\Gamma$-GO}, the graph orientation problem where $G$ is guaranteed to be planar.

For an integer $j \ge 1$ and $S \subseteq \{0,1,\dots,j\}$, an $S$-in-$j$ vertex type has degree $j$ and is satisfied if and only if the number of neighboring edges directed inwards belongs to $S$. We abbreviate $\{i\}$-in-$j$ as just $i$-in-$j$.

Note that $S$-in-$j$ vertex types (in Graph Orientation) are very similar to $S$-in-$j$ clause types (in SAT). We define the \defn{bitstring} and a \defn{gap} of an $S$-in-$j$ vertex type analogously to Definition~\ref{def:bitstring}.

There is a natural reduction from $\Gamma$-GO to $\Gamma'$-SAT-E$2$, where $\Gamma' := \Gamma \cup \{\text{$1$-in-$2$}\}$ (each vertex becomes a clause from $\Gamma$, and each edge becomes a $1$-in-$2$ clause). We can use this reduction to transfer some of the easiness results from SAT and SAT-E$2$ over to Graph Orientation:

\begin{lemma}\label{lem:classic-easy-cases}
The following cases of $\Gamma$-GO are in P:
\begin{enumerate}
    \item all vertex types in $\Gamma$ are bijunctive;
    \item all vertex types in $\Gamma$ are affine;
    \item all vertex types in $\Gamma$ are symmetric and do not have gaps of size $2$ or more.
\end{enumerate}
\end{lemma}

\begin{proof}
    The first two cases follow from Theorem~\ref{thm:schaefer}. The last one follows from Theorem~\ref{thm:S-in-k-SAT-2}. Crucially, this uses the fact that $1$-in-$2$ is bijunctive, affine, symmetric, and does not have gaps of size $2$ or more.
\end{proof}

Note that the easiness for Horn or dual-Horn \emph{does not} carry over this way, because $1$-in-$2$ is not Horn or dual-Horn.

It can be somewhat tedious to use the definitions to check whether a relation is bijunctive or affine. Because we are mostly interested in symmetric relations, we provide the following helper lemma, which can make it easier to apply our main results. It applies both to SAT clauses and Graph Orientation vertices.

\begin{lemma}
    \begin{enumerate}
        \item An $S$-in-$j$ relation is bijunctive if and only if $S = \{0,\dots,j\}$, $S \subseteq \{0,j\}$, $S = \{0,1\}$, $S = \{j-1,j\}$, or $j \le 2$.
        \item An $S$-in-$j$ relation is affine if and only if $S = \{0,\dots,j\}$, $S \subseteq \{0,j\}$, $S = \{0,2,4,\dots\}$, or $S = \{1,3,5,\dots\}$.
        \item An $i$-in-$j$ relation is bijunctive if and only if it is $0$-in-$j$, $j$-in-$j$, or $1$-in-$2$.
        \item An $i$-in-$j$ relation is affine if and only if it is $0$-in-$j$, $j$-in-$j$, or $1$-in-$2$.
    \end{enumerate}
\end{lemma}
\begin{proof}[Proof sketch]
    \begin{enumerate}
        \item Suppose $S$-in-$j$ is bijunctive, meaning that for every three satisfying assignments their bitwise MAJ is also satisfying. Suppose $i \in S$ and $0 < i < j - 1$. Choose three bitstrings \texttt{001}, \texttt{010}, and \texttt{100}, and add the same padding to them so they have $i$ ones and $j-i$ zeros. Their bitwise MAJ has one fewer \texttt{1}, so $i-1 \in S$. Similarly, if $i \in S$ and $1 < i < j$, then $i+1 \in S$. This leaves a finite number of possibilities for $S$ that can be checked easily.
        \item Suppose $S$-in-$j$ is affine, meaning that for every three satisfying assignments their bitwise XOR is also satisfying. Suppose $i \in S$ and $0 < i < j - 1$. Choose three bitstrings \texttt{001}, \texttt{010}, and \texttt{100} with appropriate padding. Their bitwise XOR has two more ones, so $i+2 \in S$. Similarly, if $i \in S$ and $1 < i < j$, then $i-2 \in S$. This leaves a finite number of possibilities for $S$ that can be checked easily.
        \item This follows directly from 1.
        \item This follows directly from 2.\qedhere
    \end{enumerate}
\end{proof}

\begin{definition}
    A \defn{constant} vertex type is a $0$-in-$1$ or a $1$-in-$1$. A \defn{terminator} vertex type is a $0$-in-$j$ or a $j$-in-$j$ for some $j \ge 1$ (which can be thought of as a bundle of constants).
\end{definition}

\begin{definition}
    The \defn{net flow} of a vertex in a given satisfying edge orientation is its in-degree minus out-degree.
    
    The \defn{net flow} of a set of vertices in a given satisfying edge orientation is the sum of the net flows of the individual vertices.
\end{definition}

During reductions, we often build larger gadgets by connecting vertices together so that the resulting structure behaves like a single vertex of another type. For example, connecting two $1$-in-$3$ vertices by an edge creates a gadget with $4$ external edges, whose locally satisfying assignments are exactly those of a $1$-in-$4$ vertex.
Thus, given a GO instance with $1$-in-$4$ vertices,
we can replace each $1$-in-$4$ vertex with this gadget to get an equivalent GO instance with $1$-in-$3$ vertices instead.
We say that $1$-in-$3$ vertices \defn{simulate} $1$-in-$4$ vertices.

More generally, a \defn{gadget construction} is a GO graph $G$ where some
degree-1 vertices are marked \defn{external} meaning they have no constraint.
Considering all satisfying orientations of $G$,
restricting these to look at orientations of just the external edges
(those incident to external vertices),
and imagining these external edges are being incident to a single ``vertex'',
we obtain a vertex type $T$;
we say that $G$ is a \defn{simulation} of~$T$.
If $G$ uses vertex types from a set $\Gamma$, 
then we say that $\Gamma$ \defn{simulates}~$T$.
Simulations form a local notion of reduction
(inspired by similar notions in motion-planning gadgets
\cite{GadgetsChecked_FUN2022}).
In particular, if $(\Gamma \cup \{T\})$-GO is NP-hard and
$\Gamma'$ simulates $T$, then $(\Gamma \cup \Gamma')$-GO is NP-hard.
We will mostly consider \defn{planar} gadget constructions,
where the graph can be drawn in the plane with all external vertices
on the outer face, so replacement preserves planarity.
Thus, if Planar $(\Gamma \cup \{T\})$-GO is NP-hard and
$\Gamma'$ planarly simulates $T$,
then Planar $(\Gamma \cup \Gamma')$-GO is NP-hard.

The rest of this section is organized as follows. In Section~\ref{subsec:duplicators}, we consider a class of asymmetric vertex types called ``duplicators'' that are useful for copying variables in reductions, and fully characterize what they can simulate. In Section~\ref{subsec:S-in-j-with-const}, we prove dichotomies for $S$-in-$j$ vertex types in the case when both constant vertex types are allowed. We generalize these results in Section~\ref{subsec:removing-const} to terminators. Finally, in Section~\ref{subsec:i-in-j}, we consider $i$-in-$j$ vertex types without constants and prove more dichotomy theorems.

\subsection{Duplicators}\label{subsec:duplicators}

In reductions, we often need a way to copy a variable, possibly negating some of the copies. In this subsection, we study such vertex types and fully classify them.

A \defn{duplicator} is a vertex with exactly two satisfying assignments, which are inverses of each other. A duplicator is \defn{trivial} if it has degree at most $2$. There are only three such trivial duplicators (which happen to all be symmetric): $\{0,1\}$-in-$1$, $1$-in-$2$, and $\{0,2\}$-in-$2$. All other duplicators are \defn{non-trivial}. We say that a set of duplicators is \defn{non-trivial} if at least one of its members is non-trivial.

The two satisfying assignments of a duplicator have net flows $+f$ and $-f$ for some $f \ge 0$. We say that the \defn{net flow} of the duplicator is $\pm f$. 

The following three types of duplicators are particularly useful:

\begin{itemize}
    \item A \defn{$k$-equalizer} is a degree-$k$ duplicator whose two satisfying assignments have all edges pointing in or all edges pointing out.
    \item A \defn{synchronizer} is a degree-$4$ duplicator with two satisfying assignments as shown in Figure~\ref{fig:synchronizer}, left. One can think of it as a pair of edges that must point in the same direction.
    \item An \defn{alternator} is an even-degree duplicator whose two satisfying assignments alternate between inwards and outwards around the vertex (see Figure~\ref{fig:synchronizer}, right).
\end{itemize}

\begin{figure}
\centering
\begin{tikzpicture}[
node/.style={circle, draw, thick, circle, draw, minimum size=5mm, inner sep=0pt},
edge/.style={-latex, thick}
]
\begin{scope}[shift={(0,0)}, rotate=-45]
    \node[node] (A) at (0,0) {$S$};
    \draw[edge, red] (-0.75,0) -- (A);
    \draw[edge, blue] (A) -- (0.75,0);
    \draw[edge, blue] (0,-0.75) -- (A);
    \draw[edge, red] (A) -- (0,0.75);
\end{scope}
\begin{scope}[shift={(3,0)}, rotate=-45]
    \node[node] (A) at (0,0) {$S$};
    \draw[edge, red] (A) -- (-0.75,0);
    \draw[edge, blue] (0.75,0) -- (A);
    \draw[edge, blue] (A) -- (0,-0.75);
    \draw[edge, red] (0,0.75) -- (A);
\end{scope}
\draw (4.5,-1) -- (4.5,1);
\begin{scope}[shift={(6,0)}]
    \node[node] (A) at (0,0) {$A$};
    \draw[edge, red] (-0.75,0) -- (A);
    \draw[edge, blue] (A) -- (-0.375,0.6459);
    \draw[edge, red] (0.375,0.6459) -- (A);
    \draw[edge, blue] (A) -- (-0.375,-0.6459);
    \draw[edge, red] (0.375,-0.6459) -- (A);
    \draw[edge, blue] (A) -- (0.75,0);
\end{scope}
\begin{scope}[shift={(9,0)}]
    \node[node] (A) at (0,0) {$A$};
    \draw[edge, red] (A) -- (-0.75,0);
    \draw[edge, blue] (-0.375,0.6459) -- (A);
    \draw[edge, red] (A) -- (0.375,0.6459) ;
    \draw[edge, blue] (-0.375,-0.6459) -- (A);
    \draw[edge, red] (A) -- (0.375,-0.6459);
    \draw[edge, blue] (0.75,0) -- (A);
\end{scope}
\end{tikzpicture}
\caption{Left: The two satisfying assignments of a synchronizer vertex. Right: The two satisfying assignments of a $6$-alternator vertex. }
\label{fig:synchronizer}
\end{figure}
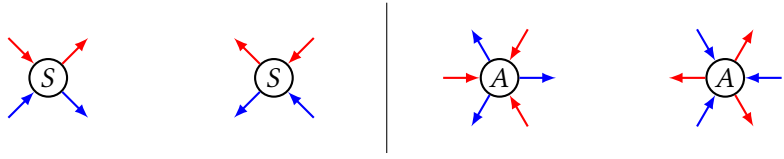

We now state the main theorem of this subsection, which provides a complete characterization of duplicators by showing which sets of duplicators can simulate which other duplicators.

\begin{theorem}\label{thm:dup-characterization}
Let $\Gamma$ be a finite non-empty set of duplicators, with net flows $\pm f_1, \dots, \pm f_k$.
\begin{enumerate}
    \item If $\Gamma$ only has trivial duplicators, they can only simulate themselves and $1$-in-$2$.
    \item Otherwise, if $\Gamma$ has only alternators, then $\Gamma$ can simulate all other alternators, and only them.
    \item Otherwise, $\Gamma$ can simulate a duplicator with net flow $\pm f$ if and only if $f$ is a multiple of $\gcd(f_1, \dots, f_k)$ (including synchronizers and alternators, which have net flow $0$).
\end{enumerate}

If we do not care about planarity, we just have cases 1 and 3.
\end{theorem}

Note that any network of duplicators has 0 or $2^k$ solutions, where $k$ is the number of connected components. So in order to simulate a duplicator with other duplicators, the graph must be connected.

First, we show that all non-trivial alternators are equivalent and are in fact the weakest type of non-trivial duplicator:

\begin{lemma}\label{lem:sim-alternators}
    A non-trivial alternator can simulate any other alternator, and nothing else.
\end{lemma}
\begin{proof}
    Connecting two copies of the degree-$2d$ alternator results in a larger alternator of degree $4d-2 > 2d$. We can repeat this process to get to arbitrarily high degree. We can then add self-loops to get an alternator of the desired degree.

    Note that any simulation of a connected structure can be thought of as a sequence of connecting two gadgets and adding a self-loop between two adjacent edges. Both of these operations preserve the property of being an alternator.
\end{proof}

Next, we show that the synchronizer is the second-weakest type of non-trivial duplicator:

\begin{lemma}\label{lem:sim-synch}
    Every non-trivial duplicator other than alternators simulates a synchronizer.
\end{lemma}
\begin{proof}
    Suppose that our duplicator has net flow $\pm f$. We perform a case analysis on $f$:
    \begin{itemize}
        \item If $f \ge 3$, choose any adjacent opposite-direction edges and annihilate them by connecting them with a self-loop. Repeat this until all edges point the same way and we get the $f$-equalizer. Connect two of these equalizers together with $f-2$ edges (see Figure~\ref{fig:sim-synch}, left) to get a synchronizer.
        \item If $f = 2$ or $f = 1$, we apply the same process, but stop when the degree becomes $3$ or $4$. Connect two of these duplicators together (see Figure~\ref{fig:sim-synch}, right) to get a synchronizer.
        \item If $f = 0$, we apply the same process but we only annihilate edges if doing so does not result in an alternator. We continue this until we can no longer annihilate any edges. Let $d = 2k$ be the current degree. We have $d \ge 4$ because $d = 2$ corresponds to a trivial alternator. Without loss of generality, the first two edges point in opposite directions: $\uparrow\downarrow$. Because we do not currently have an alternator but annihilating the first two edges results in an alternator, we must have $\uparrow\downarrow(\downarrow\uparrow)^{k-1}$. But we can annihilate the third and fourth edges to get $\uparrow\downarrow(\downarrow\uparrow)^{k-2}$, which somehow has to be an alternator. It is only an alternator when $k=2$. Therefore, the only way for this process to terminate is if $d=4$ and we get a synchronizer.
        \qedhere
    \end{itemize}
\end{proof}

\begin{figure}
\centering
\begin{tikzpicture}[
node/.style={circle, draw, thick, circle, draw, minimum size=5mm, inner sep=0pt},
edge/.style={-latex, thick}
]
\begin{scope}[shift={(-0.85,0)}]
    \node[node] (A) at (-0.8,0) {$=$};
    \node[node] (B) at (0.8,0) {$=$};
    \draw[edge] ([yshift=0.1cm]B.west) -- ([yshift=0.1cm]A.east) node [midway, above] {$f-2$};
    \draw[edge] ([yshift=-0.1cm]B.west) -- ([yshift=-0.1cm]A.east);
    \draw[edge, red] (-1.33,0.53) -- (A);
    \draw[edge, red] (B) -- (1.33,0.53);
    \draw[edge, blue] (-1.33,-0.53) -- (A);
    \draw[edge, blue] (B) -- (1.33,-0.53);
\end{scope}
\draw (1.5,-1) -- (1.5,1);
\begin{scope}[shift={(3.7,0)}]
    \node[node] (A) at (-0.5,0) {$\cdot$};
    \node[node] (B) at (0.5,0) {$\cdot$};
    \draw[edge] (A) -- (B);
    \draw[edge, red] (-1.03,0.53) -- (A);
    \draw[edge, red] (B) -- (1.03,0.53);
    \draw[edge, blue] (-1.03,-0.53) -- (A);
    \draw[edge, blue] (B) -- (1.03,-0.53);
\end{scope}
\begin{scope}[shift={(6.7,0)}]
    \node[node] (A) at (-0.5,0) {$\cdot$};
    \node[node] (B) at (0.5,0) {$\cdot$};
    \draw[edge] ([yshift=0.1cm]B.west) -- ([yshift=0.1cm]A.east);
    \draw[edge] ([yshift=-0.1cm]A.east) -- ([yshift=-0.1cm]B.west);
    \draw[edge, red] (-1.03,0.53) -- (A);
    \draw[edge, red] (B) -- (1.03,0.53);
    \draw[edge, blue] (-1.03,-0.53) -- (A);
    \draw[edge, blue] (B) -- (1.03,-0.53);
\end{scope}
\end{tikzpicture}
\caption{Left: $f$-equalizers simulate a synchronizer. Right: Degree-$3$ or -$4$ duplicators of net flow $\pm1$ or $\pm2$ also simulate a synchronizer.}
\label{fig:sim-synch}
\end{figure}
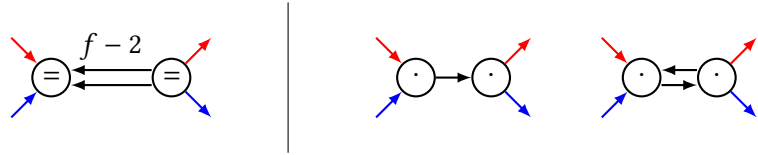

\begin{lemma}\label{lem:sim-lin-comb-dup}
    Duplicators with net flows $\pm f_1, \dots, \pm f_k$, together with a synchronizer, can simulate a duplicator with net flow $\pm f$ if and only if $f = \sum a_if_i$ for some integer coefficients $a_i$.
\end{lemma}
\begin{proof}
    First observe that in any satisfying assignment, the net flow of the simulated duplicator is the sum of net flows of the individual vertices. So $f$ must be an integer linear combination of $f_i$. 
    
    We now show that this is sufficient. Take $|a_i|$ instances of a duplicator with net flow $\pm f_i$. For each duplicator, choose the assignment that has net flow $+f_i$ if $a_i \ge 0$ and $-f_i$ otherwise. Connect these duplicators together using synchronizers (Figure~\ref{fig:any-duplicator}, left) to get \emph{some} duplicator with net flow $\pm f$.
    
    To simulate the target duplicator with net flow $\pm f$, place our current gadget inside the target gadget (Figure~\ref{fig:any-duplicator}, right). Pair up the edges so that their orientations align. Connect each matched pair with a path, ensuring that the entire structure is connected. Finally, place a synchronizer at every intersection point.
\end{proof}

\begin{figure}
\centering
\begin{tikzpicture}[
  node2/.style={circle, draw, thick, circle, draw, minimum size=5mm, inner sep=0pt},
  node/.style={rectangle, draw, thick, align=center, minimum height=8mm},
  sum/.style={circle, draw, thick, minimum size=5mm,inner sep=0pt},
  edge/.style={-latex, thick},
  arrow/.style={-latex, thick, blue}
]

\foreach \i in {1,2,3} {
\node[node] (d\i) at ({2*(\i-1)},0) {Duplicator};
\draw[arrow] (d\i) -- ({2*(\i-1)},-1);
\draw[arrow] ({2*(\i-1)},1) -- (d\i);
\ifthenelse{\i = 1 \OR \i = 3}{\draw[arrow] ({2*(\i-1)-0.4},-1) -- ([xshift=-0.4cm]d\i.south);}{};
\ifthenelse{\i = 1 \OR \i = 2}{\draw[arrow] ([xshift=-0.4cm]d\i.north) -- ({2*(\i-1)-0.4},1);}{};
\ifthenelse{\i = 2 \OR \i = 3}{\draw[arrow] ({2*(\i-1)+0.4},-1) -- ([xshift=0.4cm]d\i.south);}{};
\ifthenelse{\i = 1 \OR \i = 3}{\draw[arrow] ({2*(\i-1)+0.4},1) -- ([xshift=0.4cm]d\i.north);}{};
}

\node[sum] (s1) at (1,-1) {$S$};
\node[sum] (s2) at (3,1) {$S$};

\draw[edge] (s1) -- (d1);
\draw[edge] (s1) -- (d2);
\draw[edge] (s2) -- (d2);
\draw[edge] (d3) -- (s2);

\draw[arrow] (1.4,-1.75) -- (s1);
\draw[arrow] (0.6,-1.75) -- (s1);
\draw[arrow] (3.4,1.75) -- (s2);
\draw[arrow] (s2) -- (2.6,1.75);

\end{tikzpicture}
\hspace{1.5cm}
\begin{tikzpicture}[
  edge/.style={-latex, thick},
  dot/.style={circle, fill=red, inner sep=1.3pt},
  box/.style={rectangle, draw, thick, minimum width=3cm, minimum height=3cm},
  innerbox/.style={rectangle, draw, thick, minimum width=1.3cm, minimum height=1.3cm}
]

\node[box] (outer) at (0,0) {};
\node[innerbox] (inner) at (0,0) {Dup};

\draw[edge, blue] (-2,0.5) -- ([yshift=0.5cm]outer.west);
\draw[edge, blue] ([yshift=-0.5cm]outer.west) -- (-2,-0.5);
\draw[edge, blue] (2,0.5) -- ([yshift=0.5cm]outer.east);
\draw[edge, blue] ([yshift=-0.5cm]outer.east) -- (2,-0.5);
\draw[edge, blue] ([xshift=0.5cm]outer.north) -- (0.5,2);
\draw[edge, blue] (-0.5,2) -- ([xshift=-0.5cm]outer.north);
\draw[edge, blue] (0.5,-2) -- ([xshift=0.5cm]outer.south);
\draw[edge, blue] (-0.5,-2) -- ([xshift=-0.5cm]outer.south);

\draw[edge] (inner.west) -- (-0.35,0); %
\draw[edge] (inner.east) -- (0.35,0);   %
\draw[edge] (0,-0.35) -- (inner.south); %
\draw[edge] (inner.north) -- (0,0.35);   %

\draw[edge, name path=curve1] ([xshift=-0.5cm]outer.south) .. controls (-0.5,-1) and (-1,-0.5) .. ([yshift=-0.5cm]outer.west);
\draw[edge, rounded corners=8pt, name path=curve2] ([xshift=0.5cm]outer.south) -- (0.5,-1.1) -- (-1.1,-1.1) -- (-1.1,0.95) -- (1.1,0.95) -- (1.1,0) -- (inner.east);
\draw[edge, rounded corners, name path=curve3] ([yshift=0.5cm]outer.west) -- (-1.3,0.5) -- (-0.9,0) -- (inner.west);
\draw[edge, rounded corners, name path=curve4] ([xshift=-0.5cm]outer.north) -- (-0.5,1.1) -- (0.5,1.1) -- ([xshift=0.5cm]outer.north);
\draw[edge, rounded corners, name path=curve5] (inner.south) -- (0,-0.9) -- (0.9,-0.9) -- (1.1,-0.5) -- ([yshift=-0.5cm]outer.east);
\draw[edge, rounded corners, name path=curve6] ([yshift=0.5cm]outer.east) -- (1.3,0.5) -- (1.3,1.2) -- (0,1.2) -- (inner.north);

\path [name intersections={of=curve1 and curve2, by={int1,int2}}];
\path [name intersections={of=curve2 and curve3, by={int3}}];
\path [name intersections={of=curve2 and curve6, by={int4}}];
\path [name intersections={of=curve4 and curve6, by={int5,int6}}];
\node[dot] at (int1) {};
\node[dot] at (int2) {};
\node[dot] at (int3) {};
\node[dot] at (int4) {};
\node[dot] at (int5) {};
\node[dot] at (int6) {};

\end{tikzpicture}
\caption{Left: Linking copies of a duplicator with synchronizers. Right: Simulating any duplicator of net flow $\pm f$. A synchronizer is placed at every red dot.}
\label{fig:any-duplicator}
\end{figure}
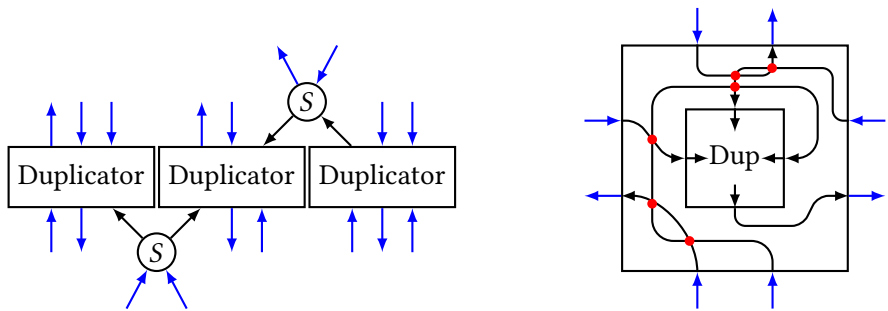

Combining the two previous lemmas, we get the following:

\begin{corollary}\label{lem:sim-gcd-dup}
    A non-trivial set of duplicators with net flows $\pm f_1, \dots, \pm f_k$ can simulate a duplicator with net flow $\pm f$ if and only if $f$ is an integer multiple of $\gcd(f_1, \dots, f_k)$.
\end{corollary}

We can finally prove the main theorem:

\begin{proof}[Proof of Theorem~\ref{thm:dup-characterization}]
    Let $\Gamma$ be a finite set of duplicators with net flows $\pm f_1, \dots, \pm f_k$.
    \begin{enumerate}
        \item If $\Gamma$ has only trivial duplicators, there are very limited ways to combine them while staying connected, and it is easy to check. Moreover, one can always vacuously simulate $1$-in-$2$, as it is just an edge.
        \item If $\Gamma$ has only alternators, then, by Lemma~\ref{lem:sim-alternators}, they can simulate other alternators and nothing else.
        \item If $\Gamma$ has at least one non-trivial non-alternator, then by Lemma~\ref{lem:sim-gcd-dup} they can simulate duplicators with net flow that is any multiple of $\gcd(f_1, \dots, f_k)$, and nothing else. It is also possible that $\Gamma$ has a non-trivial alternator and a trivial non-alternator (namely, $\{0,1\}$-in-$1$ or $\{0,2\}$-in-$2$). In this case we can combine them to get a non-trivial non-alternator and proceed with the proof.
        \qedhere
    \end{enumerate}
\end{proof}

This characterization has a useful corollary:
\begin{corollary}\label{cor:dups-equiv-dup}
    Any non-trivial set of duplicators is equivalent to a single duplicator. That is, for any non-trivial set $\Gamma_d$ of duplicators, there exists a duplicator $D$, such that $\Gamma_d$ together simulate $D$ and $D$ simulates each of $\Gamma_d$.
\end{corollary}
\begin{proof}
If all duplicators are alternators, they are equivalent to an alternator. Otherwise, they are equivalent to a non-alternator duplicator with net flow $\pm \gcd(f_1, \dots, f_k)$.
\end{proof}

It turns out that NP-hardness usually requires being able to simulate a synchronizer, and a synchronizer usually suffices. On the other hand, an alternator is often too weak.

\subsection{$S$-in-$j$ Graph Orientation with Constants}\label{subsec:S-in-j-with-const}

We now consider (Planar) $\Gamma$-Graph Orientation assuming we have constants (both $0$-in-$1$ and $1$-in-$1$), which we denote by (Planar) $\Gamma$-GO$_c$. In this subsection, we provide a complete dichotomy for $\Gamma$-GO$_c$ and Planar $\Gamma$-GO$_c$.

Connecting a constant to a vertex drops the first or last bit of its bitstring. For example, $\{0,1,4\}$-in-$5$ has bitstring \texttt{110010}. If we force one of the edges to point inwards, we get $\{0,3\}$-in-$4$ (\texttt{10010}). Forcing another edge to point out, we get $\{0,3\}$-in-$3$ (\texttt{1001}). We can use this fact to extract any substring of a bitstring:

\begin{observation}\label{obs:substrings}
    If the bitstring of $S'$-in-$j'$ is a substring of the bitstring of $S$-in-$j$, then $S$-in-$j$ together with constants simulates $S'$-in-$j'$.
\end{observation}

In particular, we can use gaps to create equalizers:

\begin{lemma}\label{lem:sim-k-dup}
    If $S$ has a gap of size $k - 1$ (i.e., $a,a+k \in S$, but $a+1, \dots, a+k-1 \notin S$), then $S$-in-$j$ and constants simulate a $k$-equalizer.
\end{lemma}
\begin{proof}
    Attaching $a$ inwards edges and $j-(a+k)$ outwards edges to $S$-in-$j$ simulates a $k$-equalizer.
\end{proof}

Let $\Delta S$ be the set of differences between consecutive elements of $S$ (for example, $\Delta \{0,1,4\} = \{1,3\}$). In other words, $k \in \Delta S$ if and only if $S$ has a gap of size $k-1$. For a set $\Gamma$ of $S$-in-$j$ vertex types, define $\Delta \Gamma = \bigcup_{\text{$S$-in-$j$} \in \Gamma} \Delta S$.

\begin{lemma}\label{lem:gcd-delta-gamma}
If $\max \Delta \Gamma > 2$, $\Gamma$ can simulate any duplicator whose net flow is a multiple of $\gcd \Delta \Gamma$.
\end{lemma}

\begin{proof}
    For every $k \in \Delta \Gamma$, we can simulate a $k$-equalizer by Lemma~\ref{lem:sim-k-dup}. Because $\max \Delta \Gamma > 2$, one of these duplicators is non-trivial, so by Theorem~\ref{thm:dup-characterization} we can simulate any duplicator with net flow that is a multiple of $\gcd \Delta \Gamma$.
\end{proof}

On the other hand, if $\max \Delta \Gamma \le 2$, we do not have any gaps of size $2$ or more, so $\Gamma$-GO$_c$ is in P by Lemma~\ref{lem:classic-easy-cases}.

We can now state and prove the dichotomy for Graph Orientation with constants:

\begin{theorem}\label{thm:S-in-j-and-const}
    Let $\Gamma$ be a set of $S$-in-$j$ vertex types. Then $\Gamma$-GO$_c$ is NP-hard, except for the following cases, which are in P:
    \begin{enumerate}
        \item all vertex types in $\Gamma$ are bijunctive;
        \item all vertex types in $\Gamma$ are affine;
        \item $\max \Delta \Gamma \le 2$.
    \end{enumerate}
\end{theorem}

\begin{proof}
    These 3 cases are in P by Lemma~\ref{lem:classic-easy-cases}. Suppose we do not have any of them.

    By Lemma~\ref{lem:gcd-delta-gamma}, we can simulate any duplicator with net flow that is a multiple of $k := \gcd \Delta \Gamma$.

    Let $\Gamma' = \Gamma \cup \{R' \mid R \in \Gamma\}$, where $R'$ is the relation obtained from $R$ by negating all inputs.
    
    We reduce from Planar $\Gamma'$-SAT$_c$. Not all relations in $\Gamma'$ are affine, and not all are bijunctive. If all relations in $\Gamma'$ were Horn, then every relation in $\Gamma$ would be both Horn and dual-Horn at the same time. However, a relation that is both Horn and dual-Horn is also bijunctive (because bitwise MAJ can be expressed in terms of bitwise AND and bitwise OR). So not all relations in $\Gamma'$ are Horn, and similarly, not all are dual-Horn. By Theorem~\ref{thm:schaefer-const}, the problem we are reducing from is in fact NP-hard. 

    Duplicate every clause $k$ times. Now every variable is used a multiple of $k$ times positively and a multiple of $k$ times negatively. Take the usual clause-variable bipartite graph, and represent each clause by the corresponding vertex type from $\Gamma$ and each variable by a duplicator with net flow a multiple of $k$ (as said earlier, we can simulate such a duplicator). In a satisfying orientation, edges directed from variable vertices to clause vertices encode true literals, and the reverse direction encodes false literals.
\end{proof}

For the planar dichotomy, we will need to construct a crossover. In the following, we describe several scenarios where we can construct a crossover:

\begin{lemma}\label{lem:crossover}
    The following sets of vertex types simulate a crossover:
    \begin{enumerate}
        \item $2$-in-$4$ and synchronizer;
        \item $1$-in-$3$, $2$-in-$3$, and synchronizer;
        \item $\{1,j\}$-in-$j$ and constants (for $j \ge 4$);
        \item $1$-in-$3$ and $3$-equalizer;
        \item $1$-in-$3$ and $4$-equalizer.
    \end{enumerate}
\end{lemma}

\begin{proof}
    See Figure~\ref{fig:crossover-construction}.
    \begin{enumerate}
        \item We can construct the crossover shown in Figure~\ref{fig:crossover-2-in-4-sync}.
        \item We can construct a $2$-in-$4$ using $1$-in-$3$ and $2$-in-$3$ as shown in Figure~\ref{fig:2-in-4-from-1-in-3-and-2-in-3}, thus reducing to Case~1.
        \item Connect two $\{1,j\}$-in-$j$ vertices with $j-2$ edges as shown in Figure~\ref{fig:2-in-3-from-1j-in-j-and-const}. Fix one of the edges to point in from a constant. The in-degree of every vertex is $1 \pmod{j-1}$, so the in-degree of the entire gadget is $2 \pmod{j-1}$. But it is also at most $3$. So it is exactly $2$. It is not hard to check that any pair of edges can point in, so it is $2$-in-$3$. Feeding constants into a $\{1,j\}$-in-$j$ makes $1$-in-$3$ and $(j-1)$-equalizers, which make a synchronizer, which reduces to Case~2. %
        \item We can construct the crossover shown in Figure~\ref{fig:crossover-from-1-in-3-3eq}. The middle $1$-in-$4$ is formed from two $1$-in-$3$s. Its sole incoming edge forces all other edges and ensures that the gadget is one of the four crossover states.
        \item We can construct a $2$-in-$3$ as shown in Figure~\ref{fig:2-in-3-from-1-in-3-4eq}. The $4$-equalizer can simulate a synchronizer, reducing to Case~2.\qedhere
    \end{enumerate}
\end{proof}

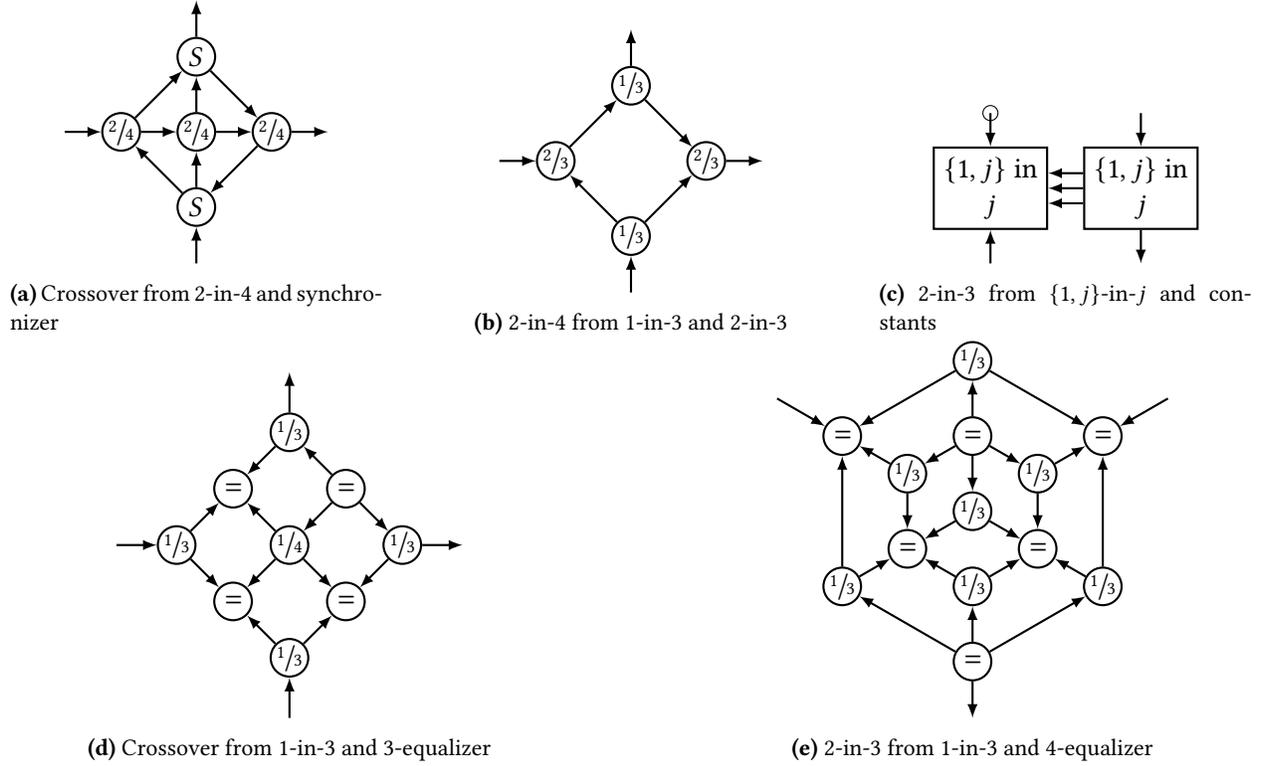
\begin{figure}
\centering
\begin{subfigure}{0.3\textwidth}
\centering
\begin{tikzpicture}[
node/.style={circle, draw, thick, circle, draw, minimum size=5mm, inner sep=0pt},
edge/.style={-latex, thick}
]
    \node[node] (A) at (0,1) {$S$};
    \node[node] (B) at (-1,0) {$\sfrac24$};
    \node[node] (C) at (0,0) {$\sfrac24$};
    \node[node] (D) at (1,0) {$\sfrac24$};
    \node[node] (E) at (0,-1) {$S$};
    \draw[edge] (A) -- (0,1.75);
    \draw[edge] (B) -- (A);
    \draw[edge] (C) -- (A);
    \draw[edge] (A) -- (D);
    \draw[edge] (-1.75,0) -- (B);
    \draw[edge] (B) -- (C);
    \draw[edge] (C) -- (D);
    \draw[edge] (D) -- (1.75,0);
    \draw[edge] (E) -- (B);
    \draw[edge] (E) -- (C);
    \draw[edge] (D) -- (E);
    \draw[edge] (0,-1.75) -- (E);
\end{tikzpicture}
\subcaption{Crossover from $2$-in-$4$ and synchronizer}
\label{fig:crossover-2-in-4-sync}
\end{subfigure}%
\hfill%
\begin{subfigure}{0.3\textwidth}
\centering
\begin{tikzpicture}[
node/.style={circle, draw, thick, circle, draw, minimum size=5mm, inner sep=0pt},
edge/.style={-latex, thick}
]
    \node[node] (A) at (0,1) {$\sfrac13$};
    \node[node] (B) at (-1,0) {$\sfrac23$};
    \node[node] (D) at (1,0) {$\sfrac23$};
    \node[node] (E) at (0,-1) {$\sfrac13$};
    \draw[edge] (A) -- (0,1.75);
    \draw[edge] (B) -- (A);
    \draw[edge] (A) -- (D);
    \draw[edge] (-1.75,0) -- (B);
    \draw[edge] (D) -- (1.75,0);
    \draw[edge] (E) -- (B);
    \draw[edge] (E) -- (D);
    \draw[edge] (0,-1.75) -- (E);
\end{tikzpicture}
\subcaption{$2$-in-$4$ from $1$-in-$3$ and $2$-in-$3$}
\label{fig:2-in-4-from-1-in-3-and-2-in-3}
\end{subfigure}%
\hfill%
\begin{subfigure}{0.3\textwidth}
\centering
\begin{tikzpicture}[
node/.style={circle, draw, thick, circle, draw, minimum size=5mm, inner sep=0pt},
edge/.style={latex-, thick}
]
    \draw (-1,0) node[rectangle,draw,thick,align=center] (A) {$\{1,j\}\text{ in}$\\$j$};
    \draw (1,0) node[rectangle,draw,thick,align=center] (B) {$\{1,j\}\text{ in}$\\$j$};
    \draw[edge] (A) -- (B);
    \draw[edge] ([yshift=0.2cm]A.east) -- ([yshift=0.2cm]B.west);
    \draw[edge] ([yshift=-0.2cm]A.east) -- ([yshift=-0.2cm]B.west);
    \draw (-1,1) circle (0.1);
    \draw[edge] (A) -- (-1,1);
    \draw[edge] (A) -- (-1,-1);
    \draw[edge] (B) -- (1,1);
    \draw[edge] (1,-1) -- (B);
\end{tikzpicture}
\subcaption{$2$-in-$3$ from $\{1,j\}$-in-$j$ and constants}
\label{fig:2-in-3-from-1j-in-j-and-const}
\end{subfigure}

\begin{subfigure}{0.45\textwidth}
\centering
\begin{tikzpicture}[
node/.style={circle, draw, thick, circle, draw, minimum size=5mm, inner sep=0pt},
edge/.style={-latex, thick}
]
    \node[node] (A) at (0,1.5) {$\sfrac13$};
    \node[node] (B) at (-0.75,0.75) {$=$};
    \node[node] (C) at (0.75,0.75) {$=$};
    \node[node] (D) at (-1.5,0) {$\sfrac13$};
    \node[node] (E) at (0,0) {$\sfrac14$};
    \node[node] (F) at (1.5,0) {$\sfrac13$};
    \node[node] (G) at (-0.75,-0.75) {$=$};
    \node[node] (H) at (0.75,-0.75) {$=$};
    \node[node] (I) at (0,-1.5) {$\sfrac13$};
    \draw[edge] (A) -- (0,2.3);
    \draw[edge] (A) -- (B);
    \draw[edge] (C) -- (A);
    \draw[edge] (D) -- (B);
    \draw[edge] (E) -- (B);
    \draw[edge] (C) -- (E);
    \draw[edge] (C) -- (F);
    \draw[edge] (-2.3,0) -- (D);
    \draw[edge] (F) -- (2.3,0);
    \draw[edge] (D) -- (G);
    \draw[edge] (E) -- (G);
    \draw[edge] (E) -- (H);
    \draw[edge] (F) -- (H);
    \draw[edge] (I) -- (G);
    \draw[edge] (I) -- (H);
    \draw[edge] (0,-2.3) -- (I);
\end{tikzpicture}
\subcaption{Crossover from $1$-in-$3$ and $3$-equalizer}
\label{fig:crossover-from-1-in-3-3eq}
\end{subfigure}%
\hfill%
\begin{subfigure}{0.45\textwidth}
\centering
\begin{tikzpicture}[
node/.style={circle, draw, thick, circle, draw, minimum size=5mm, inner sep=0pt},
edge/.style={-latex, thick}
]
    \node[node] (A) at (0,2) {$\sfrac13$};
    \node[node] (B) at (-1.732,1) {$=$};
    \node[node] (C) at (0,1) {$=$};
    \node[node] (D) at (1.732,1) {$=$};
    \node[node] (E) at (-0.866,0.5) {$\sfrac13$};
    \node[node] (F) at (0.866,0.5) {$\sfrac13$};
    \node[node] (G) at (0,0) {$\sfrac13$};
    \node[node] (H) at (-0.866,-0.5) {$=$};
    \node[node] (I) at (0.866,-0.5) {$=$};
    \node[node] (J) at (-1.732,-1) {$\sfrac13$};
    \node[node] (K) at (0,-1) {$\sfrac13$};
    \node[node] (L) at (1.732,-1) {$\sfrac13$};
    \node[node] (M) at (0,-2) {$=$};
    \draw[edge] (A) -- (B);
    \draw[edge] (C) -- (A);
    \draw[edge] (A) -- (D);
    \draw[edge] (-2.598,1.5) -- (B);
    \draw[edge] (2.598,1.5) -- (D);
    \draw[edge] (E) -- (B);
    \draw[edge] (C) -- (E);
    \draw[edge] (C) -- (F);
    \draw[edge] (F) -- (D);
    \draw[edge] (C) -- (G);
    \draw[edge] (E) -- (H);
    \draw[edge] (F) -- (I);
    \draw[edge] (G) -- (H);
    \draw[edge] (G) -- (I);
    \draw[edge] (J) -- (B);
    \draw[edge] (J) -- (H);
    \draw[edge] (K) -- (H);
    \draw[edge] (K) -- (I);
    \draw[edge] (L) -- (I);
    \draw[edge] (L) -- (D);
    \draw[edge] (M) -- (J);
    \draw[edge] (M) -- (K);
    \draw[edge] (M) -- (L);
    \draw[edge] (M) -- (0,-2.75);
\end{tikzpicture}
\subcaption{$2$-in-$3$ from $1$-in-$3$ and $4$-equalizer}
\label{fig:2-in-3-from-1-in-3-4eq}
\end{subfigure}
\caption{Constructions used in Lemma~\ref{lem:crossover}. Vertices labeled $S$ denote synchronizers, $\sfrac{a}{b}$ denote $a$-in-$b$, $=$ denote equalizers.}
\label{fig:crossover-construction}
\end{figure}

It turns out that all of these cases are also NP-hard:

\begin{lemma}\label{lem:base-hard-cases}
    Planar Graph Orientation is NP-hard with the following sets of vertex types:
    \begin{enumerate}
        \item $2$-in-$4$ and synchronizer;
        \item $1$-in-$3$, $2$-in-$3$, and synchronizer;
        \item $\{1,j\}$-in-$j$ and constants (for $j \ge 4$);
        \item $1$-in-$3$ and $3$-equalizer;
        \item $1$-in-$3$ and $4$-equalizer.
    \end{enumerate}
\end{lemma}

\begin{proof}
    See Figure~\ref{fig:crossover-construction}.
    \begin{enumerate}
        \item Reduction from $2$-in-$4$SAT. Duplicate every clause and negate all literals in the copy (which does not change the clause's meaning). Now every variable appears positively $k$ times and negatively $k$ times for some $k$, so we can represent it with a duplicator of degree $2k$ and net flow $0$. Finally, we have crossovers from Lemma~\ref{lem:crossover}.
        \item Reduction from Case~1 as shown in Figure~\ref{fig:2-in-4-from-1-in-3-and-2-in-3}.
        \item Reduction from Case~2 as shown in Figure~\ref{fig:2-in-3-from-1j-in-j-and-const}.
        \item Reduction from $1$-in-$3$SAT. Create three copies of each clause. Now every variable is used $3k$ times for some $k$, so we can represent it with a $3k$-equalizer. Finally, we have crossovers from Lemma~\ref{lem:crossover}.
        \item Reduction from Case~2 as shown in Figure~\ref{fig:2-in-3-from-1-in-3-4eq}.\qedhere
    \end{enumerate}
\end{proof}

We showed Planar Graph Orientation is NP-hard with $1$-in-$3$ and a $3$-equalizer or a $4$-equalizer. It is natural to ask what happens if we instead have a $5$-equalizer. Surprisingly, there is a phase transition, and the problem becomes polynomial!

\begin{lemma}\label{lem:k>=5-easy}
    For each fixed $k \ge 5$, Planar Graph Orientation where all vertices are either unsatisfiable, terminators, $k$-equalizers, or $1$-in-$j$ for some $j$ is in P.
\end{lemma}

We will prove this lemma in Section~\ref{sec:k>=5}, but for now we can state and prove the dichotomy for Planar Graph Orientation with constants:%

\begin{theorem}\label{thm:planar-S-in-j-and-const}
    Let $\Gamma$ be a set of $S$-in-$j$ vertices. Then Planar $\Gamma$-GO$_c$ is NP-hard, except for the three polynomial cases from Theorem~\ref{thm:S-in-j-and-const}, and two new polynomial cases:
    \begin{enumerate}
        \setcounter{enumi}{3}
        \item $\gcd \Delta \Gamma \ge 5$, and all $S$-in-$j$ in $\Gamma$ have $S \subseteq \{0,j\}$ or $S = \{1\}$;
        \item $\gcd \Delta \Gamma \ge 5$, and all $S$-in-$j$ in $\Gamma$ have $S \subseteq \{0,j\}$ or $S = \{j-1\}$.
    \end{enumerate}
\end{theorem}

\begin{proof}
    The first 3 cases are easy by Lemma~\ref{lem:classic-easy-cases}.

    Suppose $\gcd \Delta \Gamma \ge 5$ and all $S$-in-$j$ vertices have $S \subseteq \{0,j\}$ or $S = \{1\}$. If there is an unsatisfiable vertex, the problem is trivially unsatisfiable. For each $f$-equalizer vertex, $f$ is a multiple of $\gcd \Delta \Gamma$ by definition (because $f$ is one of the arguments to the GCD), so we can simulate it with $\gcd \Delta \Gamma$-equalizers. Now the problem is in P by Lemma~\ref{lem:k>=5-easy}.

    Similarly, if $\gcd \Delta \Gamma \ge 5$ and all $S$-in-$j$ vertices have $S \subseteq \{0,j\}$ or $S = \{j-1\}$, we can reverse all the edges before applying Lemma~\ref{lem:k>=5-easy}.

    Suppose we do not have any of the 5 easy cases.

    The only step in the proof of Theorem~\ref{thm:S-in-j-and-const} that does not preserve planarity is the duplication. If $\gcd \Delta \Gamma = 1$, then that step is a no-op, so we are done. Otherwise, it suffices to construct a crossover.

    Suppose $2 \le \gcd \Delta \Gamma \le 4$. We can construct either a $3$-equalizer (if $\gcd \Delta \Gamma = 3$) or a $4$-equalizer (if $\gcd \Delta \Gamma \in \{2,4\}$). If the bitstring of some $S$ contains \texttt{0100} or \texttt{0010}, we have $1$-in-$3$ or $2$-in-$3$, which simulates a crossover (Lemma~\ref{lem:crossover}, Case~4 or 5). We also know that our bitstrings cannot contain \texttt{11}, because then $\gcd \Delta \Gamma$ would be $1$. If there is \texttt{1} not at either end, then in order to avoid these three patterns, all bits must alternate between \texttt{0} and \texttt{1}, so it is affine. If \texttt{1}s are only at the ends, then it is also affine. So all vertex types are affine, which contradicts our assumption that we do not have any of the easy Schaefer cases.

    Now suppose $5 \le \gcd \Delta \Gamma$. If some bitstring contains \texttt{1} far away from the ends, then it must be surrounded by at least $2$ zeros on each side, so we have a $2$-in-$4$, which simulates a crossover (Lemma~\ref{lem:crossover}, Case~1). This leaves the following cases:
    \begin{itemize}
        \item \texttt{000..000} --- unsatisfiable
        \item \texttt{000..001} --- terminator
        \item \texttt{100..000} --- terminator
        \item \texttt{100..001} --- equalizer
        \item \texttt{000..010} --- $(j-1)$-in-$j$
        \item \texttt{010..000} --- $1$-in-$j$
        \item \texttt{100..010} --- simulates crossover by Lemma~\ref{lem:crossover}, Case~3 (reverse all edges)
        \item \texttt{010..001} --- simulates crossover by Lemma~\ref{lem:crossover}, Case~3
        \item \texttt{010..010} --- simulates crossover by Lemma~\ref{lem:crossover}, Case~2
    \end{itemize}

    Three of these can simulate a crossover on their own, so suppose we do not have any of them. To avoid the easy planar cases, we must have at least one $(j-1)$-in-$j$ and at least one $1$-in-$j$. But the former simulates $2$-in-$3$, and the latter simulates $1$-in-$3$, so together they simulate a crossover by Lemma~\ref{lem:crossover}, Case~2.
\end{proof}

\subsubsection{A Planar Easiness Result}\label{sec:k>=5}

We now prove Lemma~\ref{lem:k>=5-easy} by giving an algorithm. The high-level insight is that the large degree of the $k$-equalizers, combined with planarity, forces the existence of local simplifying structures in the graph. Formally, the crux of the algorithm is the following lemma, where we adopt the convention that self-loops have two possible orientations.

\begin{lemma}
Let $k \ge 5$, and let $G$ be an instance of Planar Graph Orientation where all vertices are terminators, $k$-equalizers, or $1$-in-$j$ for some $j$. Then as long as there is at least one $1$-in-$j$ vertex with $j \ge 3$, there is an edge whose orientation is fixed in all solutions. Moreover, this edge and its orientation can be found in polynomial time.
\end{lemma}

If we find such an edge, we can delete it and adjust the incident vertices appropriately. For example, if it pointed into a $1$-in-$j$ vertex, that vertex becomes $0$-in-$(j-1)$, i.e., a terminator. It is not hard to see that this always produces either an allowed vertex type (terminator, $k$-equalizers, or $1$-in-$j$), or an unsatisfiable vertex.

We repeat this until we get a contradiction, or there are no more $1$-in-$j$ vertices with $j \ge 3$, after which the problem becomes affine and can be finished in polynomial time.

\begin{proof}
Consider the following reduction rules in order and apply the first applicable one:
\begin{enumerate}
    \item If there is a terminator vertex, it fixes the orientation of its incident edges.
    \item If there is a $1$-in-$2$ vertex, contract it into an edge connecting its neighbors and recurse on the resulting graph.
    \item If there is a self-loop on a $k$-equalizer vertex, conclude the instance is unsatisfiable.
    \item If there is a self-loop on a $1$-in-$j$ vertex ($j \ge 3$), the other incident edges of this vertex must point outward. There is at least one other edge because $j \ge 3$.
    \item If there is an edge between a $1$-in-$j_1$ and a $1$-in-$j_2$ vertex, merge them into a $1$-in-$(j_1+j_2-2)$ vertex and recurse on the resulting graph.
    \item If there are at least two edges between a $k$-equalizer vertex and a $1$-in-$j$ vertex, both edges must be directed away from the $1$-in-$j$ vertex.
    \item If there is at least one edge (possibly multiple edges) between two distinct $k$-equalizer vertices, say $u$ and $v$, take one such edge $uv$. Starting from this edge, let $v$ followed by $u_1,\dots,u_{k'}$ be the neighbors of $u$ other than $v$ in clockwise order, and similarly let $u$ followed by $v_1,\dots,v_{k'}$ be the neighbors of $v$ other than $u$ in counterclockwise order ($0 \le k' < k$). Form a new graph by deleting $u$ and $v$ and connecting each $u_i$ with $v_i$, which preserves planarity.

    Note that every solution to the original instance has a corresponding solution in the new instance. For example, if the solution had $u_i \rightarrow u \leftarrow v \rightarrow v_i$, the corresponding solution has $u_i \rightarrow v_i$.
    
    We recurse on this new graph to find an edge $e$ whose orientation is the same in all solutions to the new instance. If $e$ is one of $u_iv_i$, then the orientation of $uv$ is fixed. Otherwise, $e$ must have the same orientation in all solutions to the original instance.
\end{enumerate}

If none of the above rules applies, then $G$ is simple, nonempty, planar, and bipartite, with $1$-in-$j$ and $k$-equalizer vertices forming a bipartition, and all $1$-in-$j$ vertices have $j \ge 3$.

If $k \ge 6$, we claim that this is impossible. The degrees of the $k$-equalizer vertices are all exactly $k$ and the degrees of the $1$-in-$j$ vertices are all at least $3$. Suppose that there are $x$ $k$-equalizer vertices and $y$ $1$-in-$j$ vertices. Then there are $kx$ edges, so $kx = \sum_{v\text{ is }1\text{-in-}j}\deg(v) \ge 3y \implies y \le \frac{kx}{3}$. Because $y > 0$, a corollary of Euler's formula for bipartite planar graphs is that $|E| \le 2|V|-4$. But $kx \le 2(x+y)-4 \le 2(x+\frac{kx}{3})-4 \implies \frac{kx}{3} \le 2x-4$ is a contradiction when $k \ge 6$. Therefore, this case is impossible.

Otherwise, $k = 5$. We add one more reduction rule:
\begin{enumerate}
    \setcounter{enumi}{7}
    \item Suppose the adjacent faces of a $5$-equalizer vertex all have degree $4$ and two non-consecutive neighboring vertices have degree $3$ (see Figure~\ref{fig:rule}). 
    
    If the blue vertex in Figure~\ref{fig:rule} was directed outwards (as shown), the other blue edges are forced, and a contradiction is reached at the red vertex. Therefore, the blue vertex must be directed inwards. Return one of its incident edges.\label{item:rule}
\end{enumerate}

\begin{figure}[h]
\centering
\begin{tikzpicture}[
node/.style={circle, draw, thick, circle, draw, minimum size=5mm, inner sep=0pt},
edge/.style={-latex, thick}
]
\node[node] (center) at (0,0) {$=$};

\node[node] (outer0) at ({2.5*cos(72*(0)-18)}, {2.5*sin(72*(0)-18)}) {$=$};
\node[node] (outer1) at ({2.5*cos(72*(1)-18)}, {2.5*sin(72*(1)-18)}) {$=$};
\node[node, blue] (outer2) at ({2.5*cos(72*(2)-18)}, {2.5*sin(72*(2)-18)}) {$=$};
\node[node] (outer3) at ({2.5*cos(72*(3)-18)}, {2.5*sin(72*(3)-18)}) {$=$};
\node[node] (outer4) at ({2.5*cos(72*(4)-18)}, {2.5*sin(72*(4)-18)}) {$=$};

\node[node, green] (mid0) at ({1.545*cos(72*(0)+18)}, {1.545*sin(72*(0)+18)}) {$\sfrac13$};
\node[node] (mid1) at ({1.545*cos(72*(1)+18)}, {1.545*sin(72*(1)+18)}) {$\sfrac1?$};
\node[node] (mid2) at ({1.545*cos(72*(2)+18)}, {1.545*sin(72*(2)+18)}) {$\sfrac1?$};
\node[node, green] (mid3) at ({1.545*cos(72*(3)+18)}, {1.545*sin(72*(3)+18)}) {$\sfrac13$};
\node[node, red] (mid4) at ({1.545*cos(72*(4)+18)}, {1.545*sin(72*(4)+18)}) {$\sfrac1?$};

\foreach \i in {0,...,4} {
    \draw[edge, blue] (mid\i) -- (center);
    \node[green] at ({1.25*cos(72*(\i)-18)}, {1.25*sin(72*(\i)-18)}) {$4$};
    \draw[thick] (outer\i) -- ({3.3*cos(72*(\i)-18)}, {3.3*sin(72*(\i)-18)});
    \draw[thick] (outer\i) -- ({3.3*cos(72*(\i)-18+9)}, {3.3*sin(72*(\i)-18+9)});
    \draw[thick] (outer\i) -- ({3.3*cos(72*(\i)-18-9)}, {3.3*sin(72*(\i)-18-9)});
}
\foreach \i in {1,2,4} {
    \draw[thick] (mid\i) -- ({3*cos(72*(\i)+18+5)}, {3*sin(72*(\i)+18+5)});
    \draw[thick] (mid\i) -- ({3*cos(72*(\i)+18-5)}, {3*sin(72*(\i)+18-5)});
}
\draw[edge, blue] (outer2) -- (mid1);
\draw[edge, blue] (outer2) -- (mid2);
\draw[edge, blue] (mid1) -- (outer1);
\draw[edge, blue] (mid2) -- (outer3);
\draw[edge, blue] (mid0) -- (outer1);
\draw[edge, blue] (mid3) -- (outer3);
\draw[edge, blue] (outer0) -- (mid0);
\draw[edge, blue] (outer4) -- (mid3);
\draw[edge, blue] (outer0) -- (mid4);
\draw[edge, blue] (outer4) -- (mid4);
\end{tikzpicture}
    \caption{The structure in reduction rule \ref{item:rule}. Green vertices/faces denote degree constraints. }
    \label{fig:rule}
\end{figure}
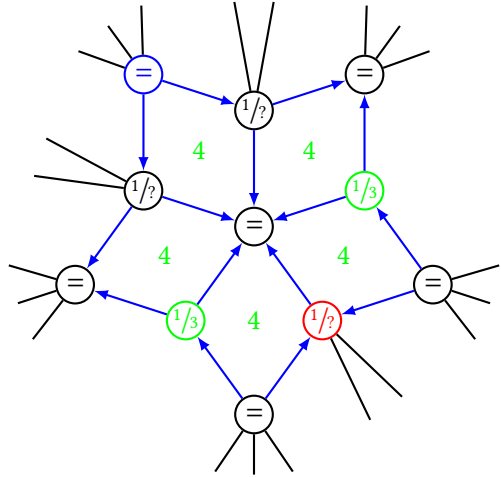

We show that this structure must always appear.

Suppose there are $x$ $5$-equalizer vertices and $y$ $1$-in-$j$ vertices. We have $|V| = x+y$, $|E| = 5x$, and Euler's formula $|V|-|E|+|F| = 2$. They imply that $(x+y)-5x+|F| = 2$, so $y+|F| = 4x+2$.

We use the discharging method \cite{discharging}. Initially, assign a charge of $6-\frac{3}{2}\deg(f)$ on each face $f$, a charge of $6-2\deg(v)$ on each $1$-in-$j$ vertex $v$, and a charge of $1$ on each $5$-equalizer vertex. The sum of all charges is
\begin{align*}
    &x + \sum_{f \in F} 6-\frac{3}{2}\deg(f) + \sum_{v\text{ is }1\text{-in-}j}6-2\deg(v) \\
    &= x + 6|F|-\frac{3}{2}(2|E|) + 6y-2(|E|) = x + 6(y+|F|)-5|E| \\
    &= x + 6(4x+2) - 5(5x) = 12 > 0.
\end{align*}
Now, apply the following discharging rules:
\begin{enumerate}
    \item Each face $f$ with $\deg(f) \ge 6$ takes $1$ charge from each of its $\frac12 \deg(f)$ $5$-equalizer neighbors.
    \item Each $1$-in-$j$ vertex $v$ with $\deg(v) \ge 4$ takes $\frac{1}{3}$ charge from each of its neighboring $5$-equalizer vertices.
\end{enumerate}
The final charge on each face $f$ is either $0$ if $\deg(f) = 4$ or $6-\frac{3}{2}\deg(f)+\frac12\deg(f) \le 0$ if $\deg(f) \ge 6$. The final charge on each $1$-in-$j$ vertex $v$ is either $0$ if $\deg(v) = 3$ or $6-2\deg(v)+\frac{1}{3}\deg(v) < 0$ if $\deg(v) \ge 4$. As the total charge is positive, some vertex $v$ has a positive final charge, which must be a $5$-equalizer vertex. This vertex cannot be adjacent to a face of degree $\ge 6$ nor $3$ or more $1$-in-$j$ vertices of degree $\ge 4$. Thus, some two non-consecutive neighbors of $v$ both have degree $3$, so $v$ is the center vertex of the structure in Figure~\ref{fig:rule}.

Therefore, one of the above rules always applies.
\end{proof}

Planar Graph Orientation with $1$-in-$j$ and $k$-equalizers is closely related to Planar Exact Cover by sets of size $k$ (Planar X$k$C). The $k$-equalizer vertices correspond to sets, and $1$-in-$j$ vertices correspond to elements. Planar X$3$C is well known to be NP-complete \cite{planar3dm}, but to our knowledge, other values of $k$ have not been studied. We note the following corollary.
\begin{corollary}
    Planar Exact Cover by sets of size $k$ is NP-hard for $k = 3,4$ and in P for all other $k$.
\end{corollary}

\subsection{Removing Constants}\label{subsec:removing-const}

Now suppose that instead of single constants, we have $0$-in-$f$ and $f$-in-$f$, so the constants appear only in groups of $f$. We claim that this does not change the dichotomy.

\begin{definition}
    An \defn{$f$-terminator} is a degree-$|f|$ vertex type that can be satisfied only by $f$ incoming edges (if $f > 0$) or $-f$ outgoing edges (if $f < 0$).

    A \defn{positive terminator} is a $f$-terminator for some $f > 0$, and a \defn{negative terminator} is a $f$-terminator for some $f < 0$.
\end{definition}

Note that $k \ge 1$ disconnected $f$-terminators can simulate a $kf$-terminator. We now prove that having both positive and negative terminators is almost as good as having constants, by establishing an equivalence between Planar $\Gamma$-GO with $f$-terminators and Planar $\Gamma$-GO$_c$.

\begin{lemma}\label{lem:nonlocally-const}
    Let $\Gamma$ be a set of $S$-in-$j$ and duplicator vertex types such that $\Gamma$ can simulate both positive and negative terminators. Then (Planar) $\Gamma$-GO is equivalent to (Planar) $\Gamma$-GO$_c$.
\end{lemma}
\begin{proof}
    Suppose we can simulate a $-f_1$-terminator and a $+f_2$-terminator. Then we can simulate an $f$-terminator and a $-f$-terminator for $f := f_1f_2$. %

    Let $G$ be a (Planar) $\Gamma$-GO$_c$ instance. We perform a series of transformations to $G$ until we can replace all constants with $\pm f$-terminators as the last step. First, make $f$ copies of $G$, so that the total net flow from all the constants is a multiple of $f$. In the non-planar setting, we are done. The rest of this proof is for the planar setting.

    If all vertex types are bijunctive, both Planar $\Gamma$-GO and Planar $\Gamma$-GO$_c$ are in P by Lemma~\ref{lem:classic-easy-cases}. Similarly, if no $S$-in-$j$ has a gap of size $2$ or more and we do not have any non-trivial duplicators, then both problems are in P by Lemma~\ref{lem:classic-easy-cases} (all trivial duplicators are symmetric and do not have gaps of size $2$ or more). We will assume that neither of these easy cases holds.

    \begin{enumerate}
        \item We first consider the case that the bitstring of some $S$ contains a \emph{strict} substring of the form \texttt{10$^k$1} ($k \ge 1$). We connect two copies of $S$-in-$j$ with $k$ edges and add at least one constant to simulate a $1$-in-$2$.
        
        Subdivide each (non-constant) edge of the duplicated graph with $f$ copies of this $1$-in-$2$ (Figure~\ref{fig:sim-1-in-2}). Crucially, each original edge now has at least $f$ constants hanging off of it. We have the freedom to choose which of these constants are on the top side of the edge and which are on the bottom side. We will use this freedom to make sure that the constants in each face of the graph add up to a multiple of $f$.
    
        Let $T$ be a spanning tree of the dual graph. We now ``fix'' each face (vertex in $T$) so that the total net flow of all constants within the face is $0 \pmod f$. Process the faces in a post-order traversal of $T$ (starting at an arbitrary root). For each face $u$, consider the edge adjacent to its parent $p$ in $T$. By moving some of the (at least $f$) constants hanging off of this edge to the opposite face, we can make $u$ have net flow $0 \pmod f$ without disturbing the net flow of any other face except $p$. After repeating this process for all non-root faces, the root face automatically has $0 \pmod f$ net flow because the duplicated graph began with $0 \pmod f$ total net flow from constants.
        
        Now each face has some positive and negative constants hanging inside the face. Attach an $f$-terminator to all constants that point inward to the face to split them into $f-1$ outward-pointing constants. We now have only (a multiple of $f$ number of) outward-pointing constants, so we can simulate them in groups of $f$ using $-f$-terminators.
        
        \item If the bitstring of some $S$ contains a \emph{strict} substring \texttt{010}, feed constants to it to simulate a $1$-in-$2$. Subdivide each graph edge with $f$ of these. The face fixing procedure still works.

        \item If we have an alternator and the bitstring of some $S$ contains a \emph{strict} substring \texttt{110}, we can feed constants to it to simulate a $\{0,1\}$-in-$2$, and instead replace each graph edge with the construction in Figure~\ref{fig:sim-01-in-2}. This introduces two new faces, which are also included in the dual spanning tree.
        \item If we have an alternator and the bitstring of some $S$ contains a \emph{strict} substring \texttt{011}, we can do a symmetric construction.
    \end{enumerate}

    We claim that this casework is exhaustive. To avoid the first case, each of our vertex types must be either an interval $[a,b]$-in-$j$, a duplicator, or unsatisfiable. By assumption, not all of our vertex types are bijunctive, so we must have some $[a,b]$-in-$j$ with $j \ge 3$ (so the bitstring has at least 4 bits). Similarly, we must have a non-trivial duplicator, because intervals do not have any gaps of size $2$ or more. By Theorem~\ref{thm:dup-characterization}, every non-trivial duplicator simulates an alternator.

    To avoid cases 2-4, the bitstring of our interval cannot contain \texttt{010}, \texttt{110}, or \texttt{011}. This is only possible if it is \texttt{10..00}, \texttt{00..01}, or \texttt{11..11}. However, all of these are also bijunctive.
\end{proof}

\begin{figure}
\centering
\begin{subfigure}{\linewidth}
\centering
\begin{tikzpicture}[
  node2/.style={circle, draw, thick, circle, draw, minimum size=5mm, inner sep=0pt},
  node/.style={rectangle, draw, thick, align=center, minimum width=10mm, minimum height=10mm, inner sep=0pt},
  edge/.style={-latex, thick},
  arrow/.style={-latex, thick, blue}
]

\node[node2] (i1) at (-3,0) {$u$};
\node[node] (b1) at (-1.5,0) {$1$\\in $2$};
\node[node] (b2) at (0,0) {$1$\\in $2$};
\node[node] (b3) at (1.5,0) {$1$\\in $2$};
\node[node2] (i2) at (3,0) {$v$};

\foreach \i in {1,2,3} {
    \draw[arrow] ([xshift=-0.2cm]b\i.south) -- ([xshift=-0.2cm,yshift=-0.3cm]b\i.south);
    \draw[arrow] ([xshift=0.2cm,yshift=-0.27cm]b\i.south) -- ([xshift=0.2cm,yshift=0.03cm]b\i.south);
}

\draw[edge] (i1) -- (b1);
\draw[edge] (b1) -- (b2);
\draw[edge] (b2) -- (b3);
\draw[edge] (b3) -- (i2);

\node at (0,1.1) {$f$};
\draw [decorate,decoration={brace,amplitude=6pt}] (-2.2,0.6) -- (2.2,0.6);

\end{tikzpicture}
\subcaption{Subdividing an edge $(u,v)$ with $1$-in-$2$ vertices.}
\label{fig:sim-1-in-2}
\end{subfigure}

\begin{subfigure}{\linewidth}
\centering
\begin{tikzpicture}[
  node2/.style={circle, draw, thick, circle, draw, minimum size=5mm, inner sep=0pt},
  node/.style={rectangle, draw, thick, align=center, minimum width=10mm, minimum height=10mm, inner sep=0pt},
  sync/.style={rectangle, draw, thick, minimum width=6mm, minimum height=6mm, inner sep=0pt},
  edge/.style={-latex, thick},
  arrow/.style={-latex, thick, blue}
]

\node[node] (f1) at (-1.5,3) {$\{0,1\}$\\in $2$};
\node[node] (f2) at (0,3) {$\{0,1\}$\\in $2$};
\node[node] (f3) at (1.5,3) {$\{0,1\}$\\in $2$};

\node[node] (m1) at (-1.5,1.5) {$\{0,1\}$\\in $2$};
\node[node] (m2) at (0,1.5) {$\{0,1\}$\\in $2$};
\node[node] (m3) at (1.5,1.5) {$\{0,1\}$\\in $2$};

\node[node] (b1) at (-1.5,0) {$\{0,1\}$\\in $2$};
\node[node] (b2) at (0,0) {$\{0,1\}$\\in $2$};
\node[node] (b3) at (1.5,0) {$\{0,1\}$\\in $2$};

\foreach \i in {1,2,3} {
    \foreach \n in {f,m,b} {
        \draw[arrow] ([xshift=-0.2cm]\n\i.south) -- ([xshift=-0.2cm,yshift=-0.3cm]\n\i.south);
        \draw[arrow] ([xshift=0.2cm,yshift=-0.27cm]\n\i.south) -- ([xshift=0.2cm,yshift=0.03cm]\n\i.south);
    }
}

\node[node2] (i1) at (-4,0) {$u$};
\node[sync] (s1) at (-3,0.75) {$A$};
\node[node2] (i2) at (4,0) {$v$};
\node[sync] (s2) at (3,0.75) {$A$};

\draw[edge] (i1) -| ([xshift=-0.13cm]s1.south);
\draw[edge] ([xshift=-0.13cm]s1.north) |- (f1);
\draw[edge] (f1) -- (f2);
\draw[edge] (f2) -- (f3);
\draw[edge] (f3) -| ([xshift=0.13cm]s2.north);
\draw[edge] ([xshift=-0.13cm]s2.north) |- (m3);
\draw[edge] (m3) -- (m2);
\draw[edge] (m2) -- (m1);
\draw[edge] (m1) -| ([xshift=0.13cm]s1.north);
\draw[edge] ([xshift=0.13cm]s1.south) |- (b1);
\draw[edge] (b1) -- (b2);
\draw[edge] (b2) -- (b3);
\draw[edge] (b3) -| ([xshift=-0.13cm]s2.south);
\draw[edge] ([xshift=0.13cm]s2.south) |- (i2);

\node at (0,4.1) {$f$};
\draw [decorate,decoration={brace,amplitude=6pt}] (-2.2,3.6) -- (2.2,3.6);
\end{tikzpicture}

\subcaption{Subdividing an edge $(u,v)$ with $\{0,1\}$-in-$2$ vertices and two alternators ($A$). The two loops ensure that each $\{0,1\}$-in-$2$ has exactly one incoming edge.}
\label{fig:sim-01-in-2}
\end{subfigure}
\caption{Replace each edge $(u,v)$ with the gadgets as shown. Blue arrows denote constants fed into the $S$-in-$j$ to make it simulate $1$-in-$2$ or $\{0,1\}$-in-$2$ respectively.}
\end{figure}
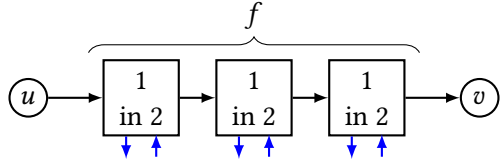
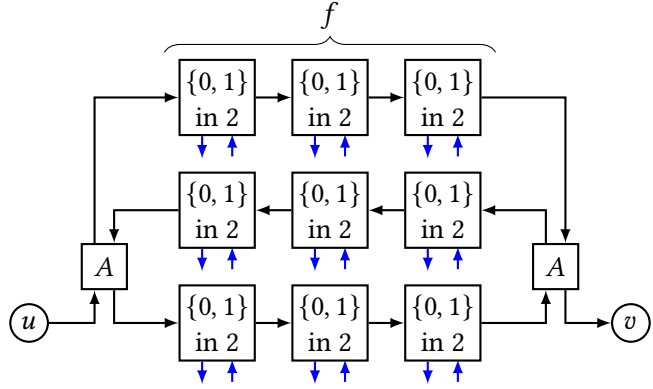

Now we consider the case when we only have either a positive or a negative terminator. We get two new polynomial cases.

\begin{definition}
    An $S$-in-$j$ vertex type is \defn{top-down closed} if $\forall i > j/2$, $i \in S \implies i-1 \in S$, that is, ``the top half is downward closed.''
    
    An $S$-in-$j$ vertex type is \defn{bottom-up closed} if $\forall i < j/2$, $i \in S \implies i+1 \in S$, that is, ``the bottom half is upward closed.''
\end{definition}

\begin{lemma}\label{lem:top-down-easy}
    Graph Orientation (planar or non-planar) with only top-down closed vertices, or only bottom-up closed vertices, is in P.
\end{lemma}
\begin{proof}
    We prove the top-down case; the bottom-up case is symmetric. We temporarily allow edges to be undirected (contributing in-flow to neither of its endpoints) or bidirected (contributing in-flow to both of its endpoints). First, arbitrarily assign edges around every $S$-in-$j$ vertex so that exactly $\max S$ edges point in. Then, we fix the undirected and bidirected edges.
    
    We first fix the undirected edges iteratively. Suppose there is an undirected edge. Perform a breadth-first search from all undirected edges in the graph formed by the reverse of the current orientation until we reach a bidirected edge or the search terminates. In the first case, we find a path $u_1 - u_2 \leftarrow \dots \leftarrow u_{k-1} \leftrightarrow u_k$. Reverse this path into $u_1 \rightarrow u_2 \rightarrow \dots \rightarrow u_{k-1} \rightarrow u_k$. This does not change any in-degrees and fixes one undirected and one bidirected edge.

    If instead the search terminates, we find a cut $V = A \sqcup B$ with only undirected and singly directed edges in $A$, and all edges between $A$ and $B$ are singly directed towards $B$. It follows that the instance is unsatisfiable because the total in-flow within $A$ is already the maximum possible but does not use up all the edges.

    Repeat the above until the instance is deemed unsatisfiable or there are no undirected edges left. Then, we claim the instance must be satisfiable. We can fix the bidirected edges one by one. Take a bidirected edge $u_1 \leftrightarrow u_2$, and take a maximal walk $u_1,u_2,\dots,u_k$ from it that does not revisit any edges. The last vertex $u_k$ must have an in-degree strictly larger than out-degree. Let $u_t \leftrightarrow u_{t+1}$ be the last bidirected edge on this walk, so we have a subwalk $u_t \leftrightarrow u_{t+1} \rightarrow \dots \rightarrow u_k$. Reverse this walk into $u_t \leftarrow u_{t+1} \leftarrow \dots \leftarrow u_k$. This fixes one bidirected edge. It also reduces the in-degree of $u_k$ by one, but because $u_k$ is top-down closed, the new in-degree is also valid.
\end{proof}

\begin{lemma}\label{lem:sim-other-terminator}
    Let $\Gamma$ be a set of $S$-in-$j$ vertex types which simulate a positive terminator. Then either all vertex types in $\Gamma$ are bottom-up closed (and $\Gamma$-GO is in P), or $\Gamma$ (planarly) simulates a negative terminator.
\end{lemma}
\begin{proof}
    If all vertices are bottom-up closed, $\Gamma$-GO is in P by Lemma~\ref{lem:top-down-easy}. Otherwise, $\Gamma$ has an $S$-in-$j$ that is not bottom-up closed ($\exists i < j/2 \colon i \in S, i+1 \not\in S$). Add $i$ self-loops to this vertex so that, without loss of generality, $0 \in S, 1 \not\in S$.
    
    Let the positive terminator we have be an $f$-terminator ($f > 0$). Take $f$ copies of this $S$-in-$j$ and connect $j-1$ edges of each one to a $(j-1)f$-terminator. Each $S$-in-$j$ now has $j-1$ outgoing edges, so the last edge must also point out, obtaining a $-f$-terminator.
\end{proof}

In summary, we obtain the following theorem:
\begin{theorem}\label{thm:S-in-j-term}
    Let $\Gamma$ be a set of $S$-in-$j$ vertex types that can simulate either a positive terminator or a negative terminator. Then we have one of the following cases:
    \begin{enumerate}
        \item all vertex types in $\Gamma$ are bottom-up closed, and $\Gamma$-GO is in P; or
        \item all vertex types in $\Gamma$ are top-down closed, and $\Gamma$-GO is in P; or
        \item (Planar) $\Gamma$-GO is equivalent to (Planar) $\Gamma$-GO$_c$, and the dichotomy in Theorem~\ref{thm:planar-S-in-j-and-const} (in the planar setting) or Theorem~\ref{thm:S-in-j-and-const} (in the non-planar setting) may be applied.
    \end{enumerate}
\end{theorem}
\begin{proof}
    In the first two cases, $\Gamma$-GO is in P by Lemma~\ref{lem:top-down-easy}. Otherwise, we can simulate a terminator of the opposite type by Lemma~\ref{lem:sim-other-terminator}, and apply Lemma~\ref{lem:nonlocally-const}.
\end{proof}

Being able to simulate some terminator is a pretty weak condition. If for some $S$-in-$j$ vertex type, $S$ has a unique closest element to $j/2$ (as opposed to there being a tie between $j/2 - a$ and $j/2 + a$), then we can simulate a terminator simply by adding self-loops.

Unfortunately, there are still some cases where we cannot simulate any terminator, and the problem remains open. For example, we do not know whether Graph Orientation with $\{1,4\}$-in-$8$ and $\{0,2\}$-in-$2$ (which is equivalent to $\{1,4\}$-in-$8$SAT-E$2$) is in P or NP-hard. A complete dichotomy for Symmetric Graph Orientation would imply having a complete dichotomy for Symmetric $\Gamma$-SAT-E$2$ (without constants), which we do not have.

\subsection{$i$-in-$j$ Dichotomy}\label{subsec:i-in-j}

One case where we \emph{can} get a complete dichotomy without constants is if sets $S$ are singletons, i.e., $\Gamma$ is a set of $i$-in-$j$ vertex types. Graph Orientation with only $\Gamma$ is in P by Lemma~\ref{lem:classic-easy-cases}, so to make the problem not trivial, we also allow duplicators. By Corollary~\ref{cor:dups-equiv-dup}, without loss of generality, we only have a single non-trivial duplicator $D$. We want to characterize the complexity of $(\Gamma \cup \{D\})$-GO and Planar $(\Gamma \cup \{D\})$-GO.

This case is particularly useful for tiling problems; see Section~\ref{sec:tiling}. Any subregion that needs to be tiled with $k$-ominoes needs to have an area that is a multiple of $k$. This naturally leads to graph orientation vertices whose net flow is constant modulo $k$, such as $i$-in-$k$, $k$-equalizers, and synchronizers.

In the following subsections, we consider the cases where $D$ is an alternator, synchronizer, or neither (i.e., has a non-zero net flow) and prove a complete dichotomy in each case (Theorems~\ref{thm:i-in-j-and-dup},~\ref{thm:i-in-j-and-sync},~\ref{thm:i-in-j-and-alt}). Surprisingly, if $D$ is an alternator, Planar Graph Orientation turns out to always be in P. These results together fully characterize the complexity of Graph Orientation and Planar Graph Orientation for any set of $i$-in-$j$ vertex types and non-trivial duplicators.

\subsubsection{$i$-in-$j$ and Nonzero Net Flow Duplicators}

We first consider when $D$ is neither a synchronizer nor an alternator, i.e., it has a nonzero net flow. The main idea is that unless $i = j/2$, we can add self-loops to get a terminator, and then apply Theorem~\ref{thm:S-in-j-term}. If instead $i = j/2$, we can make a $2$-in-$4$ vertex.

\begin{theorem}\label{thm:i-in-j-and-dup}
    Let $\Gamma$ be a set of $i$-in-$j$ vertex types, and $D$ be a non-trivial duplicator with net flow $\pm f$ ($f > 0$). (Planar) Graph Orientation with $\Gamma$ and $D$ is NP-complete, except for the following cases, which are in P:
    \begin{enumerate}
        \item all vertex types in $\Gamma$ are bijunctive ($0$-in-$j$, $j$-in-$j$, or $1$-in-$2$);
        \item (Planar only) $f \ge 5$ and $i \in \{0, 1, j\}$ for all $i$-in-$j$ in $\Gamma$;
        \item (Planar only) $f \ge 5$ and $i \in \{0, j-1, j\}$ for all $i$-in-$j$ in $\Gamma$.
    \end{enumerate}
\end{theorem}

\begin{proof}
    The first case is in P by Lemma~\ref{lem:classic-easy-cases} (duplicators are bijunctive).
    
    Suppose that we have an $i$-in-$j$ vertex type that is not $0$-in-$j$, $j$-in-$j$, or $1$-in-$2$. If $i = j/2$, add self-loops to it to get $2$-in-$4$, which is hard by Lemma~\ref{lem:base-hard-cases} (note that $j$ has to be at least $4$, because if it were $2$ the vertex type would be bijunctive). Otherwise, WLOG $i > j/2$ and we can add self-loops to it to get a positive terminator.

    To apply Theorem~\ref{thm:S-in-j-term}, we want all our vertex types to be symmetric, but $D$ might not be. We could try replacing $D$ with an $f$-equalizer, which is symmetric. However, this fails for $f = 1$ and $f = 2$, because an $f$-equalizer would be trivial but $D$ is not. Instead, we replace $D$ with a pair of a $2f$-equalizer and a $3f$-equalizer. Let $\Gamma' := \Gamma \cup \{2f\text{-equalizer}, 3f\text{-equalizer}\}$. By Theorem~\ref{thm:dup-characterization}, this pair of vertex types is equivalent to $D$, so our problem is equivalent to (Planar) $\Gamma'$-GO. Also note that $\max \Delta \Gamma' = \max\{2f,3f\} \ge 3$ and $\gcd \Delta \Gamma' = \gcd\{2f,3f\} = f$.

    Now that every vertex type is symmetric, we can apply Theorem~\ref{thm:S-in-j-term}. Because $3f$-equalizer is neither top-down closed nor bottom-up closed, we can then apply Theorem~\ref{thm:S-in-j-and-const} (in the non-planar setting) or Theorem~\ref{thm:planar-S-in-j-and-const} (in the planar setting).

    Not all relations are bijunctive, not all are affine (note that an $i$-in-$j$ relation is affine if and only if it is bijunctive), and $\max \Delta \Gamma' \ge 3$. So in the non-planar setting, we get NP-hardness unconditionally. In the planar setting, the remaining two easy cases of Theorem~\ref{thm:planar-S-in-j-and-const} correspond exactly to easy cases 2 and 3 of this theorem.
\end{proof}

\subsubsection{$i$-in-$j$ and Synchronizers}

Next, we consider the case when $D$ is a synchronizer. Unlike the previous case, we cannot replace $D$ with symmetric vertex types and apply Theorem~\ref{thm:S-in-j-and-const}, because synchronizers are fundamentally not symmetric.

\begin{theorem}\label{thm:i-in-j-and-sync}
    Let $\Gamma$ be a set of $i$-in-$j$ vertex types. (Planar) Graph Orientation with $\Gamma$ and a synchronizer is NP-complete, except for the following cases, which are in P:
    \begin{enumerate}
        \item all vertex types in $\Gamma$ are bijunctive ($0$-in-$j$, $j$-in-$j$, or $1$-in-$2$);
        \item $i < j/2$ or $i = 1, j = 2$ for all $i$-in-$j$ in $\Gamma$;
        \item $i > j/2$ or $i = 1, j = 2$ for all $i$-in-$j$ in $\Gamma$;
        \item $i \in \{0, 1, j\}$ for all $i$-in-$j$ in $\Gamma$;
        \item $i \in \{0, j-1, j\}$ for all $i$-in-$j$ in $\Gamma$.
    \end{enumerate}
\end{theorem}

\begin{proof}
    The first case is in P by Lemma~\ref{lem:classic-easy-cases}. In the second case, the in-degree of each vertex must be less than or equal to the out-degree (with equality only for duplicators). If there are any $i$-in-$j$ vertices for $j > 2$, the instance is unsatisfiable. Otherwise, all vertices are bijunctive and we can solve the problem in polynomial time. Similarly for the third case. In the fourth and fifth cases, we can first eliminate all constants and then reduce to the second or third case. Now assume we do not have any of the easy cases.

    Consider some $i$-in-$j$ such that $i \ge j/2$ and $(i,j) \ne (1,2)$. Add self-loops until we get $2$-in-$4$ or a positive terminator. Similarly, we can get $2$-in-$4$ or a negative terminator.

    If we got $2$-in-$4$, the problem is hard by Lemma~\ref{lem:base-hard-cases}. Otherwise, we have both a positive and a negative terminator. We also have a non-bijunctive symmetric vertex type and a synchronizer. Then, by Lemma~\ref{lem:nonlocally-const}, we can assume that we also have constants.

    Consider some $i$-in-$j$ vertex type such that $i \notin \{0, 1, j\}$. Then its bitstring contains \texttt{0010}, and we can simulate $2$-in-$3$. Similarly, we can simulate $1$-in-$3$, which gives us NP-hardness by Lemma~\ref{lem:base-hard-cases}.
\end{proof}

\subsubsection{$i$-in-$j$ and Alternators}

Finally, we consider the case when $D$ is an alternator. This is only meaningful in the planar setting, because in the non-planar setting alternators are the same as synchronizers.

We show that Planar Graph Orientation with $i$-in-$j$ and alternators is always in P, in surprising contrast with the previous subsections. We first prove a lemma whose construction will be used in the proof.

\begin{lemma}\label{lem:alternators-alwayssat}
    Planar Graph Orientation with only alternators and $(j/2)$-in-$j$ is always satisfiable.
\end{lemma}
\begin{proof}
    Every vertex has even degree, so the dual graph is bipartite. Take a 2-coloring of the faces, say with black and white, and orient every edge so that the black face is on its left. The edges around a vertex alternate as required.
\end{proof}

\begin{theorem}\label{thm:i-in-j-and-alt}
    Let $\Gamma$ be a set of $i$-in-$j$ vertex types. Planar Graph Orientation with $\Gamma$ and alternators is in P.
\end{theorem}
\begin{proof}
    \textbf{Algorithm.} We use linear programming. For every vertex $u$ and neighbor $v$, create a variable $x_{uv}$. The meaning of this variable is that $x_{uv} = 1$ if $u$ is directed towards $v$ and $x_{uv} = -1$ if $v$ is directed towards $u$ in a satisfying orientation. The constraints are:
    \begin{itemize}
        \item For each edge $(u,v)$, $x_{uv} = -x_{vu}$ and $-1 \le x_{uv} \le 1$.
        \item For a $i$-in-$j$ vertex $u$ with neighbors $v_1,\dots,v_j$, $x_{uv_1}+x_{uv_2}+\dots+x_{uv_j} = (j-i)-i$.
        \item For a $2k$-alternator vertex $u$ with neighbors $v_1,\dots,v_{2k}$ in cyclic order, $x_{uv_1} = -x_{uv_2} = x_{uv_3} = \dots = -x_{uv_{2k}}$.
    \end{itemize}
    The key claim is that this linear program is feasible if and only if there exists a satisfying graph orientation. The feasibility of this linear program can be checked in polynomial time.

    \textbf{Proof of Correctness.} Any satisfying orientation corresponds to a solution with $x_{uv} \in \{-1,1\}$ for all edges $(u,v)$, so if the linear program is infeasible, there is no satisfying orientation. Otherwise, we must prove that if there exists a feasible solution, then there exists a solution with $x_{uv} \in \{-1,1\}$ for all edges $(u,v)$. Call an edge $(u,v)$ \emph{slack} if $-1 < x_{uv} < 1$.

    Let $\{x_{uv}\}$ be a solution with at least one slack edge. We give a procedure that modifies the solution to decrease the number of slack edges. Correctness follows by repeating this procedure.

    Consider the dual graph $G^*$. For an edge $(u,v)$, let face $u^*$ be to the left of $\overrightarrow{uv}$ and face $v^*$ be to the right of $\overrightarrow{uv}$ (technically ambiguous notation but it will be clear from context). For each edge $(u,v)$, draw an edge from $u^*$ to $v^*$ in $G^*$ weighted by $x_{uv}$. The key observation is:

    \textit{Observation.} Every cycle in $G^*$ has weights summing to $0 \pmod 1$.

    \textit{Proof of Observation.} It suffices to consider the faces of $G^*$ because faces form a cycle basis. The sum of the weights along a face of $G^*$ corresponds to the net flow into a vertex of the primal graph, which is an integer as required.

    This implies that there exists a potential function $\pi$ mapping faces to elements of $\mathbb{R}/\mathbb{Z}$ such that for each edge $(u,v)$, $\pi_{v^*}-\pi_{u^*} \equiv x_{uv} \pmod 1$. By construction, if $f_1,\dots,f_{2k}$ are the faces around an alternator vertex in cyclic order, then we also have $\pi_{f_1} = \pi_{f_3} = \dots = \pi_{f_{2k-1}}$ and $\pi_{f_2} = \pi_{f_4} = \dots = \pi_{f_{2k}}$.

    Because $\{x_{uv}\}$ has a slack edge, either there are at least two distinct values in the range of $\pi$ or $x_{uv} \in \{-1,0,1\}$ for all edges $(u,v)$. In the latter case, consider the subgraph formed by edges $(u,v)$ with $x_{uv} = 0$. Replace these values with a satisfying orientation given by Lemma~\ref{lem:alternators-alwayssat}.

    Otherwise, let $\pi_0$ be any value in the range of $\pi$. Define a nonzero vector $\{y_{uv}\}$ by $y_{uv} = [\pi_{v^*}=\pi_0] - [\pi_{u^*}=\pi_0]$ where $[\cdot]$ is the Iverson bracket. For all $\epsilon \in \mathbb{R}$, note that $x+\epsilon y$ satisfies the flow constraints at every vertex. Pick $\epsilon$ so that $x+\epsilon y$ hits the boundary of one of the $-1 \le x_{uv}+\epsilon y_{uv} \le 1$ constraints. Note that $y_{uv} = 0$ for all non-slack edges $(u,v)$, so replacing $x$ with $x+\epsilon y$ strictly decreases the number of slack edges as required.
\end{proof}

\section{KPlumber}
\label{sec:kplumber}

KPlumber is a single-player computer game included in some standard Linux distributions. The game is played on a rectangular grid of cells, where each cell contains a rotatable tile that depicts some segments of a pipe. The possible tiles are shown in Figure~\ref{fig:kplumber}. Each of the four edges of a tile is either empty or contains an unterminated pipe end. The goal of the game is to rotate the tiles so that all the unterminated ends of the pipes are paired along interior edges, forming a closed pipe system. The system may contain one or more connected components of pipes. From a complexity standpoint, the natural decision problem is whether, given a grid of tile types, there exists a way to rotate the tiles to solve the puzzle.

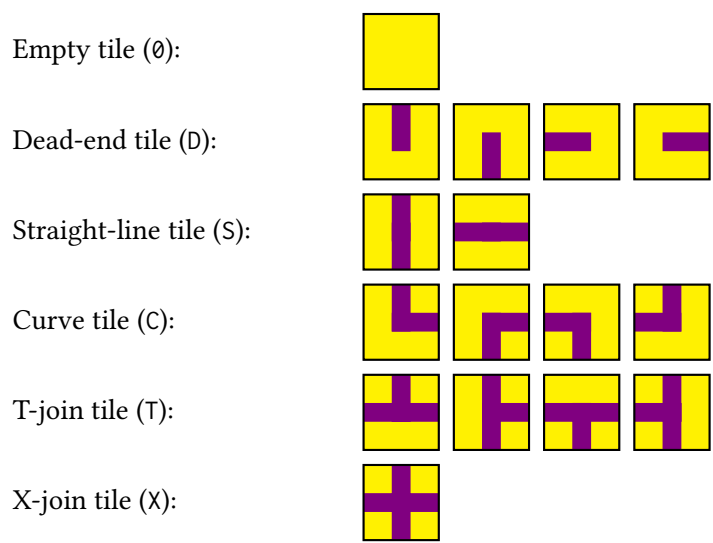
\begin{figure}[h!]
\centering
\begin{tikzpicture}[x=1cm, y=1cm, scale=1.0, transform shape]
\def\pathup{\draw[line width=0.25cm, color=violet] (0.5,0.375) -- (0.5,1);}
\def\pathdown{\draw[line width=0.25cm, color=violet] (0.5,0.625) -- (0.5,0);}
\def\pathleft{\draw[line width=0.25cm, color=violet] (0.625,0.5) -- (0,0.5);}
\def\pathright{\draw[line width=0.25cm, color=violet] (0.375,0.5) -- (1,0.5);}

\newcommand{\tile}[2]{
  \begin{scope}[shift={(#1)}]
    \draw[fill=yellow] (0,0) rectangle (1,1);
    #2
    \draw[thick] (0,0) rectangle (1,1);
  \end{scope}
}

\node[anchor=west] at (-2.8,0.5) {Empty tile (\texttt{0}):};
\node[anchor=west] at (-2.8,-0.7) {Dead-end tile (\texttt{D}):};
\node[anchor=west] at (-2.8,-1.9) {Straight-line tile (\texttt{S}):};
\node[anchor=west] at (-2.8,-3.1) {Curve tile (\texttt{C}):};
\node[anchor=west] at (-2.8,-4.3) {T-join tile (\texttt{T}):};
\node[anchor=west] at (-2.8,-5.5) {X-join tile (\texttt{X}):};

\tile{(2,0)}{}

\tile{(2,-1.2)}{\pathup}
\tile{(3.2,-1.2)}{\pathdown}
\tile{(4.4,-1.2)}{\pathleft}
\tile{(5.6,-1.2)}{\pathright}

\tile{(2,-2.4)}{\pathup \pathdown}
\tile{(3.2,-2.4)}{\pathleft \pathright}

\tile{(2,-3.6)}{\pathup \pathright}
\tile{(3.2,-3.6)}{\pathright \pathdown}
\tile{(4.4,-3.6)}{\pathdown \pathleft}
\tile{(5.6,-3.6)}{\pathleft \pathup}

\tile{(2,-4.8)}{\pathup \pathleft \pathright}
\tile{(3.2,-4.8)}{\pathright \pathdown \pathup}
\tile{(4.4,-4.8)}{\pathdown \pathleft \pathright}
\tile{(5.6,-4.8)}{\pathleft \pathup \pathdown}

\tile{(2,-6.0)}{\pathup \pathdown \pathleft \pathright}

\end{tikzpicture}
\caption{The six possible tile types in KPlumber and all their rotations.
  Based on \cite[Figure~1]{kplumber}.
}
\label{fig:kplumber}
\end{figure}

In 2004, Král et al.~\cite{kplumber} studied the complexity of KPlumber when the input is restricted to various subsets of the tiles in Figure~\ref{fig:kplumber}. They provided complexity classifications (either P or NP-complete) for most subsets of tiles, but there was one family of cases left unresolved. All cases in this family were shown to be polynomially equivalent to KPlumber with \texttt{0}, \texttt{D}, \texttt{S}, \texttt{T}, and \texttt{X} tiles (i.e., only the curve tile \texttt{C} is excluded). Here, we resolve this open question by showing that this family is (perhaps surprisingly) in P, completing the classification.

\begin{definition}
    \defn{Planar Bipartite SAT-E$2$} is a family of SAT problems where each variable appears in exactly two clauses, and the graph where clauses are vertices and variables are edges (connecting the two clauses each variable appears in) is planar and bipartite.

    In the context of Planar Bipartite SAT-E$2$, an \defn{$i$-in-$j$ clause} is satisfied if and only if exactly $i$ of its $j$ variables are true. An \defn{alternator clause} is satisfied if and only if its variables alternate between true and false in the planar embedding.
\end{definition}

\begin{corollary}\label{cor:sat-e2-easy}
    Planar Bipartite SAT-E$2$ with $i$-in-$j$ clauses and alternators is in P.
\end{corollary}
\begin{proof}
    $2$-color the clauses with black and white. Direct every true edge from black to white, and every false edge from white to black. White $i$-in-$j$ SAT clauses become $i$-in-$j$ GO vertices, black $i$-in-$j$ SAT clauses become $(j-i)$-in-$j$ GO vertices, and alternator clauses become alternator GO vertices, so the problem is equivalent to Theorem~\ref{thm:i-in-j-and-alt}.
\end{proof}

\begin{corollary}\label{cor:kplumber-easy}
    KPlumber without the curve tile \texttt{C} is in P.%
\end{corollary}
\begin{proof}
    Create a SAT-E$2$ clause for every cell in the grid and a variable for each pair of edge-adjacent cells, representing whether pipes are matched along that edge or not. \texttt{O}, \texttt{D}, \texttt{T}, \texttt{X} tiles correspond to $i$-in-$4$ clauses for $i = 0,1,3,4$ respectively, and \texttt{S} tiles are alternator clauses. The resulting grid graph is both planar and bipartite, so Corollary~\ref{cor:sat-e2-easy} applies.
\end{proof}

\section{Tiling}
\label{sec:tiling}

In this section, we apply our framework to tiling problems. A tiling problem is specified by a set of available tiles $S$ and a grid, where each cell of the grid is filled or unfilled. The goal is to place copies of the available tiles in $S$ such that no unfilled cell is covered and every filled cell is covered by exactly one tile. Rotations of the tiles are permitted but reflections are not (tiling with reflection can be represented by having $S$ include reflected versions of the tiles). Previous work by Horiyama et al.~\cite{tromino-tiling} has proved that tiling is NP-complete with only $\L$ trominoes, only $\I$ trominoes, or both $\L$ trominoes and $\I$ trominoes.

\subsection{Tromino Tiling}
\label{sec:tromino-tiling}

As an illustration of the general approach, we first reprove Horiyama et al.~\cite{tromino-tiling}'s result that tiling with $\L$ trominoes is NP-complete. The original reduction of Horiyama et al.\ needed \emph{six} types of gadgets which they named the \emph{line}, \emph{corner}, \emph{cross}, \emph{duplicator}, \emph{clause}, and \emph{negated-clause} gadgets. The former two gadgets simulate edges of a planar graph orientation instance, while the latter four gadgets, in our terminology, are equivalent to the \emph{crossover}, \emph{synchronizer}, \emph{1-in-3}, and \emph{2-in-3} respectively. Applying our framework, our reduction requires only two vertex gadgets and is thus simpler.

\begin{theorem}\label{thm:L-tromino}
    Tiling with $\L$ trominoes is NP-complete.
\end{theorem}
\begin{proof}
\begin{figure}[h]
    \centering
    \includegraphics[width=0.21\textwidth]{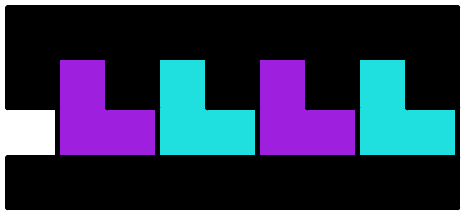}
    \includegraphics[width=0.21\textwidth]{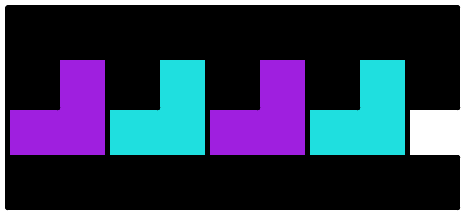}
    \hspace{1cm}
    \includegraphics[width=0.21\textwidth]{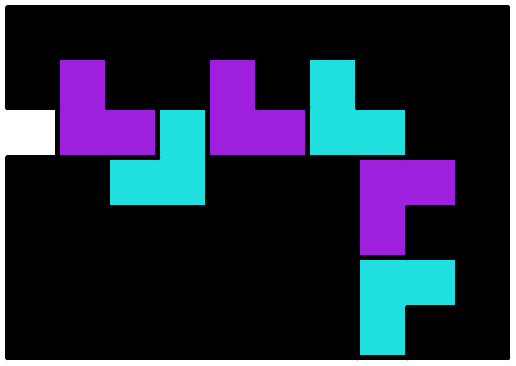}
    \includegraphics[width=0.21\textwidth]{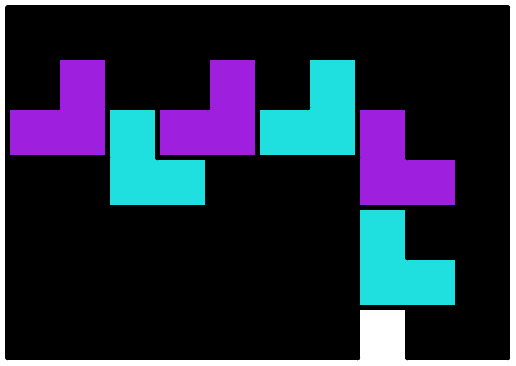}
    \caption{Left: Wire gadget for $\L$ trominoes. Right: Shift/turn gadget for $\L$ trominoes.}
    \label{fig:tromino-L-wire}
\end{figure}
    Reduce from Planar $\{\text{1-in-3},\text{3-equalizer}\}$-GO, which is NP-hard by Theorem~\ref{thm:i-in-j-and-dup}. Take a graph orientation instance and embed its vertices in the grid such that all the vertices are sufficiently far apart (a polynomial separation suffices).
    
    We represent an edge of the graph with a series of \defn{wire} and \defn{shift} gadgets, as shown in Figure~\ref{fig:tromino-L-wire}. The white cells of different wire and shift gadgets are shared between gadgets and propagate a Boolean signal from one area to another through the two possible tilings of the interior of these gadgets. We say this edge is oriented towards the direction that ``covers'' the unconnected white cell at the boundary.

    The wire gadget, using an arbitrary number of $\L$ trominoes in the middle, can move the white cell by any multiple of $2$ units horizontally. By turning in the same way as shown in the shift gadget, the white cell can be moved by any multiple of $2$ units in any direction. The leftmost blue tromino of the shift gadget allows us to shift the wire by $1$ unit. Therefore, it is possible to join any two (sufficiently far apart) cells of the grid with an edge composed of wire and shift gadgets.

\begin{figure}[h]
    \centering
    \includegraphics[width=0.16\textwidth]{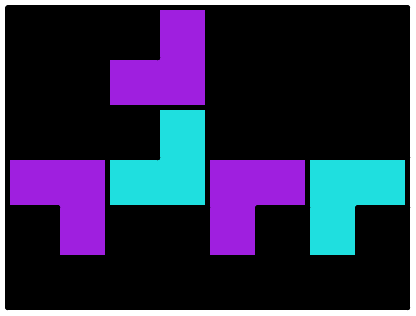}
    \includegraphics[width=0.16\textwidth]{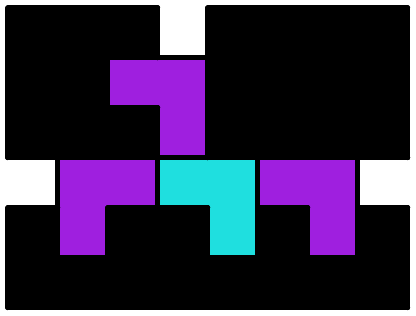}
    \hspace{1cm}
    \includegraphics[width=0.16\textwidth]{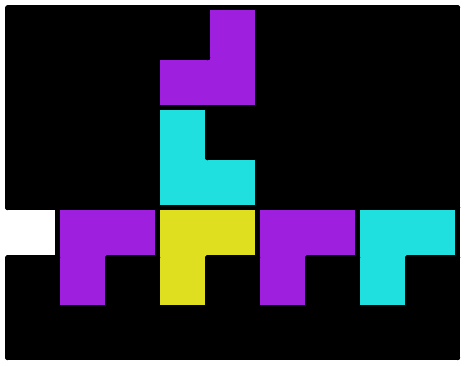}
    \includegraphics[width=0.16\textwidth]{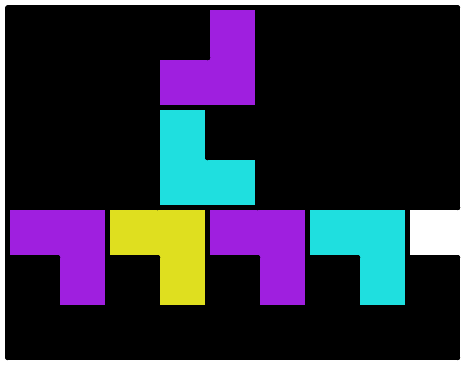}
    \includegraphics[width=0.16\textwidth]{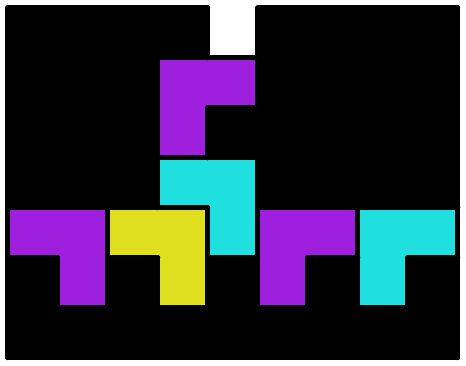}
    \caption{Left: 3-equalizer gadget for $\L$ trominoes. Right: 1-in-3 gadget for $\L$ trominoes.}
    \label{fig:tromino-L-vertex}
\end{figure}

    We now describe the vertex gadgets, which are shown in Figure~\ref{fig:tromino-L-vertex}. The left side shows a 3-equalizer gadget and the two ways to tile it without leaving any holes in its interior. Similarly, the right side shows a 1-in-3 gadget and the three ways to tile it. Both gadgets contain $3$ gadget ports which have the option of being covered by the gadget or by an incoming edge, constraining the possibilities of the incoming edges to be precisely those of the 3-equalizer or 1-in-3 respectively. One easy way to verify that no undesirable partial tilings (i.e., other assignments of the neighboring edges) exist is to observe that any partial tiling must cover a multiple of $3$ cells.

    Therefore, by replacing each vertex with 3-equalizer or 1-in-3 gadgets and wiring the edges of the graph, we can embed a Planar $\{\text{1-in-3},\text{3-equalizer}\}$-GO instance into $\L$ tromino tiling, so $\L$ tromino tiling is NP-hard.
\end{proof}

\begin{remark}
The reduction in Theorem~\ref{thm:L-tromino} still works for tiling with both $\L$ trominoes and $\I$ trominoes (reproving another theorem in \cite{tromino-tiling}) because it is impossible to fit any $\I$ trominoes into any of the gadgets without creating isolated cells.
\end{remark}

For $\I$ tromino tiling, one cannot easily reduce from Planar $\{\text{1-in-3},\text{3-equalizer}\}$-GO in the same way. The reason is that there is a ``$\bmod\,3$'' issue underlying the problem: If we color cells $(x,y)$ of the grid by $(x+y) \bmod 3$, any $\I$ tromino covers one cell of each color. Therefore any wire gadget (that transports a one-cell-wide signal) can only connect between cells of the same color. On the other hand, any 3-equalizer gadget can only produce three gadget ports of different colors. It is not possible to connect arbitrary vertices with wires.

Instead, the reduction of Horiyama et al.~\cite{tromino-tiling}, rephrased in our framework, uses \emph{crossover}, \emph{synchronizer}, \emph{1-in-3}, and \emph{2-in-3} gadgets, taking care to ensure that all gadget ports have the same $(x \bmod 3,y \bmod 3)$ (so they can be connected by wires), including a highly nontrivial crossover construction. The gadgets used by Horiyama et al.\ are shown in Figures~\ref{fig:tromino-I-wire},~\ref{fig:tromino-I-turn},~\ref{fig:tromino-I-synch},~\ref{fig:tromino-I-1in3},~\ref{fig:tromino-I-2in3},~\ref{fig:tromino-I-crossover}. The crossover is unnecessary to establish NP-completeness according to Theorem~\ref{thm:i-in-j-and-dup}.

\begin{figure}
\centering
\includegraphics[width=0.3\textwidth]{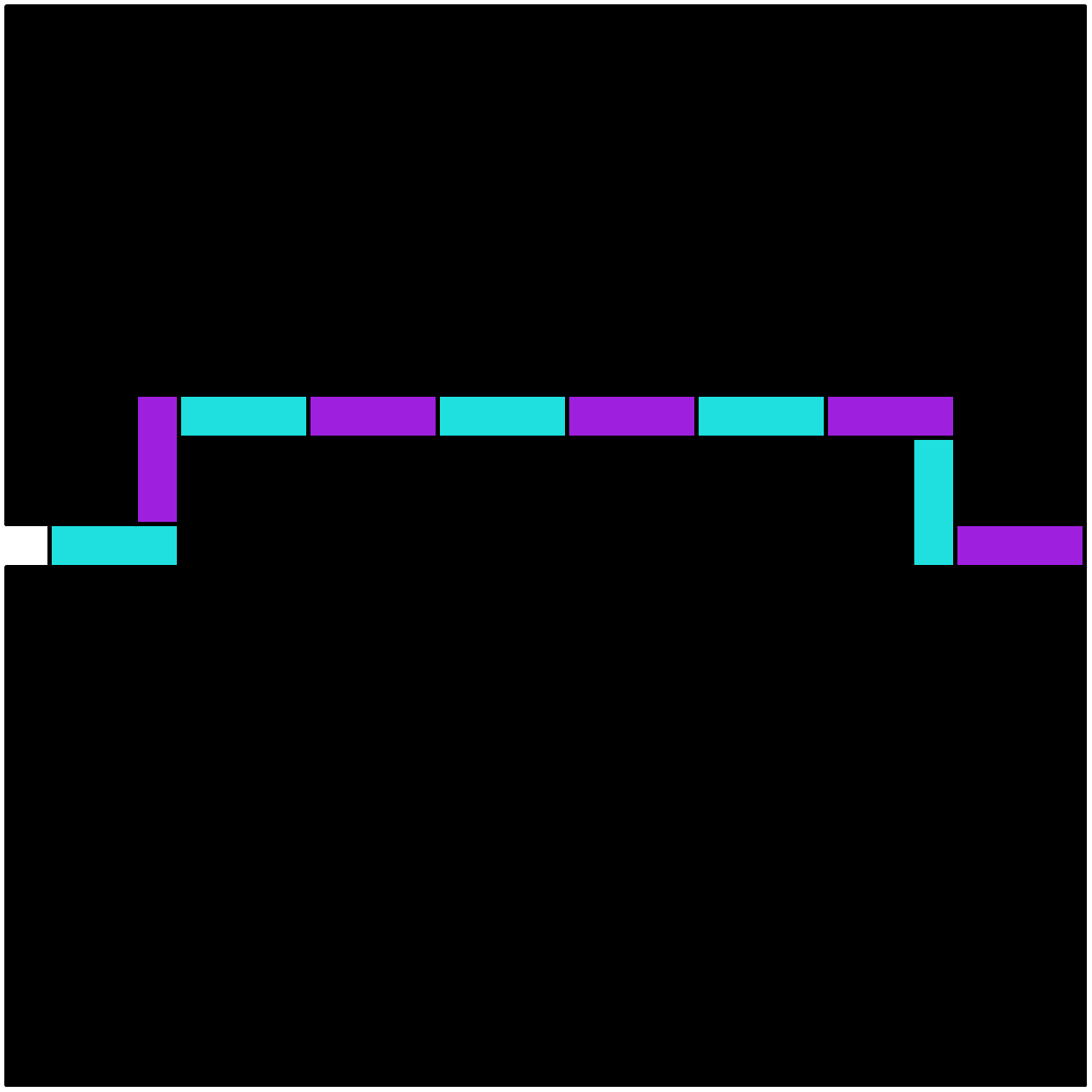}
\hspace{1cm}
\includegraphics[width=0.3\textwidth]{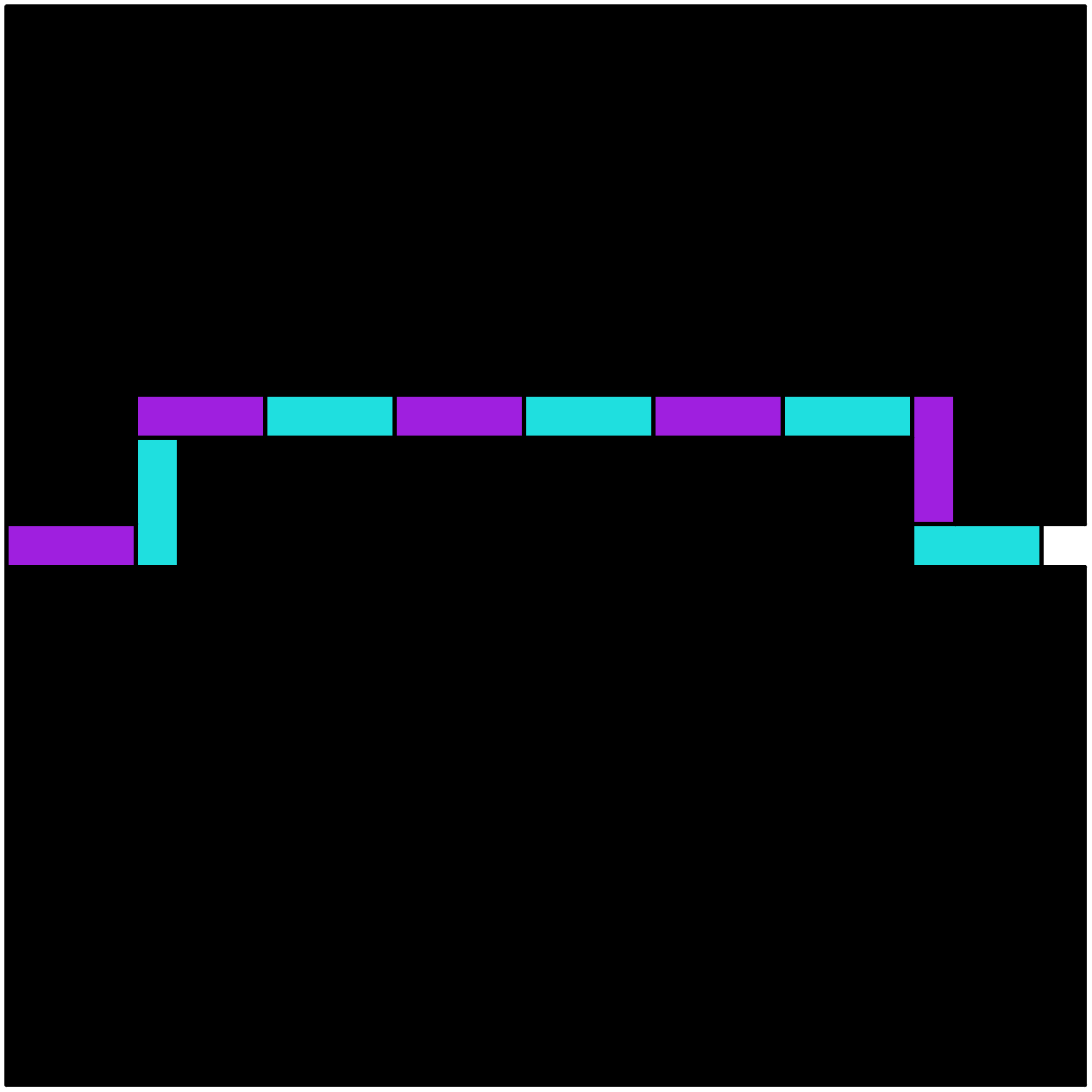}
\caption{Wire gadget for $\I$ tromino tiling used by \cite{tromino-tiling}.}
\label{fig:tromino-I-wire}
\end{figure}

\begin{figure}
\centering
\includegraphics[width=0.3\textwidth]{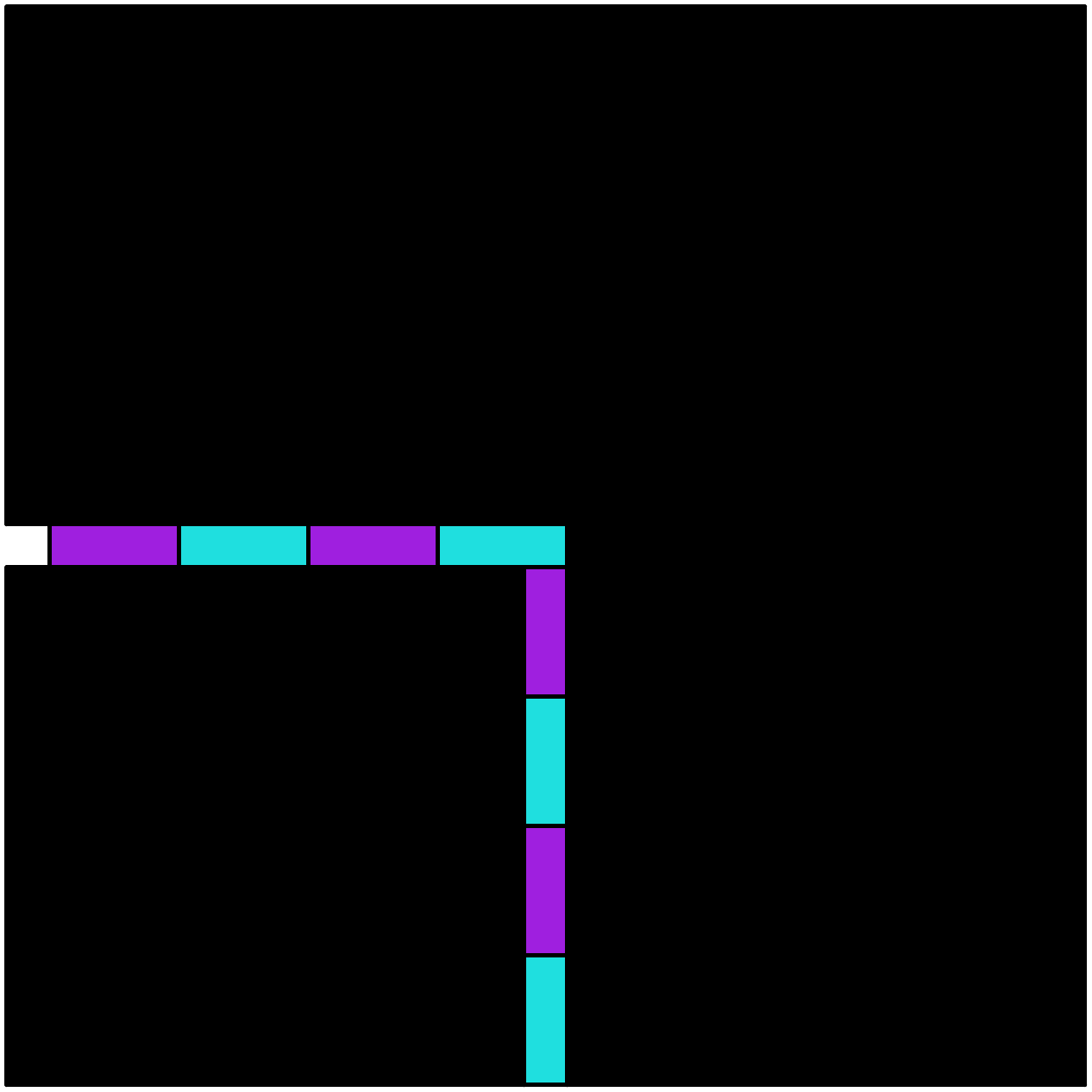}
\hspace{1cm}
\includegraphics[width=0.3\textwidth]{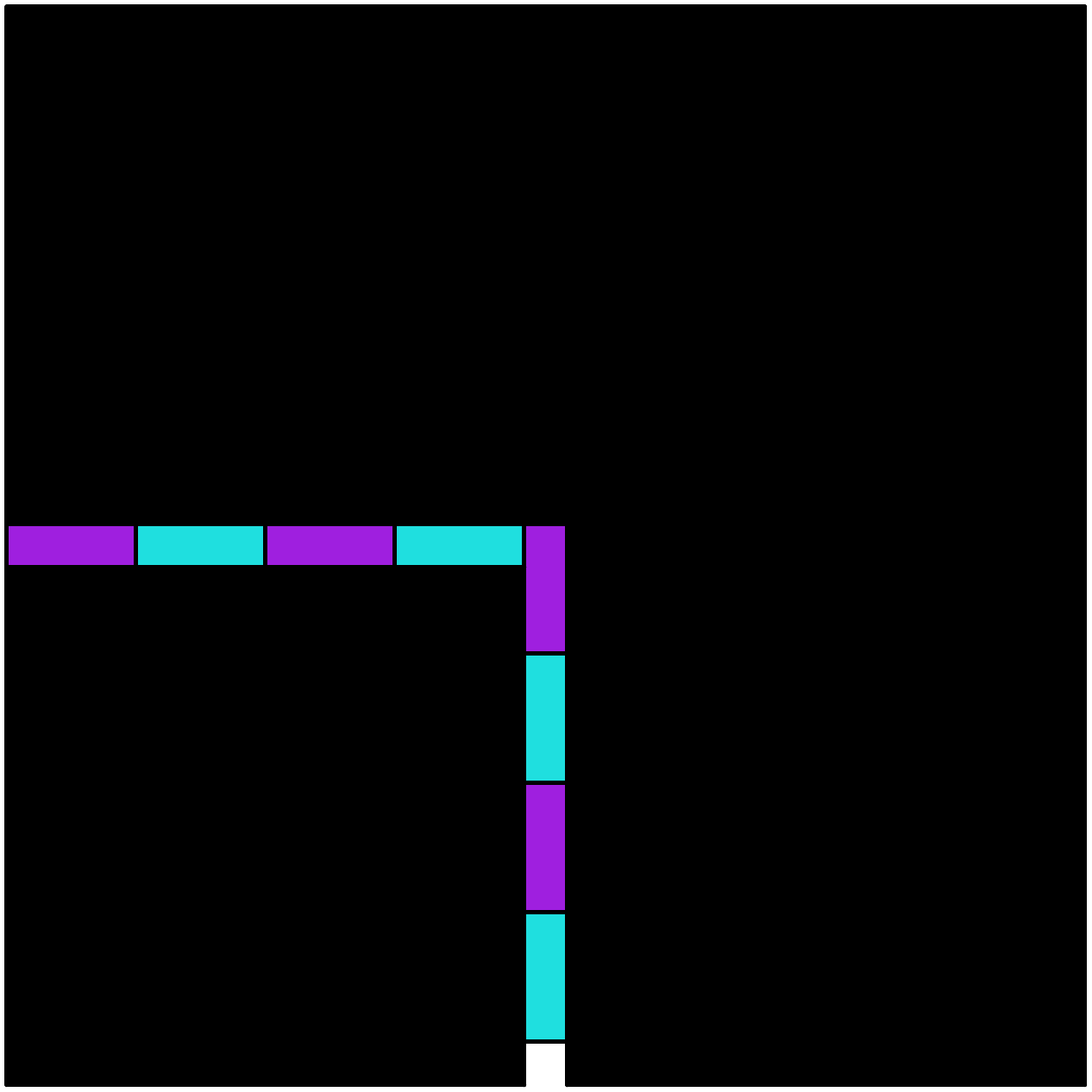}
\caption{Turn gadget for $\I$ tromino tiling used by \cite{tromino-tiling}.}
\label{fig:tromino-I-turn}
\end{figure}

\begin{figure}
\centering
\includegraphics[width=0.45\textwidth]{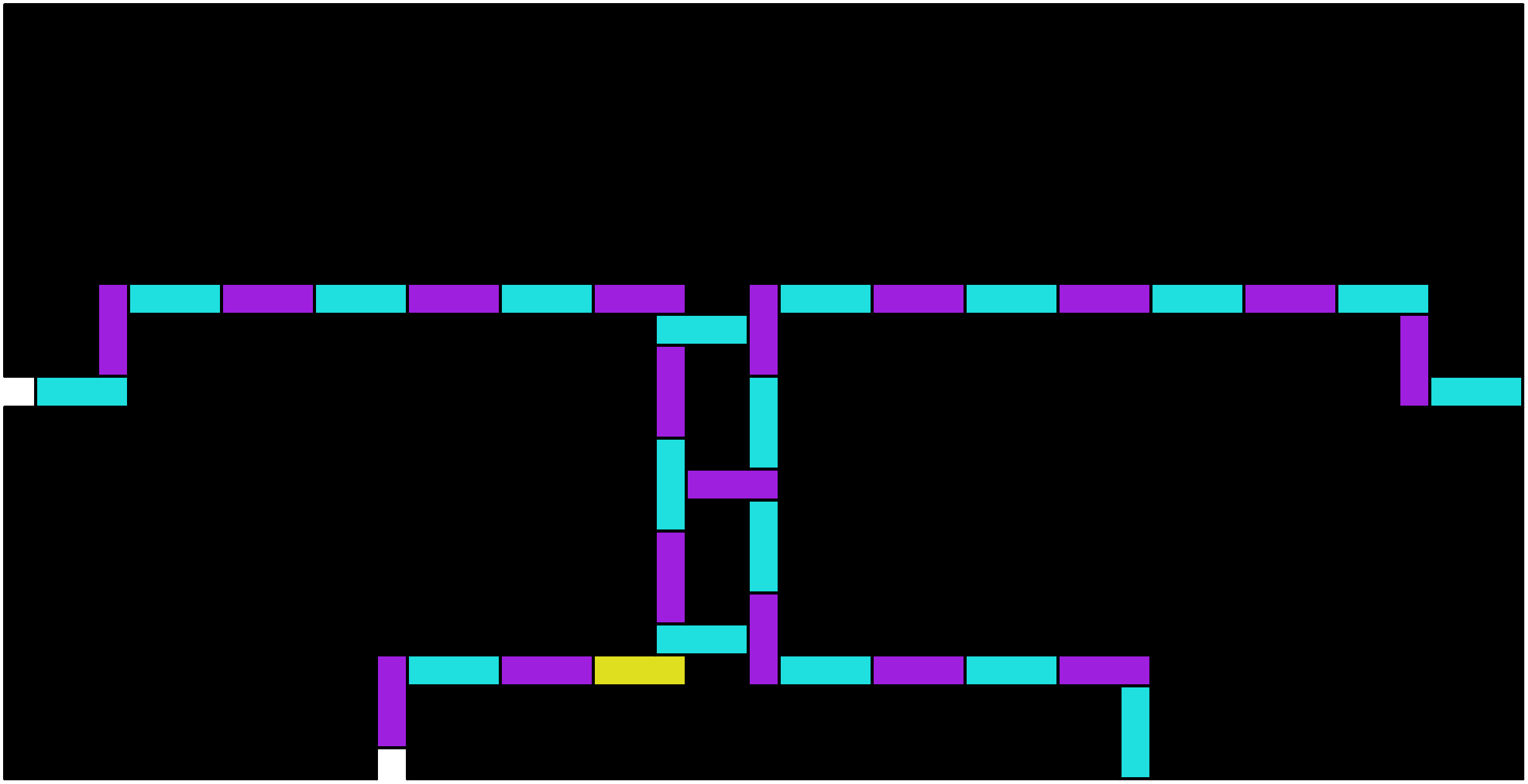}
\includegraphics[width=0.45\textwidth]{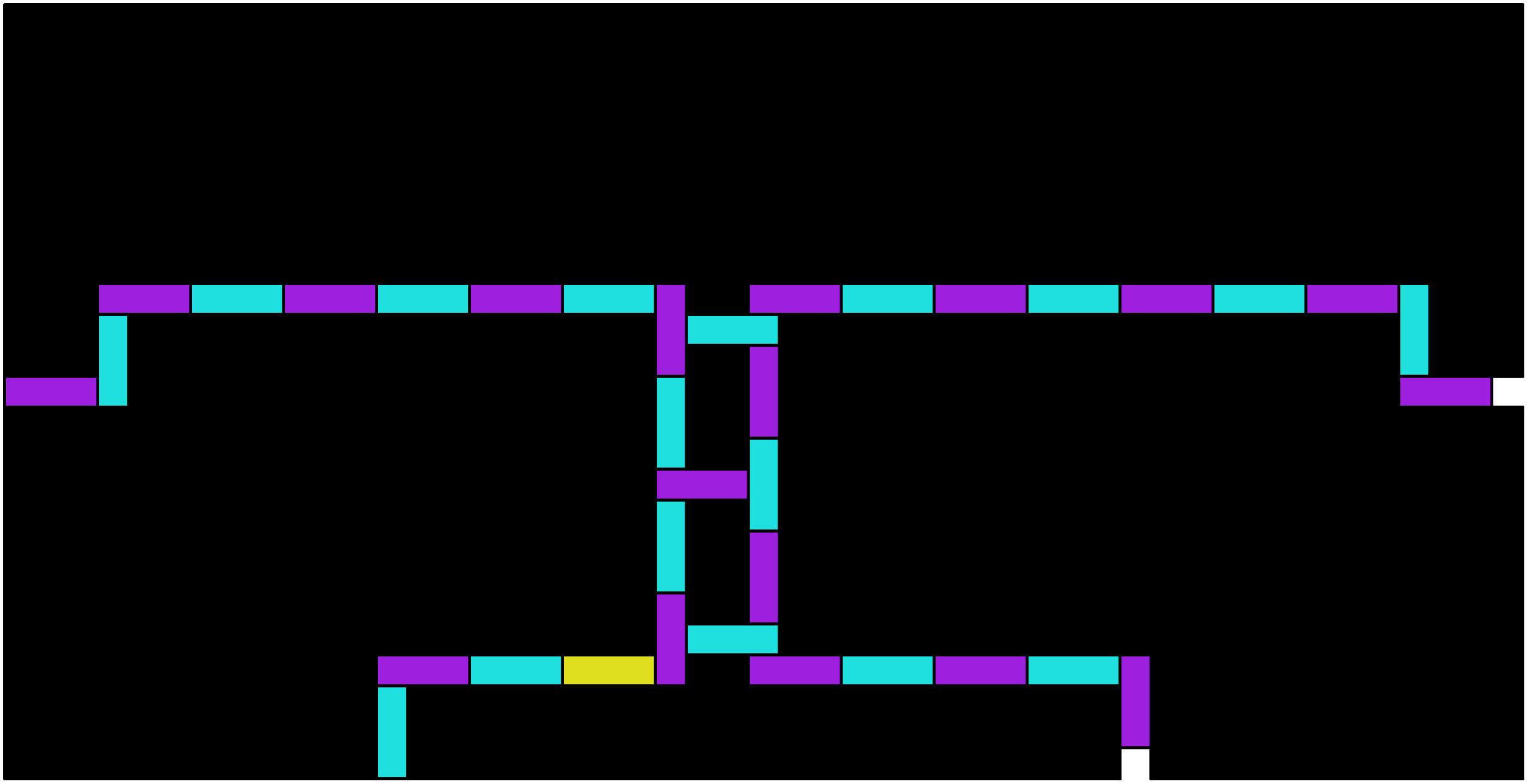}
\caption{Synchronizer gadget for $\I$ tromino tiling used by \cite{tromino-tiling}.}
\label{fig:tromino-I-synch}
\end{figure}

\begin{figure}
\centering
\includegraphics[width=0.3\textwidth]{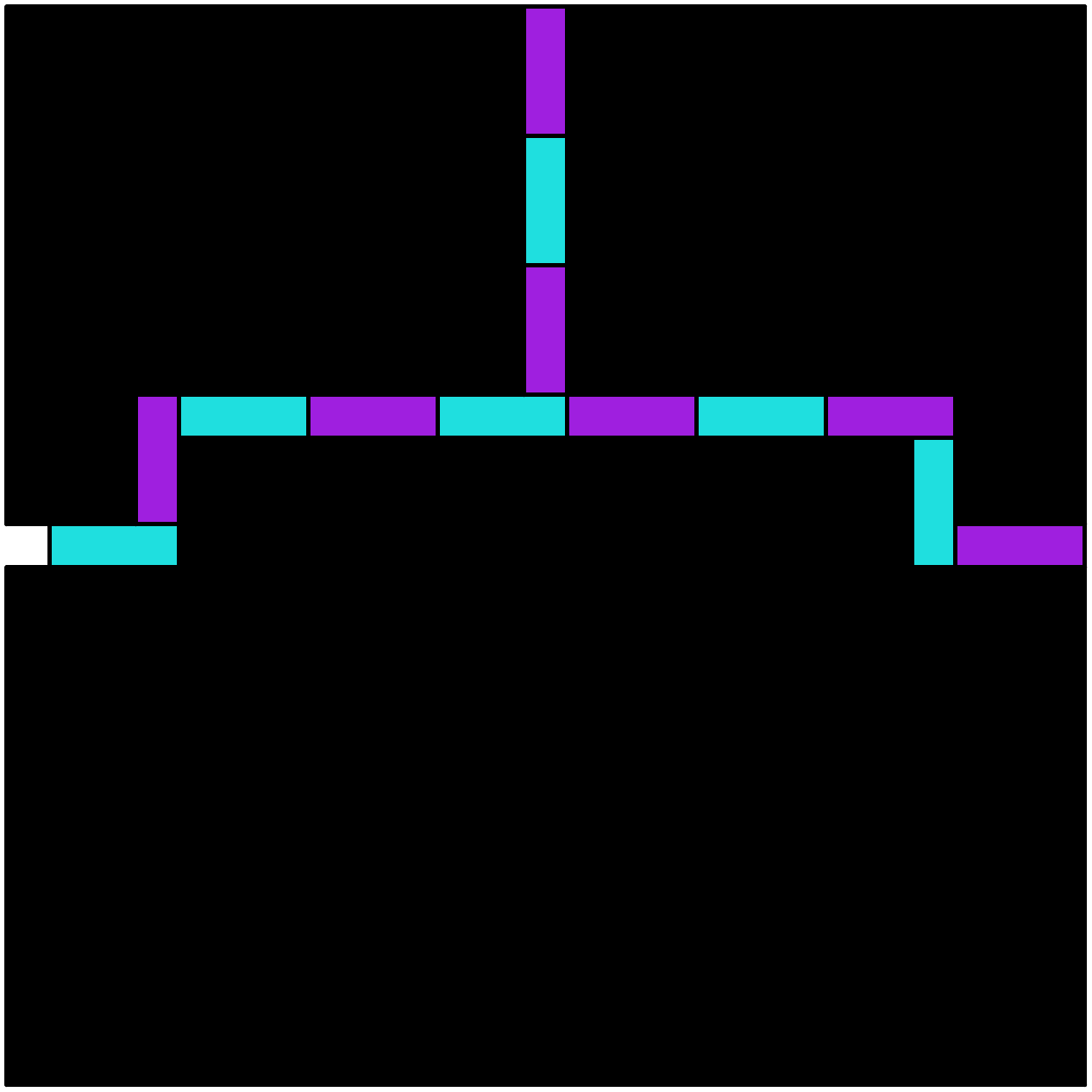}
\includegraphics[width=0.3\textwidth]{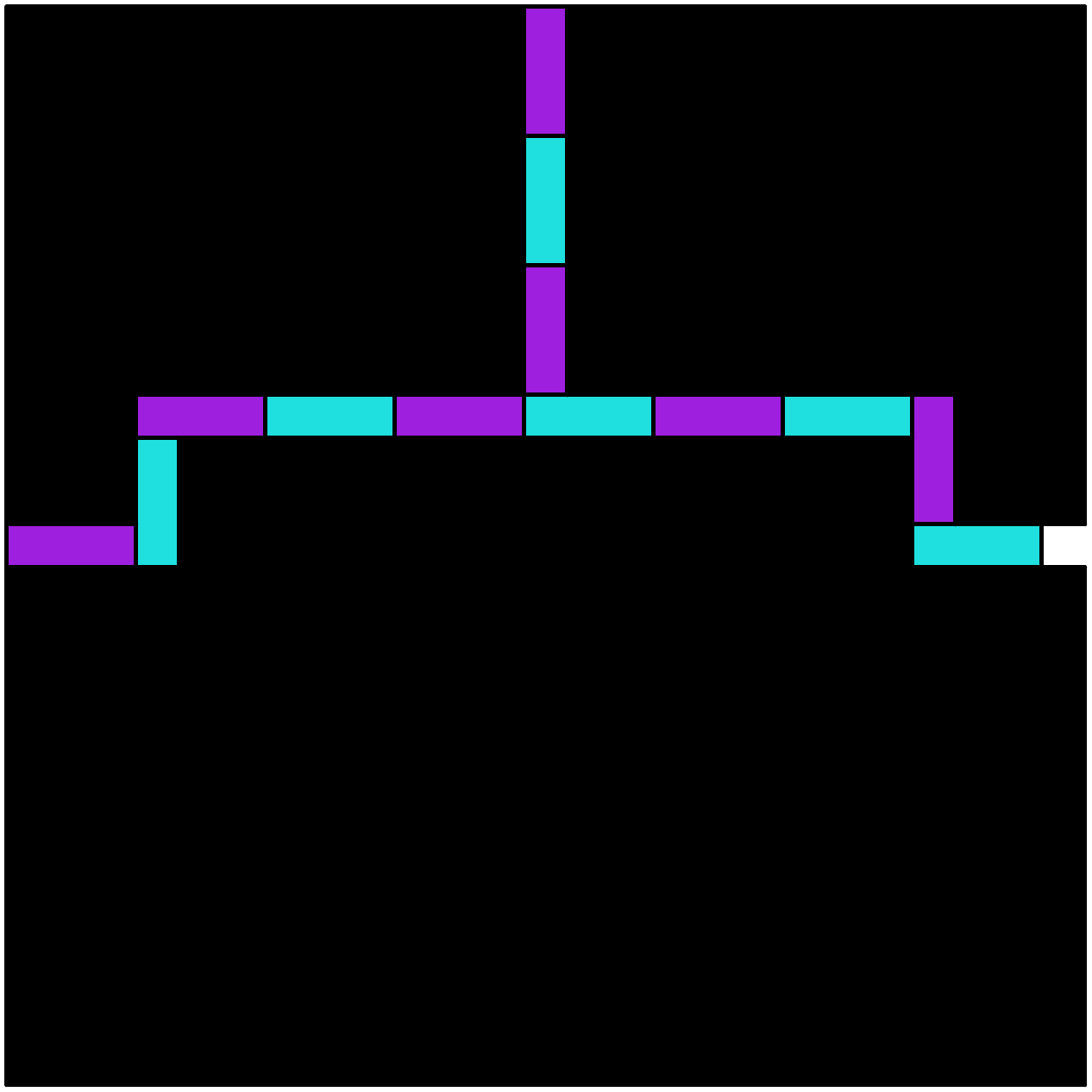}
\includegraphics[width=0.3\textwidth]{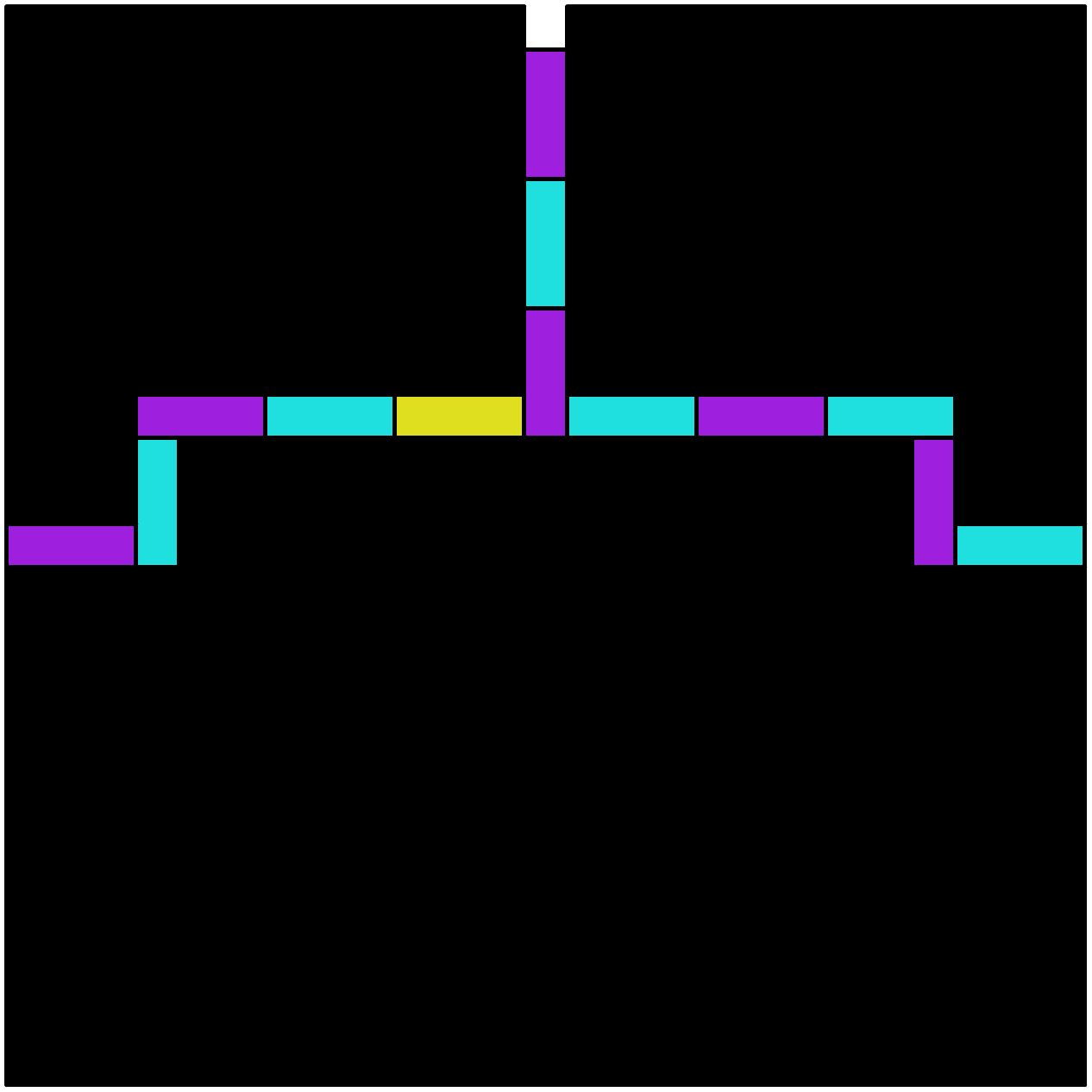}
\caption{1-in-3 (clause) gadget for $\I$ tromino tiling used by \cite{tromino-tiling}.}
\label{fig:tromino-I-1in3}
\end{figure}

\begin{figure}
\centering
\includegraphics[width=0.3\textwidth]{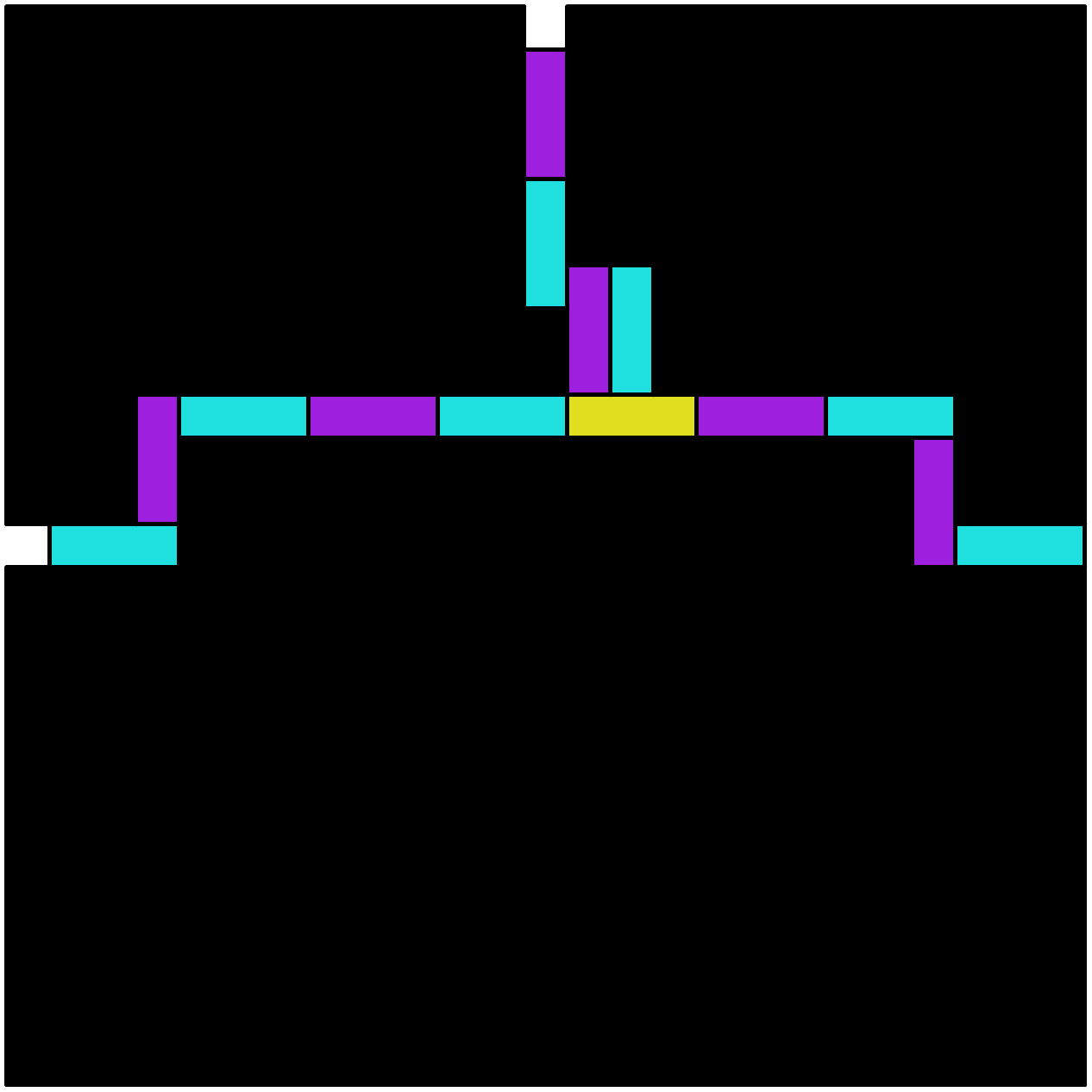}
\includegraphics[width=0.3\textwidth]{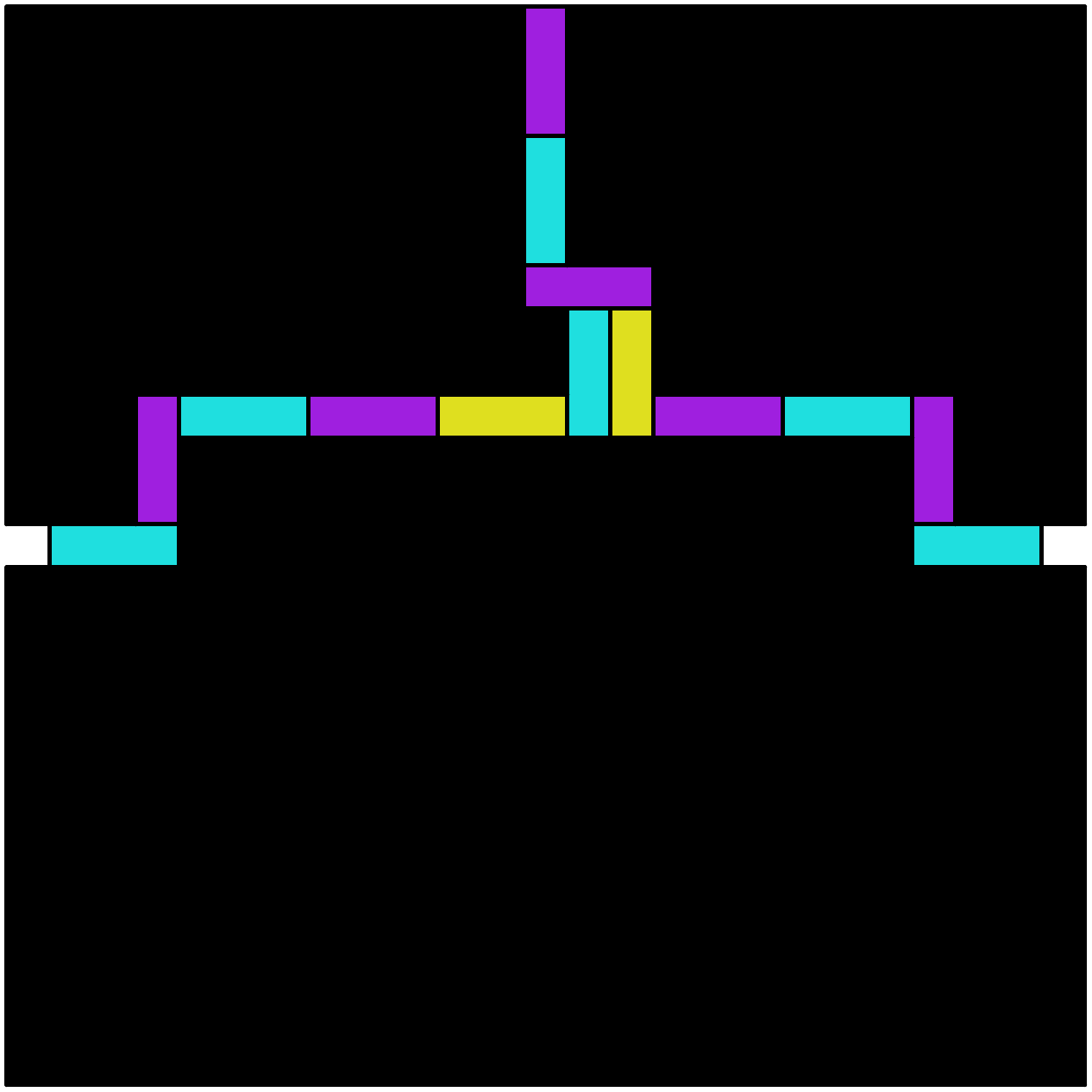}
\includegraphics[width=0.3\textwidth]{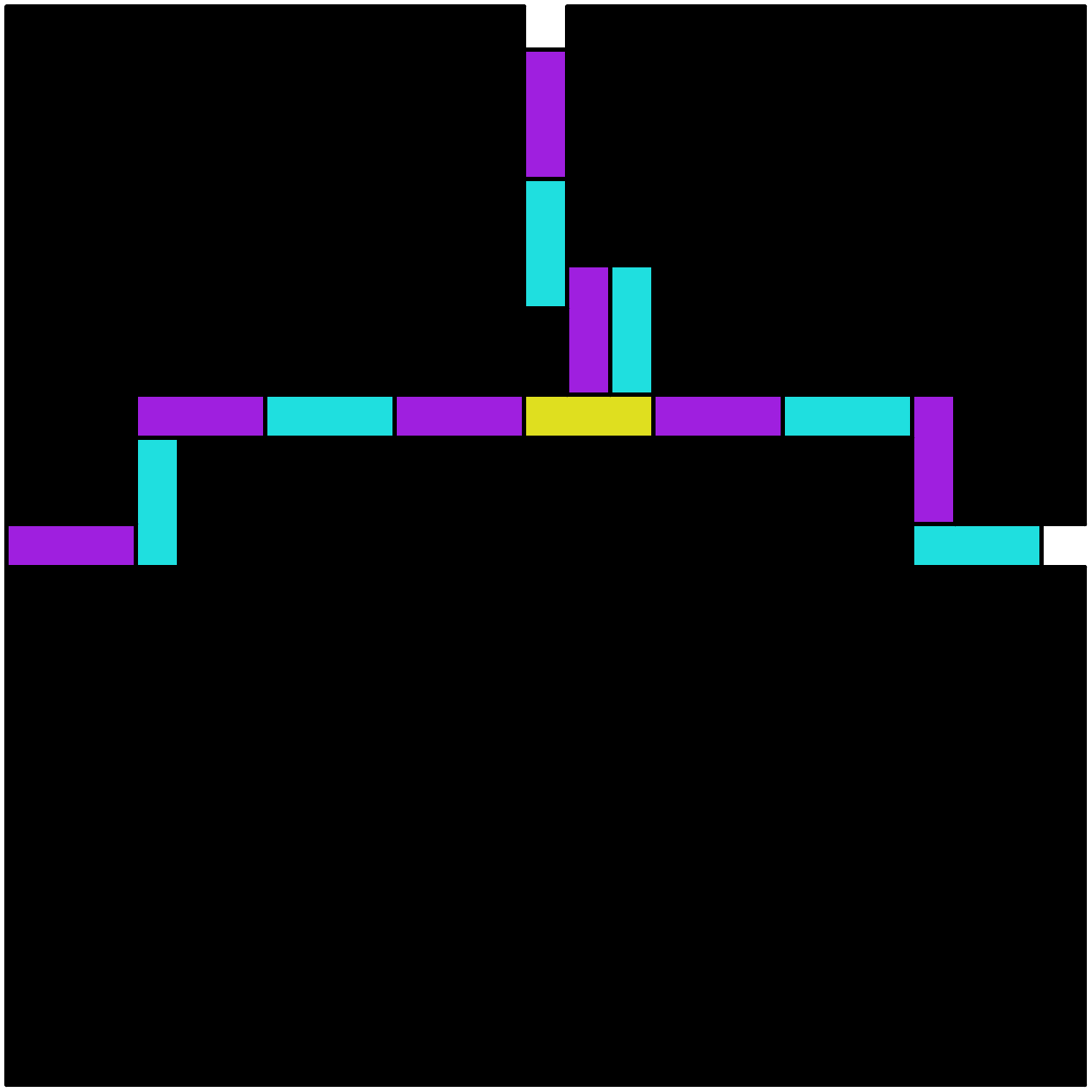}
\caption{2-in-3 (negated clause) gadget for $\I$ tromino tiling used by \cite{tromino-tiling}.}
\label{fig:tromino-I-2in3}
\end{figure}

\begin{figure}
\centering
\includegraphics[width=0.5\textwidth]{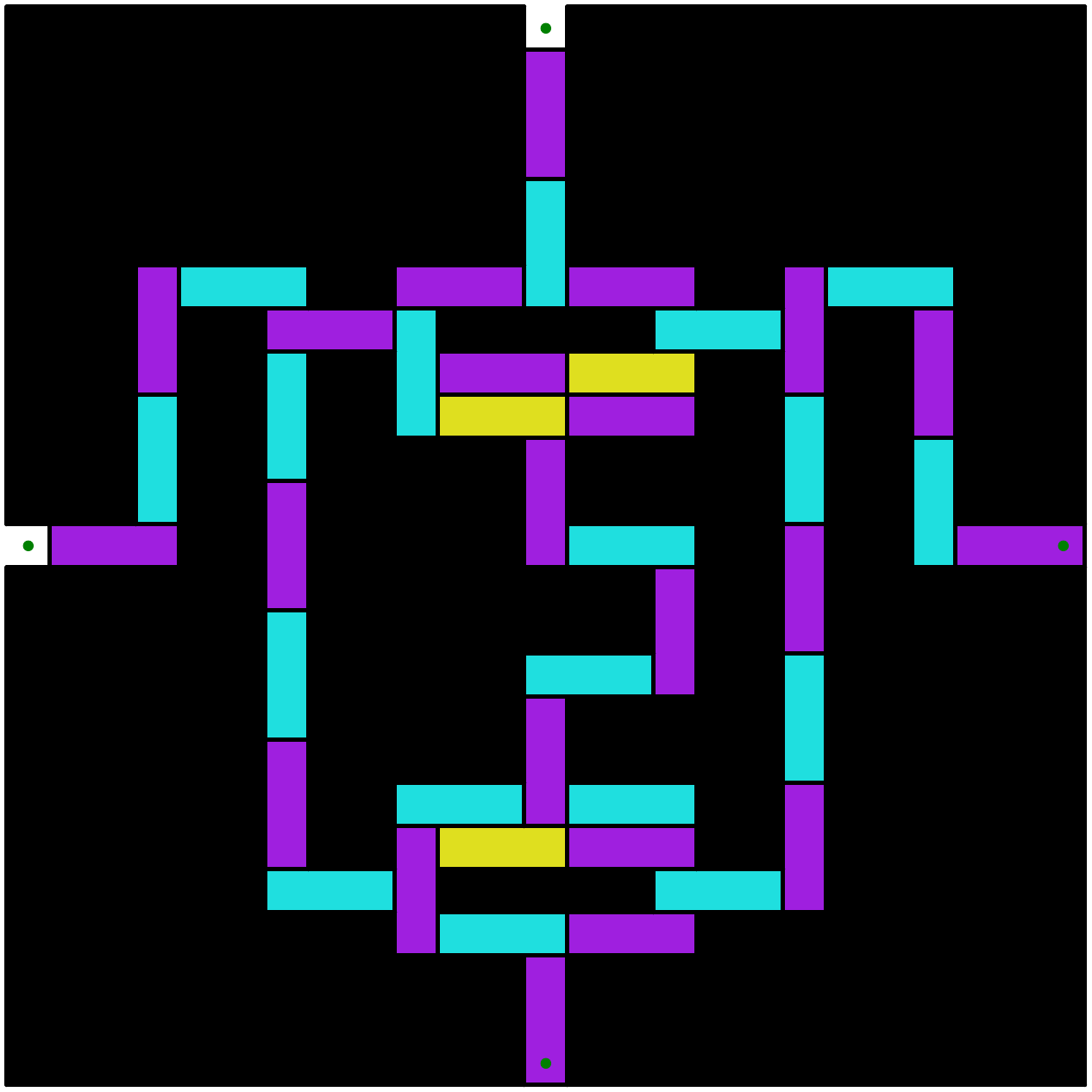}
\caption{One covering of the crossover gadget for $\I$ tromino tiling used by \cite{tromino-tiling}. This very nontrivial gadget is actually not necessary for establishing NP-hardness by Theorem~\ref{thm:i-in-j-and-dup}.}
\label{fig:tromino-I-crossover}
\end{figure}

\subsection{Tetromino Tiling}
\label{sec:tetromino-tiling}

We now apply our framework to establish several new results about the hardness of tiling with tetrominoes. We settle the complexity of tiling for all single tetromino sets, with or without reflection, and the set of all tetrominoes. Our results are summarized in Table~\ref{tab:tetrominos}.

First, we discuss the polynomial-time cases:

\begin{theorem}\label{thm:O-tetromino}
    Tiling with $\OO$ tetrominoes is in P.
\end{theorem}
\begin{proof}
    Loop through the cells of the grid in row-major order. At every uncovered cell, there is at most one way to cover it with an $\OO$ tetromino. Greedily place an $\OO$ tetromino there and continue.
\end{proof}

\begin{theorem}\label{thm:S-tetromino}
    Tiling with $\SS$ tetrominoes is in P. Symmetrically, tiling with $\ZZ$ tetrominoes is also in P.
\end{theorem}
\begin{proof}
\begin{figure}[h!]
\centering
    \includegraphics[width=0.2\textwidth]{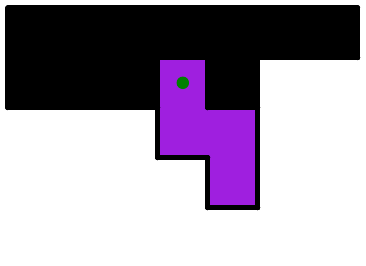}
    \hspace{1cm}
    \includegraphics[width=0.2\textwidth]{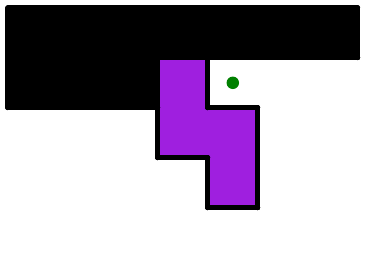}
    \hspace{1cm}
    \includegraphics[width=0.2\textwidth]{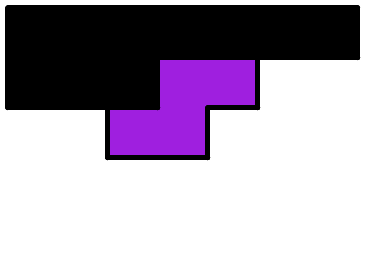}
    \caption{For $\SS$ tetromino tiling, there is always at most one choice that does not immediately cause a contradiction.}
    \label{fig:S-tetromino-easy}
\end{figure}
    Loop through the cells of the grid in row-major order. Consider an uncovered cell $c$ such that all cells before it in row-major order are already covered, as in Figure~\ref{fig:S-tetromino-easy}. If the cell to the right of $c$ is covered or off the grid, the only way to cover $c$ is to place an $\SS$ vertically (Figure~\ref{fig:S-tetromino-easy}, left). If the cell to the right of $c$ is not covered, then if we place an $\SS$ vertically to cover $c$, then this cell would be impossible to cover (Figure~\ref{fig:S-tetromino-easy}, middle). Therefore, we must place an $\SS$ horizontally (Figure~\ref{fig:S-tetromino-easy}, right). There is always at most one way to cover $c$ that does not immediately cause a contradiction, so we can greedily place an $\SS$ tetromino there and continue.
\end{proof}

We now apply our framework to establish NP-hardness for the remaining tetrominoes.

\begin{theorem}\label{thm:T-tetromino}
    Tiling with $\TT$ tetrominoes (plus any subset of the remaining tetrominoes) is NP-hard.
\end{theorem}
\begin{proof}
    Reduce from Planar $\{\text{1-in-4},\text{3-in-4},\text{synchronizer}\}$-GO, which is NP-hard by Theorem~\ref{thm:i-in-j-and-sync}.

\begin{figure}
    \centering
    \includegraphics[width=0.4\textwidth]{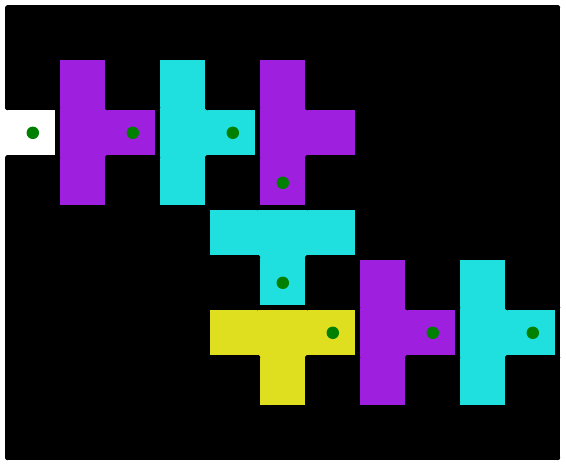}\hfil
    \includegraphics[width=0.4\textwidth]{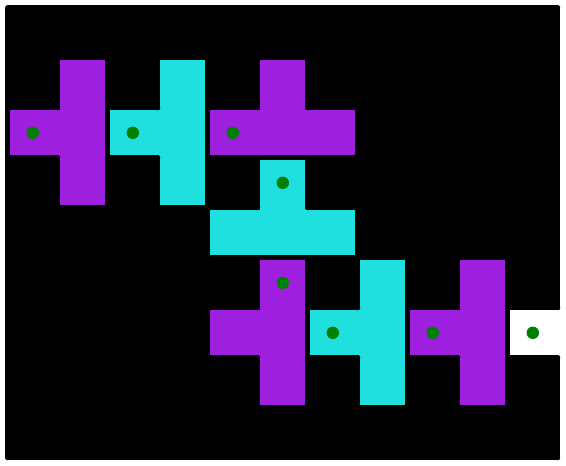}
    \caption{Wire gadget for $\TT$ tetrominoes.}
    \label{fig:tetromino-T-wire}
\end{figure}

    We represent edges using a wire as shown in Figure~\ref{fig:tetromino-T-wire}. Using turns as shown in the middle of Figure~\ref{fig:tetromino-T-wire}, the wire can be translated by any amount that preserves $(x+y) \bmod 2$. Unfortunately, because $\TT$ tetrominoes always cover 1 square of one color and 3 squares of the other color in a checkerboard coloring of the grid, it is not possible to build a wire that moves by an odd amount. Thus, we ensure that all gadget ports have the same color, so they can be connected by wires.

\begin{figure}
    \centering
    \includegraphics[width=0.23\textwidth]{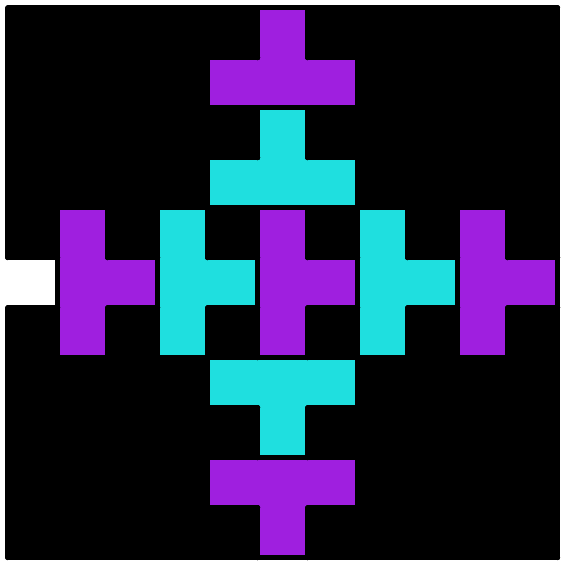}\hfill
    \includegraphics[width=0.23\textwidth]{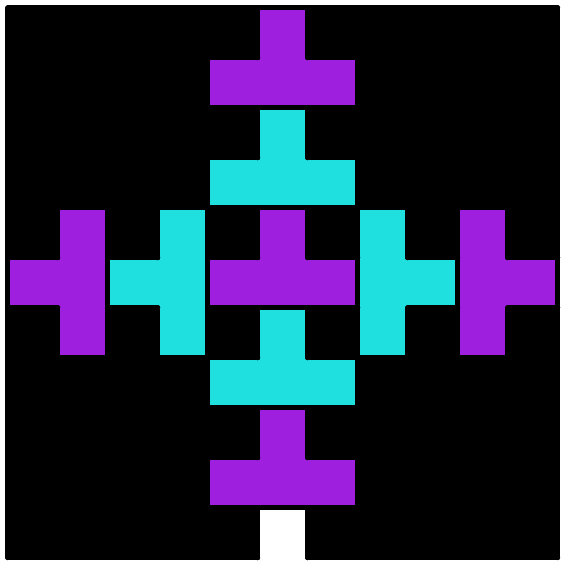}\hfill
    \includegraphics[width=0.23\textwidth]{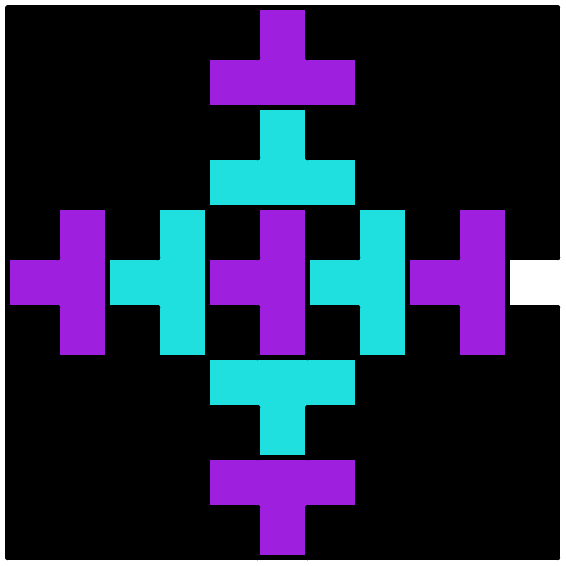}\hfill
    \includegraphics[width=0.23\textwidth]{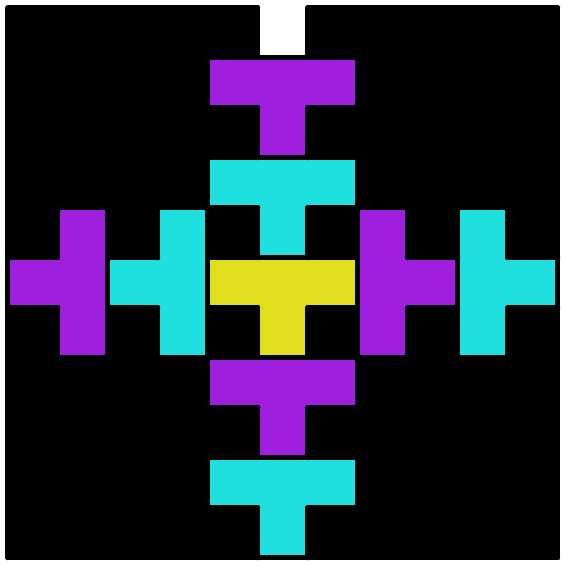}
    \caption{The four possible tilings of the 1-in-4 gadget for $\TT$ tetrominoes.}
    \label{fig:tetromino-T-1in4}
\end{figure}

\begin{figure}
    \centering
    \includegraphics[width=0.23\textwidth]{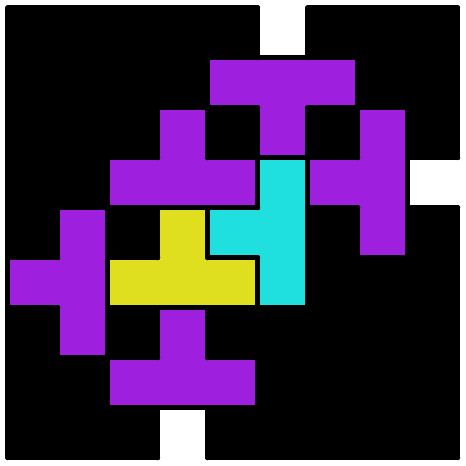}\hfill
    \includegraphics[width=0.23\textwidth]{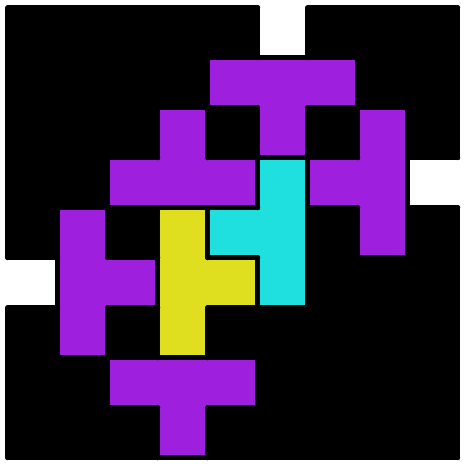}\hfill
    \includegraphics[width=0.23\textwidth]{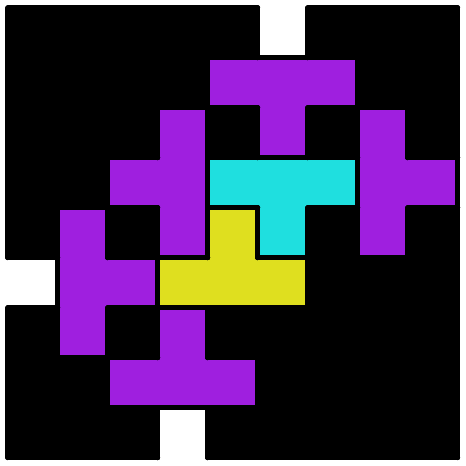}\hfill
    \includegraphics[width=0.23\textwidth]{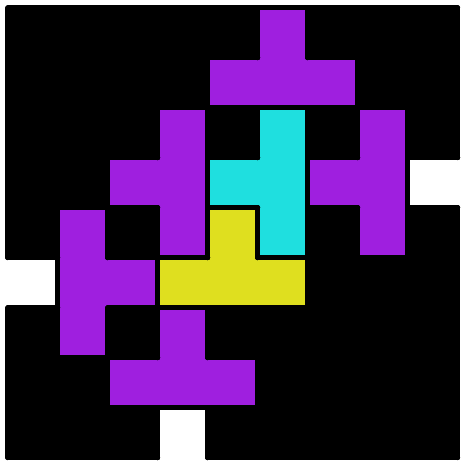}
    \caption{The four possible tilings of the 3-in-4 gadget for $\TT$ tetrominoes.}
    \label{fig:tetromino-T-3in4}
\end{figure}

    The 1-in-4 and 3-in-4 gadgets are shown in Figure~\ref{fig:tetromino-T-1in4} and Figure~\ref{fig:tetromino-T-3in4} respectively. Note that all gadget ports $(x,y)$ have the same $(x+y) \bmod 2$. Other partial covers are not possible because the number of covered cells must be a multiple of $4$.

\begin{figure}
    \centering
    \includegraphics[width=0.4\textwidth]{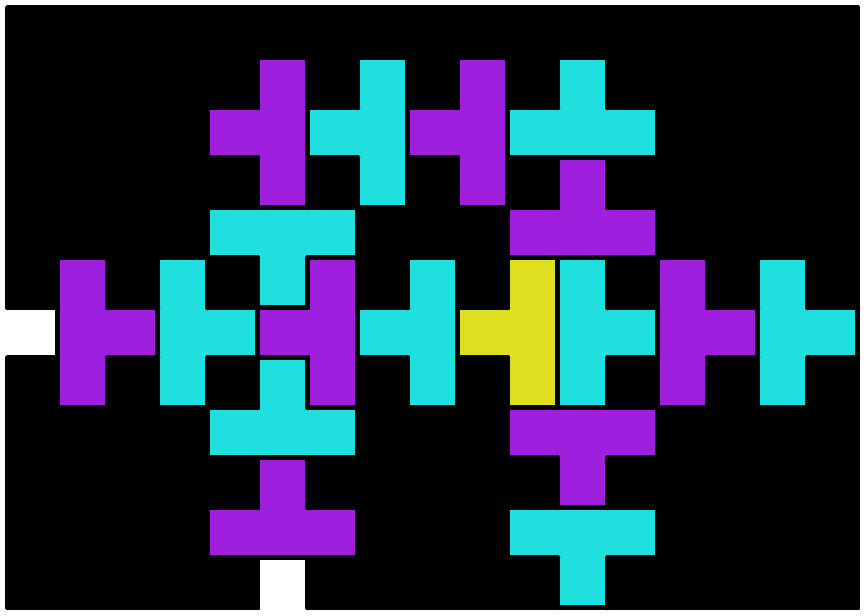}\hfil
    \includegraphics[width=0.4\textwidth]{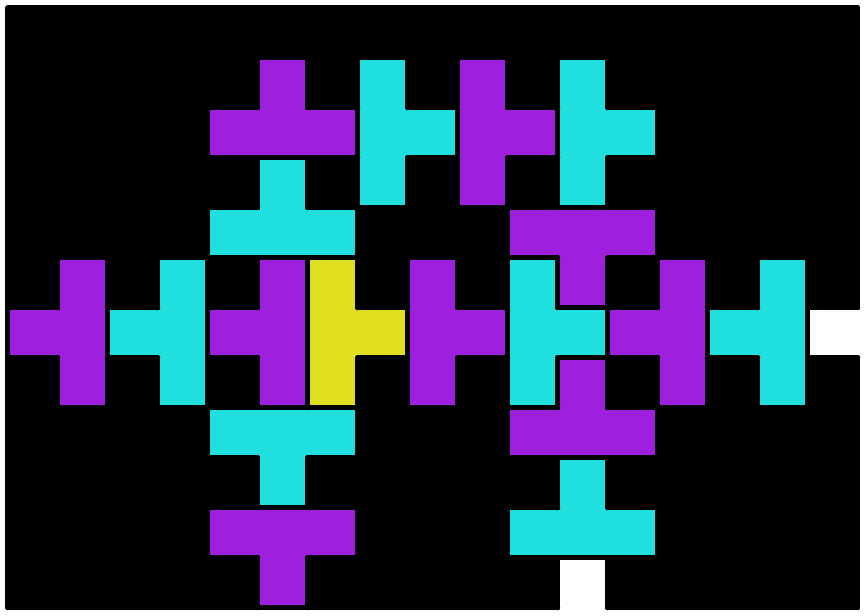}
    \caption{Synchronizer gadget for $\TT$ tetrominoes.}
    \label{fig:tetromino-T-synch}
\end{figure}

    The synchronizer is shown in Figure~\ref{fig:tetromino-T-synch}. It is built out of two ``4-equalizer'' gadgets connected through two wires. Constructing only a 4-equalizer does not work because one of the wires ends up at the wrong parity. Hence, we connect two of them to build a synchronizer where all gadget ports have the same parity.

    These gadgets suffice to prove that tiling with $\TT$ tetrominoes is NP-hard. Finally, observe that out of all seven tetrominoes, $\TT$ tetrominoes are the only tetrominoes that can be placed inside a wire without immediately creating isolated cells. Therefore, in any tiling, the wires must be tiled using only $\TT$ tetrominoes. Even if other tetrominoes are permitted, there are no undesirable partial coverings inside each gadget because the total number of covered cells is a multiple of $4$ (the synchronizer is built out of two 4-equalizers, so this reasoning also applies to it). Therefore, this proof works for any subset of tetrominoes containing $\TT$.
\end{proof}

\begin{theorem}\label{thm:L-tetromino}
    Tiling with $\LL$ tetrominoes is NP-hard. Symmetrically, tiling with $\JJ$ tetrominoes is also NP-hard. Tiling with both $\LL$ tetrominoes and $\JJ$ tetrominoes (i.e., $\LL$ tetrominoes with reflection) is also NP-hard.
\end{theorem}
\begin{proof}
\begin{figure}
    \centering
    \includegraphics[width=0.263\textwidth]{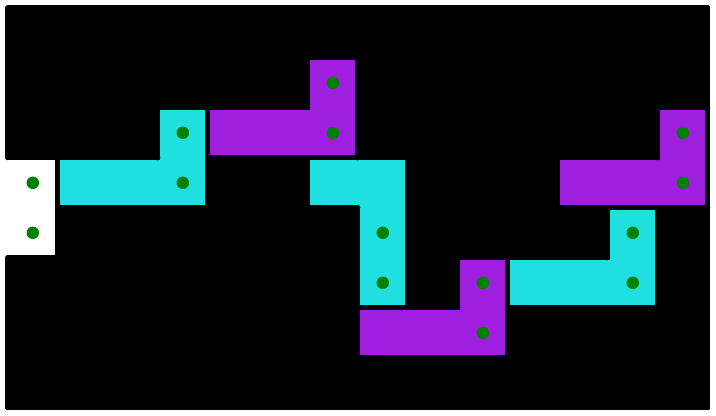}
    \includegraphics[width=0.263\textwidth]{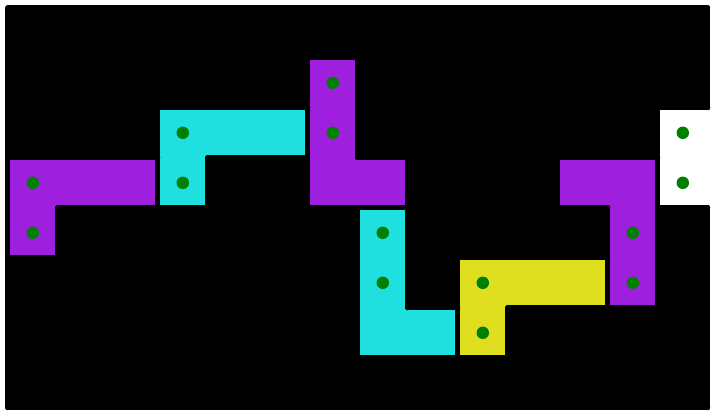}
    \hspace{0.3cm}
    \includegraphics[width=0.21\textwidth]{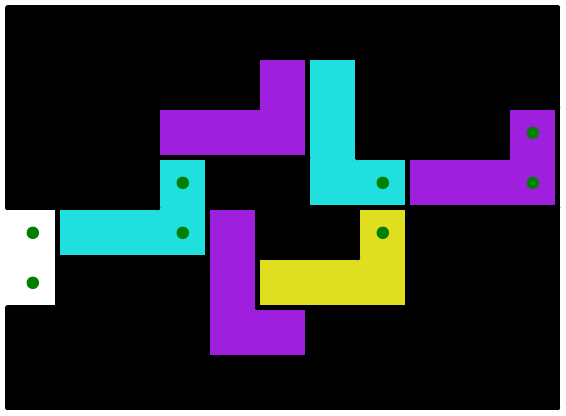}
    \includegraphics[width=0.21\textwidth]{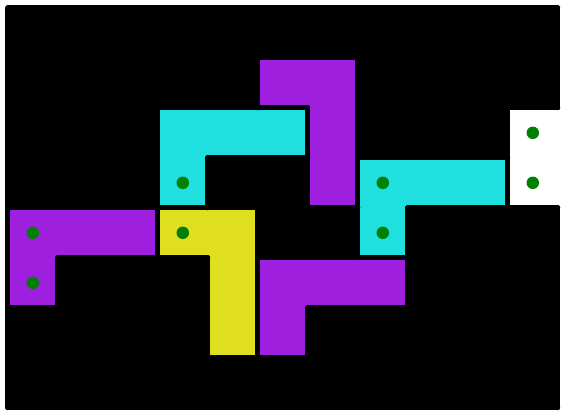}
    \caption{Left: Wire/turn gadget for $\LL$ tetrominoes. Right: Shift gadget for $\LL$ tetrominoes.}
    \label{fig:tetromino-L-wire}
\end{figure}
    Reduce from Planar $\{\text{3-in-4},\text{4-equalizer}\}$-GO, which is NP-hard by Theorem~\ref{thm:i-in-j-and-dup}.

    We represent edges using the wire and shift gadgets shown in Figure~\ref{fig:tetromino-L-wire}. Unlike previous proofs, the orientation of an edge is represented using a $2 \times 1$ rectangle of cells. The wire gadget can translate the signal by any amount in the square lattice generated by $(-2,1)$ and $(1,2)$ (using $(\text{row}, \text{column})$ coordinates), which is unfortunately not all possible translations. To correct this, the shift gadget allows shifting by $(0,4)$, so now any translation is obtainable.

\begin{figure}
    \centering
    \includegraphics[width=0.3\textwidth]{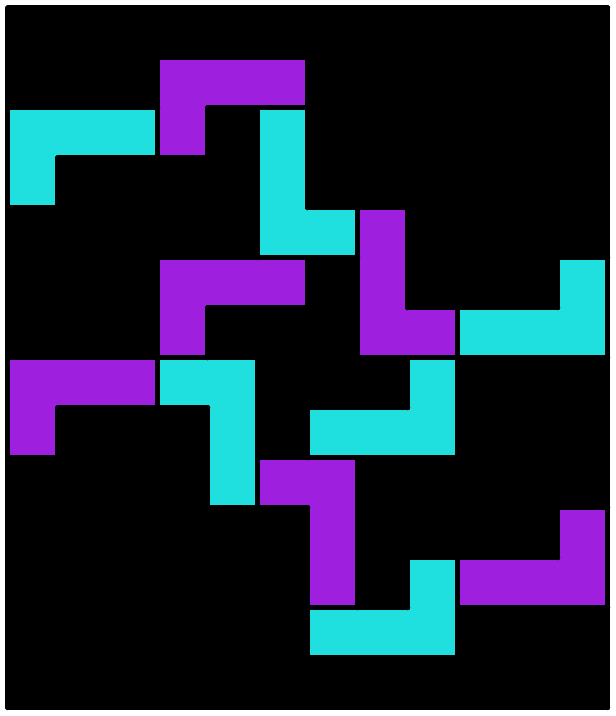}
    \hfil
    \includegraphics[width=0.3\textwidth]{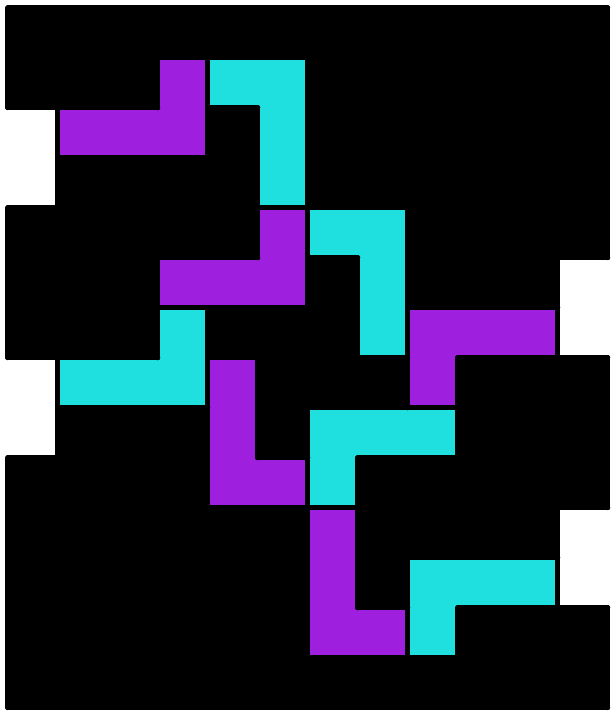}
    \caption{The two possible tilings of the 4-equalizer gadget for $\LL$ tetrominoes.}
    \label{fig:tetromino-L-0or4}
\end{figure}

\begin{figure}
    \centering
    \includegraphics[width=0.23\textwidth]{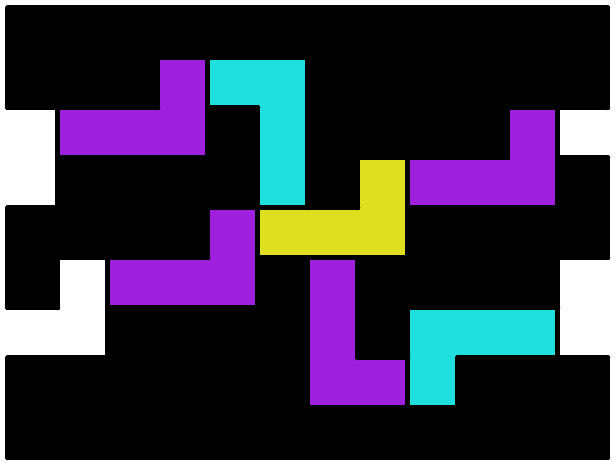}\hfill
    \includegraphics[width=0.23\textwidth]{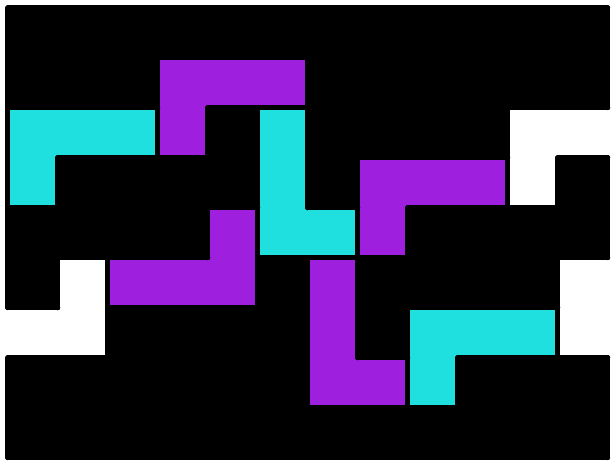}\hfill
    \includegraphics[width=0.23\textwidth]{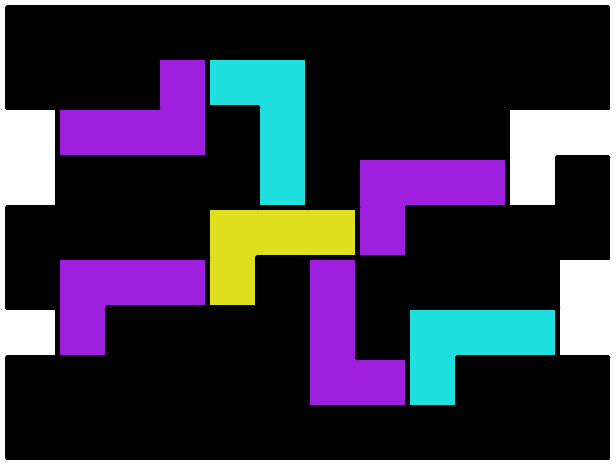}\hfill
    \includegraphics[width=0.23\textwidth]{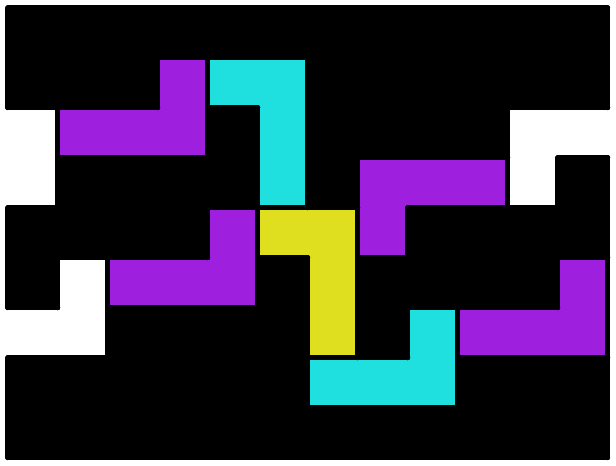}
    \caption{The four possible tilings of the 3-in-4 gadget for $\LL$ tetrominoes.}
    \label{fig:tetromino-L-3in4}
\end{figure}

    The 4-equalizer gadget is shown in Figure~\ref{fig:tetromino-L-0or4}. The two ways of covering the inner loop decide whether the neighboring edges all point inward or all outward.

    The 3-in-4 gadget is shown in Figure~\ref{fig:tetromino-L-3in4}. There are four ways to cover the $1 \times 2$ rectangle at the center of the gadget, each determining which neighboring edge points outward.

    These gadgets suffice to prove that tiling with $\LL$ tetrominoes is NP-hard. The proof still works if we additionally allow $\JJ$ tetrominoes (i.e., allow reflection in the tiles): Observe that no $\JJ$ tetrominoes can fit into any of the wire gadgets without creating isolated cells. Therefore, the wire gadgets can only be tiled using $\LL$ tetrominoes, leaving the $2 \times 1$ rectangular gadget ports in all 4-equalizers and 3-in-4 gadgets uncovered or fully covered. It can be checked that even if $\JJ$ tetrominoes are permitted, the only subsets of gadget ports that can be covered inside a gadget are the desired arrangements of Figure~\ref{fig:tetromino-L-0or4} and Figure~\ref{fig:tetromino-L-3in4}.
\end{proof}

Although tiling with only $\SS$ tetrominoes or only $\ZZ$ tetrominoes is easy, allowing both types of tiles (i.e., allowing reflections) makes the problem NP-hard.
\begin{theorem}\label{thm:SZ-tetromino}
    Tiling with both $\SS$ tetrominoes and $\ZZ$ tetrominoes is NP-hard.
\end{theorem}
\begin{proof}
    Reduce from Planar $\{\text{1-in-4},\text{3-in-4},\text{synchronizer}\}$-GO, which is NP-hard by Theorem~\ref{thm:i-in-j-and-sync}.

\begin{figure}
    \centering
    \includegraphics[width=0.3\textwidth]{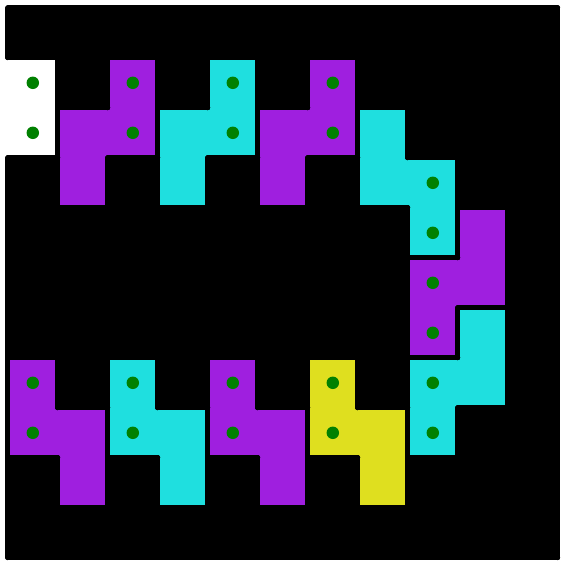}
    \hfil
    \includegraphics[width=0.3\textwidth]{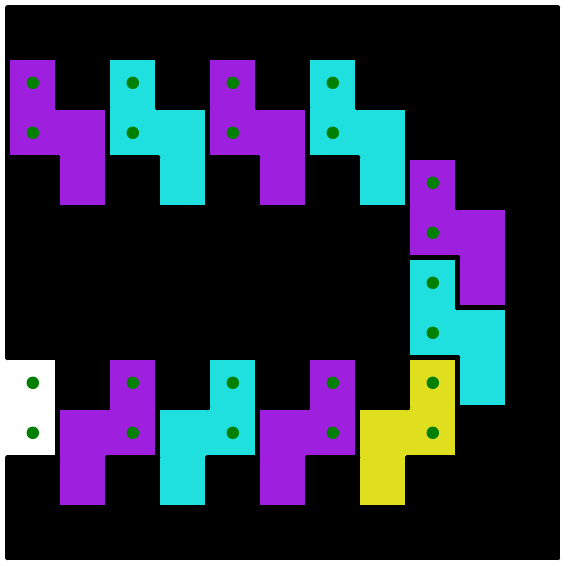}
    \caption{Wire gadget for $\SS$ tetrominoes with reflection.}
    \label{fig:tetromino-SZ-wire}
\end{figure}

    We represent edges using a wire as shown in Figure~\ref{fig:tetromino-SZ-wire}. The orientation of an edge is represented using a $2 \times 1$ rectangle of cells. By turning as shown in Figure~\ref{fig:tetromino-SZ-wire}, the wire can be moved vertically and horizontally by any multiple of $2$ units. We will ensure that all gadget ports have the same $(x \bmod 2, y \bmod 2)$.

\begin{figure}
    \centering
    \includegraphics[width=0.23\textwidth]{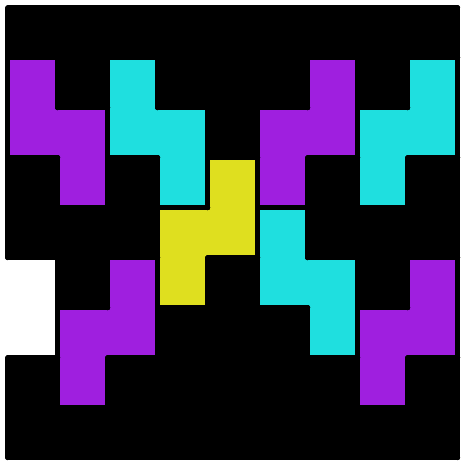}
    \includegraphics[width=0.23\textwidth]{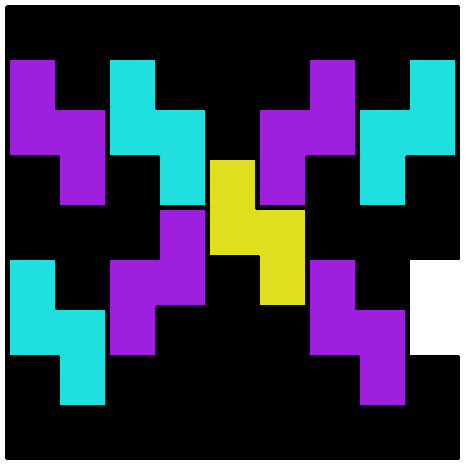}
    \includegraphics[width=0.23\textwidth]{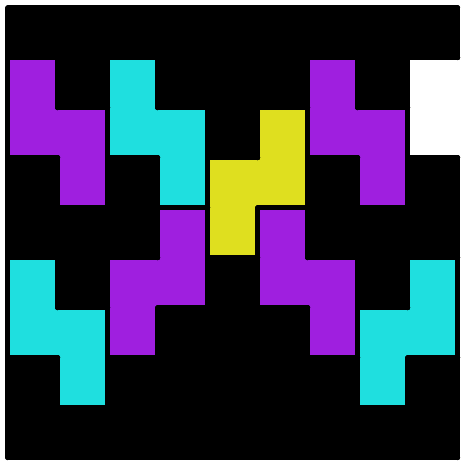}
    \includegraphics[width=0.23\textwidth]{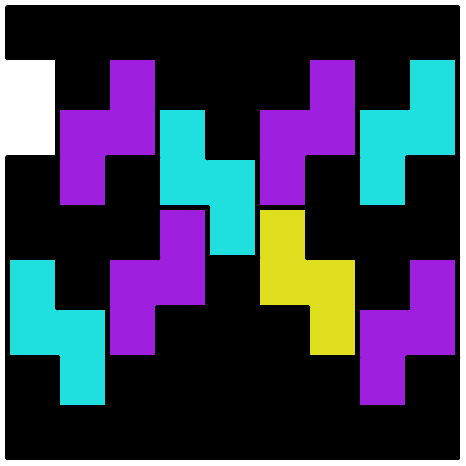}
    \caption{The four possible tilings of the 1-in-4 gadget for $\SS$ tetrominoes with reflection.}
    \label{fig:tetromino-SZ-1in4}
\end{figure}

    The 1-in-4 gadget is shown in Figure~\ref{fig:tetromino-SZ-1in4}. There are four ways to cover the $2 \times 1$ rectangle in the center of the gadget, corresponding to the four valid ways to orient the neighboring edges. All gadget ports have the same $(x \bmod 2,y \bmod 2)$, so they can be connected using the wire gadgets.

\begin{figure}
    \centering
    \includegraphics[width=0.23\textwidth]{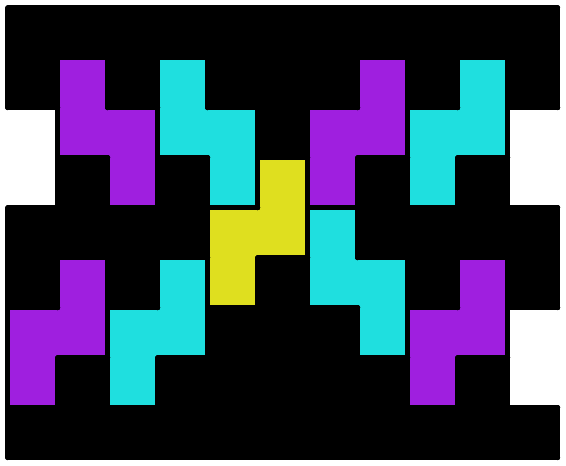}\hfill
    \includegraphics[width=0.23\textwidth]{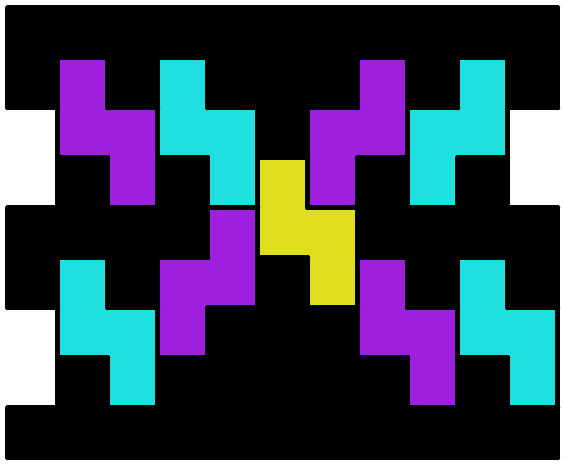}\hfill
    \includegraphics[width=0.23\textwidth]{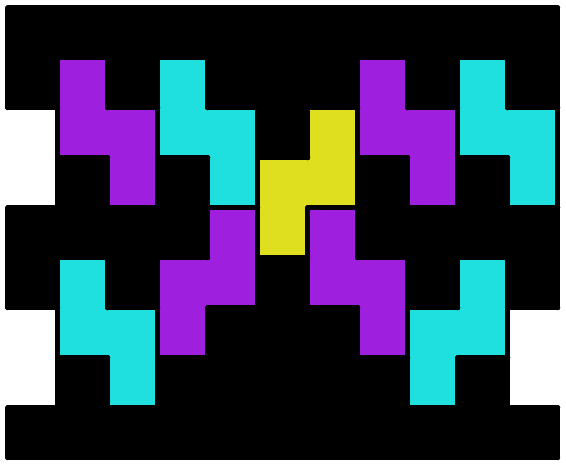}\hfill
    \includegraphics[width=0.23\textwidth]{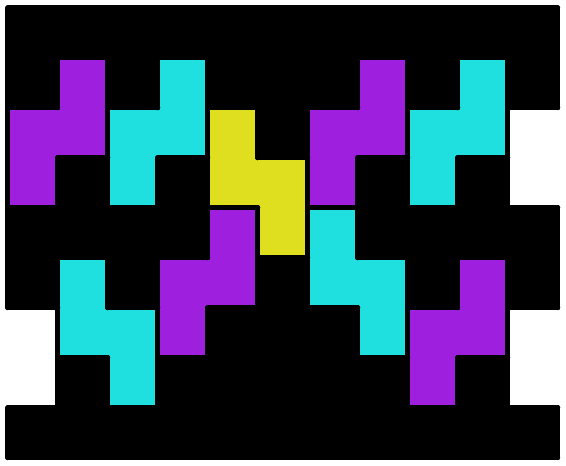}
    \caption{The four possible tilings of the 3-in-4 gadget for $\SS$ tetrominoes with reflection. This gadget is built in the same way as the 1-in-4 gadget.}
    \label{fig:tetromino-SZ-3in4}
\end{figure}

    The 3-in-4 gadget is built in the same way as the 1-in-4 gadget, except we extend the wires by ``half'' a tetromino, as shown in Figure~\ref{fig:tetromino-SZ-3in4}.

\begin{figure}[h]
    \centering
    \includegraphics[width=0.4\textwidth]{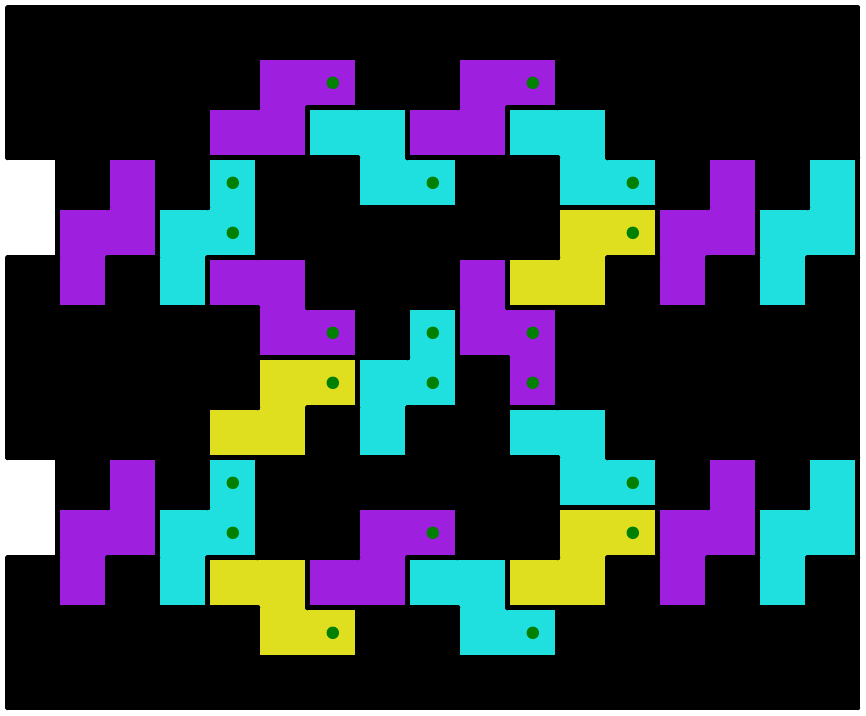}\hfil
    \includegraphics[width=0.4\textwidth]{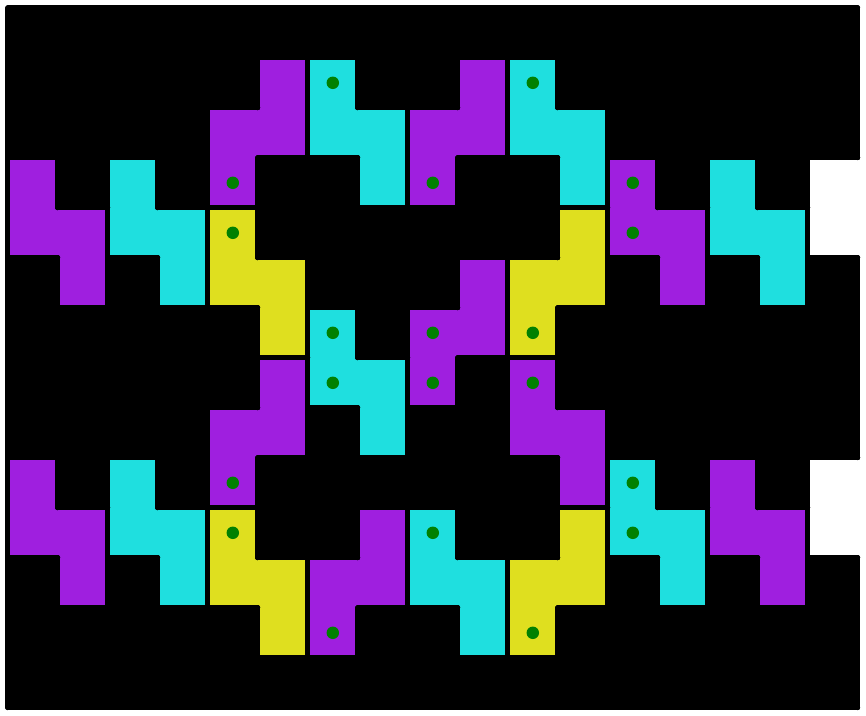}
    \caption{Synchronizer gadget for $\SS$ tetrominoes with reflection. The gadget works by splitting and recombining the wires. }
    \label{fig:tetromino-SZ-synch}
\end{figure}

    The synchronizer is shown in Figure~\ref{fig:tetromino-SZ-synch}. It essentially splits the 2-cell-wide wire signals into 1-cell-wide wires (connected at the green dots) and rejoins them. At the center of the gadget, two 1-cell-wide wires are briefly combined into a 2-cell-wide wire and then split apart again. This synchronizes these wires to point in the same direction. As a result, the only partial coverings of this gadget are the two shown in Figure~\ref{fig:tetromino-SZ-synch}.
\end{proof}

\section*{Acknowledgments}

This paper was initiated during open problem solving in the MIT class on Algorithmic Lower Bounds: Fun with Hardness Proofs (6.5440) taught by Erik Demaine in Fall 2023. We thank the other participants of that class --- specifically
Evgeniya Artemova,
Krit Boonsiriseth,
Josh Brunner,
Lily Chung,
Della Hendrickson,
Hayashi Layers,
Jayson Lynch,
Richard Qi,
Mark Saengrungkongka,
Nathan Sheffield,
Andy Tockman,
Anthony Wang,
William Wang, and
Alek Westover
--- for helpful discussions and providing an inspiring atmosphere.
The original concept is based on early discussions in 2012 with
Martin Demaine, Takahasi Horiyama, and Ryuhei Uehara,
who we thank for helping invent this problem domain.

\bibliography{bibliography}
\bibliographystyle{alpha}

\end{document}